%% file: journal_paper.tex
\documentclass[preprint, 3p]{elsarticle}

\PassOptionsToPackage{hyphens}{url}\usepackage{hyperref}
\usepackage{ecrc}
\usepackage{fixltx2e}
\usepackage{amsmath}
\usepackage{amssymb}
\usepackage{amsthm}
\usepackage{graphicx}
\usepackage{wrapfig}
\usepackage{subfigure}
\usepackage{url}
\usepackage[algoruled,vlined,boxed,linesnumbered]{algorithm2e}
\SetKwRepeat{Do}{do}{while}
\usepackage{mathtools}
\usepackage{verbatim}
\usepackage{multirow, makecell}
\usepackage{paralist}
\usepackage{listings}
\usepackage{mathptmx}
\usepackage{times}
\usepackage{tikz}
\usepackage{dirtree}
\usepackage{tabularx}
\usepackage{algpseudocode}
\usepackage{tikz}
\usepackage{blkarray, bigstrut}
\usepackage[framemethod=default]{mdframed}
\usepackage{showexpl}
\usepackage{todonotes}
\usepackage{enumitem}
\usepackage[mathscr]{eucal}
\usetikzlibrary{automata,arrows}
\usetikzlibrary{datavisualization}
\usetikzlibrary{datavisualization.formats.functions}
\usetikzlibrary{positioning, calc, backgrounds, arrows}
\usetikzlibrary{shapes.geometric}
\usetikzlibrary{decorations.pathreplacing}
\usetikzlibrary{patterns}

\volume{00}

\firstpage{1}

\journalname{Information Systems}

\runauth{D. Rei{\ss}ner et al.}

\jid{JSS}

\jnltitlelogo{Information Systems}

\usepackage{amssymb}

\usepackage[figuresright]{rotating}

\newtheorem{definition}{Definition}[section]
\newtheorem{lemma}{Lemma}[section]
\newtheorem{theorem}{Theorem}[section]

\newcommand{\nodes}{\mathit{N}}
\newcommand{\arcs}{\mathrm{A}}
\newcommand{\labels}{\Sigma}

\newcommand{\fsm}{\mathit{FSM}}

\newcommand{\outT}[1]{#1{\blacktriangleright}}
\newcommand{\inT}[1]{{\blacktriangleright}#1}

\newcommand{\logL}{\mathit{L}}

\newcommand{\dafsa}{\mathit{DAFSA}}
\newcommand{\initState}{s}
\newcommand{\finalStates}{\mathit{R}}

\newcommand{\tieconcat}{\oplus}
\newcommand{\dotieconcat}[2]{\text{\raisebox{.8ex}{$\smallfrown$}}}

\newcommand{\prefixF}[1]{\mathit{pref}(#1)}
\newcommand{\suffix}{\mathit{suff}}
\newcommand{\suffixF}[1]{\mathit{suff}(#1)}

\newcommand{\comPrefix}{\mathit{Pref}}
\newcommand{\comSuffix}{\mathit{Suff}}

\newcommand{\lnet}{\mathit{PN}}
\newcommand{\places}{\mathit{P}}
\newcommand{\netTransitions}{\mathit{T}}
\newcommand{\netArcs}{\mathit{F}}
\newcommand{\netLabel}{\lambda}

\newcommand{\reachGraph}{\mathit{RG}}
\newcommand{\inTr}[1]{\bullet{#1}}
\newcommand{\outTr}[1]{#1\bullet}
\newcommand{\inP}[1]{\bullet#1}
\newcommand{\outP}[1]{#1\bullet}
\newcommand{\initMarking}{\mathit{m_0}}
\newcommand{\finalMarkings}{\mathit{m_f}}
\newcommand{\markings}{\mathit{M}}

\newcommand{\optional}[1]{}

\newcommand{\head}[1]{\mathit{head}~#1}
\newcommand{\tail}[1]{\mathit{tail}~#1}
\newcommand{\open}{\sigma}
\newcommand{\closed}{\Omega}
\newcommand{\Closed}{\Theta}

\newcommand{\out}{\alpha}
\newcommand{\f}{\mathit{f}}
\newcommand{\n}{\mathit{m}}

\newcommand{\replaceTau}{\mathit{replaceTau}}

\newcommand{\remnodes}{\Xi}
\newcommand{\INV}{\Psi}
\newcommand{\inv}{\mathit{a}}

\newcommand{\prefMem}{\mathit{PrefixTable}}
\newcommand{\suffMem}{\mathit{SuffixTable}}

\newcommand{\psp}{\mathit{PSP}}
\newcommand{\configs}{\mathit{C}}
\newcommand{\conf}{\mathit{g}}
\newcommand{\op}{\mathit{op}}

\newcommand{\match}{\mathit{match}}
\newcommand{\lhide}{\mathit{lhide}}
\newcommand{\rhide}{\mathit{rhide}}

\newcommand{\sync}{\beta}
\newcommand{\syncs}{\mathit{B}}
\newcommand{\curCost}{\mathit{g}}
\newcommand{\futCost}{\mathit{h}}
\newcommand{\costF}{\rho}
\newcommand{\futL}{\mathit{F_{Log}}}
\newcommand{\futM}{\mathit{F_{Model}}}
\newcommand{\node}{\mathit{n}}
\newcommand{\mSet}{\mathit{MultiSet}}
\newcommand{\optimal}{\mathit{Opt}}
\newcommand{\potN}{\mathit{N_{new}}}

\newcommand{\maxcost}{\mathit{\rho_{max}}}
\newcommand{\actN}{\mathit{n_{act}}}
\newcommand{\nextN}{\mathit{n_{next}}}

\graphicspath{{images/}}

\newcommand{\scomp}{\mathscr{C}}
\newcommand{\trPr}[1]{\mathit{\upharpoonright}_{#1}}

\newcommand{\test}{\mathit{rhideMerged}}

\begin{document}

\begin{frontmatter}




\title{Scalable Alignment of Process Models and Event Logs:\\ An Approach Based on Automata and S-Components}


\author[1]{Daniel Rei{\ss}ner}
\ead{dreissner@student.unimelb.edu.au}
\author[1]{Abel Armas-Cervantes}
\ead{abel.armas@unimelb.edu.au}
\author[1]{Raffaele Conforti}
\ead{raffaele.conforti@unimelb.edu.au}
\author[2]{Marlon Dumas}
\ead{marlon.dumas@ut.ee}
\author[3]{Dirk Fahland}
\ead{d.fahland@tue.nl}
\author[1]{Marcello La Rosa}
\ead{marcello.larosa@unimelb.edu.au}

\address[1]{University of Melbourne, Australia}
\address[2]{University of Tartu, Estonia}
\address[3]{Eindhoven University of Technology, Netherlands}

\begin{abstract}
Given a model of the expected behavior of a business process and given an event log recording its observed behavior, the problem of business process conformance checking is that of identifying and describing the differences between the process model and the event log.
A desirable feature of a conformance checking technique is that it should identify a minimal yet complete set of differences.
Existing conformance checking techniques that fulfill this property exhibit limited scalability when confronted to large and complex process models and event logs.
One reason for this limitation is that existing techniques compare each execution trace in the log against the process model separately, without reusing computations made for one trace when processing subsequent traces. Yet, the execution traces of a business process typically share common fragments (e.g. prefixes and suffixes).
A second reason is that these techniques do not integrate mechanisms to tackle the combinatorial state explosion inherent to process models with high levels of concurrency.
This paper presents two techniques that address these sources of inefficiency.
The first technique starts by transforming the process model and the event log into two automata. These automata are then compared based on a synchronized product, which is computed using an A* heuristic with an admissible heuristic function, thus guaranteeing that the resulting synchronized product captures all differences and is minimal in size. The synchronized product is then used to extract optimal (minimal-length) alignments between each trace of the log and the closest corresponding trace of the model. By representing the event log as a single automaton, this technique allows computations for shared prefixes and suffixes to be made only once.
The second technique decomposes the process model into a set of automata, known as S-components, such that the product of these automata is equal to the automaton of the whole process model. A product automaton is computed for each S-component separately. The resulting product automata are then recomposed into a single product automaton capturing all the differences between the process model and the event log, but without minimality guarantees. 
An empirical evaluation using 40 real-life event logs shows that, used in tandem, the proposed techniques outperform state-of-the-art baselines in terms of execution times in a vast majority of cases, with improvements ranging from several-fold to one order of magnitude. Moreover, the decomposition-based technique leads to optimal trace alignments for the vast majority of datasets and close to optimal alignments for the remaining ones.
\end{abstract}


\begin{keyword}
Process Mining \sep Conformance Checking \sep Automata \sep Petri nets \sep S-components

\end{keyword}

\end{frontmatter}


\input{intro}
\input{background_new}
\input{approach_new}
\input{Extension_eval}
\input{conclusion}





\bibliographystyle{elsarticle-num}
\bibliography{lit}
\appendix
\input{Appendix.tex}







\end{document}

%% file: intro.tex
\sloppy
\section{Introduction}\label{sec:introduction}

Modern information systems maintain detailed business process execution trails. For example, an enterprise resource planning system keeps records of key events related to a company's order-to-cash process, such as the receipt and confirmation of purchase orders, the delivery of products, and the creation and payment of invoices.
Such records can be grouped into an \emph{event log} consisting of sequences of events (called \emph{traces}), each consisting of all event records pertaining to one case of a process.

Process mining techniques~\cite{ProcessMiningBook} allow us to exploit such event logs in order to gain insights into the performance and conformance of business processes. One widely used family of process mining techniques is conformance checking~\cite{ConformanceCheckingBook}.
A conformance checking technique takes as input a process model capturing the expected behaviour of a business process, and an event log capturing its observed behaviour. The goal of conformance checking is to identify and describe the differences between the process model and the event log. A common approach to achieve this is by computing \emph{alignments} between traces in the event log and traces that may be generated by the process model. In this context, a trace alignment is a data structure that describes the differences between a trace of the log and a possible trace of the model. These differences are captured as a sequence of moves, including synchronous moves (moving forward both in the trace of the log and in the trace of the model) and asynchronous moves (moving forward either only in the trace of the log or only in the trace of the model). A desirable feature of a conformance checking technique is that it should identify a minimal (yet complete) set of behavioural differences. In the context of trace alignments, this means that the computed alignments should have a minimal length (or more generally minimal-cost).


Existing techniques that fulfill these properties~\cite{AdriansyahDA11,BVD-Alignment}
exhibit scalability limitations when confronted to large and complex event logs. For example, in a collection of 40 real-life event logs presented later in this paper, the execution times of these techniques are over 10 seconds in about a quarter of cases and over 5 seconds in about half of cases, which hampers the use of these techniques in interactive settings as well as in use cases where it is necessary to apply conformance checking repeatedly, for example in the context of automated process discovery~\cite{AugustoCDRMMMS19}, where several candidate models need to be compared by computing their conformance with respect to a given log.

The scalability limitations of existing conformance checking techniques stem, at least in part, from two sources of inefficiency:
\begin{enumerate}
\item These techniques compute an alignment for each trace of the log separately. They do not reuse partial computations made for one trace when processing other traces. Yet, traces in an event log typically share common prefixes and suffixes.
We hypothesise that computing alignments for common prefixes or suffixes only once may improve performance in this context.
\item They do not integrate mechanisms to tackle the combinatorial state explosion inherent to process models with high levels of concurrency. Note that the number of possible interleavings of the parallel activities increases rapidly, for example four tasks in parallel can be executed in 24 different ways while eight tasks can already be executed in 40~320 different ways. This combinatorial explosion has an impact on the space of possible alignments between traces in an event log and the process model. 
\end{enumerate}

This paper presents two complementary techniques to address these issues.
The first technique starts by transforming the process model and the event log into two automata. Specifically, the process model is transformed into a minimal Deterministic Acyclic Finite State Automaton (DAFSA), while the process model is transformed into another automaton, namely its \emph{reachability graph}.
These automata are then compared using a synchronised product\footnote{A synchronised product of two automata is an automaton capturing the combined behaviour of the two automata executed in parallel, with synchronization occurring when transitions with the same symbol are taken in both automata (i.e. when both automata make the same move).}
 computed via an A* heuristic with an admissible heuristic function. The latter property guarantees that the resulting synchronised product captures all differences with a minimal number of non-synchronised transitions, which correspond to differences between the log and the model.
The synchronised product is then used to extract optimal (minimal-size) alignments between each trace of the log and the closest corresponding trace of the model.

To tackle the second of the above issues, the paper proposes a technique wherein the process model is first decomposed into a set of automata, known as S-components, such that the product of these automata is equal to the automaton of the whole process model. A synchronised product automaton is computed for each S-component separately. For example, given a model with four activities in a parallel block, this model is decomposed into four models -- each containing one of the four parallel tasks. These concurrency-free models are then handled separately, thus avoiding the computation of all possible interleavings and reducing the search space for computing a minimal synchronised product.
Once we have computed a product automaton for each S-component, these product automata are then recomposed into a single product automaton capturing all the differences between the process model and the event log. The article puts forward conditions under which this recomposition leads to a correct output. The article shows that the resulting recomposed product automaton is not necessarily minimal.





This article is an extended and revised version of a previous conference paper~\cite{ReissnerCDRA17}. This latter paper introduced the first technique mentioned above (automata-based alignment). With respect to the conference version, the additional contributions are the idea of using S-components decomposition in conjunction with automata-based alignment and the associated recomposition criteria and algorithms. These contributions are supported by correctness proofs both for the automata-based and for the decomposition-based technique as well as an empirical evaluation based on 40 real-life datasets and three baselines, including two baselines not covered in the conference version.

The next section discusses existing conformance checking techniques. Section 3 introduces definitions and notations related to finite state machines, Petri nets and event logs. Next, Section 4 introduces the automata-based technique, while Section 5 presents the technique based on S-component decomposition. Finally, Section 6 presents the empirical evaluation while Section 7 summarises the contributions and discusses avenues for future work. 

%% file: background_new.tex
\section{Related Work}\label{sec:background}


The aim of conformance checking is to characterise the differences between the behaviour observed in an event log and the behaviour captured by a process model. In this article, we specifically aim to identify behaviour observed in the log that is disallowed by the model (a.k.a.\ unfitting behaviour). Below, we review existing techniques for this task. We first discuss techniques based on \emph{token replay}, which do not aim to achieve optimality. We then review techniques based on (exact) \emph{trace alignment}, which aim to minimise a cost function, such as the length of the computed alignments. Finally, we review two families of approaches to tackle the complexity of computing optimal trace alignments: \emph{approximate trace alignment} and \emph{divide-and-conquer approaches}.



%



\paragraph{Token replay} A simple approach to detect and measure unfitting behaviour is \emph{token-based replay}~\cite{RozinatA08}. The idea is to replay each trace against the process model. In the token-based replay technique presented in~\cite{RozinatA08}, the process model is represented as a Workflow net -- a type of Petri net with a single source place, a single sink place and such that every transition is on a path from the source to the sink. The token replay technique fires transitions in the Petri net, starting from an initial marking where there is one token in the source place, following the order of the events in the trace being replayed. Whenever the current marking is such that the transition corresponding to the next event in the trace is not enabled, the replay technique adds tokens to the current marking so as to enable this transition. Once the sink state of the Petri net is reached, the technique counts the number of \emph{remaining} tokens, i.e. tokens that were left behind in places other than the sink place of the Workflow net. The unfitting behaviour is quantified in terms of the number of added tokens and remaining tokens. 
An extended version of this approach, namely \emph{continuous semantics fitness}~\cite{Medeiros06}, achieves higher efficiency at the expense of incompleteness. Another extension of token replay~\cite{BrouckeMGCBV14} decomposes the model into single-entry single-exit fragments, such that each fragment can be replayed independently. Other extensions based on model decomposition are discussed in~\cite{Munoz-GamaCA14}. 

Intuitively, replay techniques count the number of ``passing errors'' that occur when parsing each trace against the process model. Replay fitness methods fail to identify a minimum number of parsing errors required to explain the unfitting behaviour, thus overestimating the magnitude of differences. To tackle this limitation, several authors have proposed to rely on optimal trace alignment instead of replay~\cite{AdriansyahDA11}.

\paragraph{Trace alignment} Trace alignment techniques extend replay techniques with the idea of computing an optimal alignment between each trace in the log and the closest corresponding trace of the process model (specifically the trace with the smallest Levenshtein distance) . In this context, an alignment of two traces is a sequence of \emph{moves} (or \emph{edit operations}) that describe how two cursors can move from the start of the two traces to their end. In a nutshell, there are two types of edit operations. A \emph{match} operation indicates that the next event is the same in both traces. Hence, both cursors can move forward synchronously by one position along both traces. 
Meanwhile, a \emph{hide} operation (deletion of an element in one of the traces) indicates that the next events are different in each of the two traces. Alternatively, one of the cursors has reached the end of its trace while the other has not reached its end yet. Hence, one cursor advances along its traces by one position while the other cursor does not move.
An alignment is optimal if it contains a minimal number of hide operations. This means that the alignment has a minimal length. 

Conformance checking techniques that produce trace alignments can be subdivided into \emph{all-optimal} and \emph{one-optimal}. A conformance checking technique is called \emph{all-optimal} if it computes every possible minimal-distance alignment between each log trace and the model. Meanwhile, a conformance checking technique is called \emph{one-optimal}, if it computes only one minimal-distance alignment for each log trace. 

The idea of computing alignments between a process model (captured as a Petri net) and an event log was developed in Adriansyah et al.~\cite{AdriansyahDA11, Adriansyah14}. This proposal maps each trace in the log into a (perfectly sequential) Petri net. It then constructs a synchronous Petri nets as a product out of the model and the perfectly sequential net corresponding to the trace. Finally, it applies an A* algorithm to find the shortest path through the synchronous net which represents an optimal alignment.
Van Dongen~\cite{BVD-Alignment} extends Adriansyah et al's technique~\cite{AdriansyahDA11, Adriansyah14} by strengthening the underlying heuristic function. This latter technique was shown to outperform~\cite{AdriansyahDA11, Adriansyah14} 
on an artificial dataset and a handful of real-life event log-model pairs. In the evaluation reported later in this article, we use both~\cite{AdriansyahDA11, Adriansyah14} and \cite{BVD-Alignment} as baselines.

De Leoni et al~\cite{leoni2017} translate the trace alignment problem into an automated planing problem. Their argument is that a standard automated planner provides a more standardised implementation and more configuration possibilities from the route planning domain. Depending on the planner implementation, this approach can either provide optimal or approximate solutions. In their evaluation, De Leoni et al. show that their approach can outperform~\cite{AdriansyahDA11} only on very large process models. Subsequently, \cite{BVD-Alignment} empirically showed that trace alignment techniques based on the A* heuristics outperform the technique of De Leoni et al. Accordingly, in this article we do not retain the technique by De Leoni et al. as a baseline.


In the above approaches, each trace is aligned to the process model separately. An alternative approach, explored in~\cite{GarciaL17}, is to align the entire log against the process model, rather than aligning each trace separately. 
Concretely, the technique presented in~\cite{GarciaL17} transforms both the event log and the process model into \emph{event structures}~\cite{NielsenPW1981}. It then computes a synchronised product of these two event structures. Based on this product, a set of statements are derived, which characterise all behavioural relations between tasks captured in the model but not observed in the log and vice-versa. The emphasis of behavioural alignment is on the completeness and interpretability of the set of difference statements that it produces. As shown in~\cite{GarciaL17}, the technique is less scalable than that of~\cite{AdriansyahDA11, Adriansyah14}, in part due to the complexity of the algorithms used to derive an event structure from a process model. Since the emphasis of the present article is on scalability, we do not retain~\cite{GarciaL17} as a baseline. On the other hand, the technique proposed in this article computes as output the same data structure as~\cite{GarciaL17} -- a so-called Partially Synchronised Product (PSP). Hence, the output of the techniques proposed in this article can be used to derive the same natural-language difference statements produced by the technique in~\cite{GarciaL17}.

\paragraph{Approximate trace alignment} In order to cope with the inherent complexity of the problem of computing optimal alignments, several authors have proposed algorithms to compute approximate alignments.
\emph{Sequential alignments} \cite{Boudewijn17} is one such approximate technique. This technique implements an incremental approach 
to calculate alignments. The technique uses an ILP program to find the cheapest edit operations for a fixed number of steps (e.g.\ three events) taking into account an estimate of the cost of the remaining alignment. The approach then recursively extends the found solution with another fixed number of steps until a full alignment is computed. 
We do not use this approach as a baseline in our empirical evaluation since the core idea of this technique was used in the extended marking equation alignment approach presented in \cite{BVD-Alignment}, which derives optimal alignments and exhibits better performance than Sequential Alignments. In other words, \cite{BVD-Alignment} subsumes \cite{Boudewijn17}. 

Another approximate alignment approach, namely \emph{Alignments of Large Instances} or ALI~\cite{ALI}, finds an initial candidate alignment using a replay technique and improves it using a local search algorithm until no further improvements can be found. The technique has shown promising results in terms scalability when compared to the exact trace alignment techniques presented in~\cite{AdriansyahDA11, Adriansyah14, BVD-Alignment}. Accordingly, we use this technique as a baseline in our evaluation.

An approximate model-log alignment approach for detecting all possible alignments for a trace is the \emph{evolutionary approximate alignments}~\cite{EvolutionaryAllOptimal}. It encodes the computation of alignments as a genetic algorithm. Tailored crossover and mutation operators are applied to an initial population of model mismatches to derive a set of alignments for each trace. In this article, we focus on computing one alignment per trace (not all possible alignments) and thus we do not consider these approaches as baselines in our empirical evaluation. Approaches that compute all-optimal alignments are slower than those that compute a single optimal alignment per trace, and hence the comparison is unfair.

Bauer et al.~\cite{TraceSampling} propose to use \emph{trace sampling} to approximately measure the amount of unfitting behaviour between an event log and a process model.  The authors use a measure of trace similarity in order to identify subsets of traces that may be left out without substantially affecting the resulting measure of unfitting behavior. This technique does not address the problem of computing trace alignments, but rather the problem of (approximately) measuring the level of fitness between an event log and a process model. In this respect, trace sampling is orthogonal to the contribution of this article. Trace sampling can be applied as a pre-processing step prior to any other trace alignment technique, including the techniques presented in this article.


Last, Burattin et al.~\cite{OnlineConformance} propose an approximate approach to find alignments in an online setting. In this approach, the input is an event stream instead of an event log. Since traces are not complete in such an online setting, the technique computes alignments of trace prefixes and estimates the remaining cost of a possible suffix. The emphasis of this approach is on the quality of the alignments made for trace prefixes, and as such, it is not directly comparable to trace alignment techniques that take full traces as input. 



\paragraph{Divide-and-conquer approaches} In divide-and-conquer approaches, the process model is split into smaller parts to speed up the computation of alignments by reducing the size of the search space. Van der aalst et al.~\cite{van2013decomposing} propose a set of criteria for a valid decomposition of a process model in the context of conformance checking. One decomposition approach that fulfills these criteria is the \emph{single-entry-single-exit} (SESE) process model decomposition approach. Munoz-Gama et al.~\cite{Munoz-GamaCA14} present a trace alignment technique based on SESE decomposition. The idea is to compute an alignment between each SESE fragment of a process model and the event log projected onto this model fragment. An advantage of this approach is that it can pinpoint mismatches to specific fragments of the process model. However, it does not compute alignments at the level of the full traces of the log -- it only produces partial alignments between a given trace and each SESE fragment. A similar approach is presented in~\cite{wang2017aligning}. 

Verbeek et al.~\cite{verbeek2016merging} presents an extension of the technique in~\cite{Munoz-GamaCA14}, which merges the partial trace alignments produced for each SESE fragment in order to obtain a full alignment of a trace.
This latter technique sometimes computes optimal alignments, but other times it produces so-called \emph{pseudo-alignments} -- i.e., alignments that correspond to a trace in the log but not necessarily to a trace in the process model. 
In this article, the goal is to produce actual alignments (not pseudo-alignments). Therefore, we do not retain \cite{verbeek2016merging} as a baseline.

Song et al.~\cite{song2017} present another approach for recomposing partial alignments, which does not produce pseudo-alignments. Specifically, if the merging algorithm in~\cite{verbeek2016merging} can not recompose two partial alignments into an optimal combined alignment, the algorithm  merges the corresponding model fragments and re-computes a partial alignment for the merged fragment. 
This procedure is repeated until the re-composition yields an optimal alignment. In the worst case, this may require computing an alignment between the trace and the entire process model. 
A limitation of~\cite{song2017} is that it requires a manual model decomposition of the process model as input. The goal of the present article is to compute alignments between a log and a process model automatically, and hence we do not retain~\cite{song2017} as a baseline.

%% file: approach_new.tex
\newcommand{\minPI}{\mathit{PI}}
\newcommand{\wNetDecomposition}{\mathit{C}}
\newcommand{\key}{\mathit{key}}
\newcommand{\sysNet}{\mathit{SN}}
\newcommand{\wNet}{\mathit{WN}}
\newcommand{\src}[1]{\mathit{src}(#1)}
\newcommand{\tgt}[1]{\mathit{tgt}(#1)}
\newcommand{\lbl}[1]{\lambda(#1)}
\newcommand{\replaceTauBack}{\mathit{replaceTauBackwards}}
\newcommand{\prAlign}{\overline{\epsilon}}

\newcommand{\dom}{\textit{dom }}
\newcommand{\indexNetLabel}{\netLabel}
\newcommand{\indexArcDafsa}{a_{\dafsa}}
\newcommand{\indexArcReachGraph}{a_\reachGraph}
\newcommand{\indexArcLabelDafsa}{l_\dafsa}
\newcommand{\indexArcLabelReachGraph}{l_\reachGraph}
\newcommand{\indexOp}{op}
\newcommand{\logPr}[1]{#1^{a_{\dafsa}}}
\newcommand{\modPr}[1]{#1^{a_{\reachGraph}}}
\newcommand{\nodeLog}{n_{\dafsa}}
\newcommand{\nodeMod}{m_{\reachGraph}}
\newcommand{\ellLog}[1]{\ell_{\dafsa}(#1)}
\newcommand{\ellMod}[1]{\ell_{\reachGraph}(#1)}
\newcommand{\cost}{\mathit{cost}}

\newcommand{\alignMap}{\alpha}
\newcommand{\alignFunc}{\phi}
\newcommand{\futm}{f_{Model}}
\newcommand{\hypArc}{\mathfrak{a}}
\newcommand{\rep}{\mathit{\#i}}

\SetAlFnt{\scriptsize}

\section{Preliminaries}

This section defines the formal concepts and notations used throughout the paper: finite state machines, Petri nets and event logs.
The various concepts presented herein use labelling functions to assign labels to elements. For the sake of uniformity, $\labels$ denotes a finite set of labels and $\tau\in\labels$ is a special ``silent'' label
.
We use Z-notation \cite{Diller90} operators over sequences. Given a sequence $c=\langle x_1, x_2, \dots, x_n\rangle$, $|c|$ denotes the size, and \emph{head} and \emph{tail} retrieve the first and last element of a sequence, respectively, i.e., $|c| = n$, $\mathit{head}(c) = x_1$ and $\mathit{tail}(c) = x_n$. The element at index $i$ in the sequence $c$ is retrieved as $c[i] = x_i$. The operators $\mathit{for}$ and $\mathit{after}$ retrieve the elements before and after $\mathit{i}$ in a sequence, respectively. For example, $\mathit{for}(c,i) = \langle x_1, \dots, x_i\rangle$ and $\mathit{after}(c,i)=\langle x_{i+1}, \dots x_n\rangle$. Finally, \emph{$\mSet$} denotes the multiset representation of a sequence.

\subsection{Finite state machines}

A pervasive concept in our approaches is that of finite state machine (FSM), which is defined as follows.

\begin{definition}[Finite State Machine (FSM)]
Given the set of labels $\labels$, a finite state machine is a directed graph $\fsm = (\nodes, \arcs, \initState, \finalStates)$, where $\nodes$ is a finite non-empty set of states, $\arcs \subseteq \nodes \times \labels \times \nodes$ is a set of arcs, $\initState \in \nodes$ is an initial state,
and $\finalStates \subseteq \nodes$ is a set of final states.
\end{definition}

An arc in a FSM is a triplet $a = ( n_s, l, n_t)$, where $n_s$ is the \emph{source} state, $n_t$ is the \emph{target} state and $l$ is the \emph{label} associated to the arc. We define functions $\src{a} = n_s$ to retrieve the source state, $\lbl{a} = l$ to retrieve the label and $\tgt{a} = n_t$ to retrieve the target state of $a$.
Furthermore, given a node $n \in N$ and arc $a = (n_s,l,n_t) \in \arcs$, let $n \blacktriangleright a = n_t$ if $n = n_s$, and $n \blacktriangleright a = n$ otherwise.
The set of incoming and outgoing arcs of a state $n$ is defined as $\inT{n} = \{a \in \arcs \mid n = \tgt{a}\}$ and $\outT{n} = \{a \in \arcs \mid n = \src{a}\}$, respectively. Finally, a sequence of (contiguous) arcs in a FSM is called a \emph{path}.


\subsection{Process models and Petri nets}\label{sub:reachGraph}

Process models are normative descriptions of business processes and define the expected behavior of the process.
For example, we consider the loan application process model displayed in Fig.~\ref{fig:running-example} using the BPMN notation.
Once this process starts, the credit history, the income sources, personal identification and other financial information are checked. Once the application is assessed, either a credit offer is made, the application is rejected or additional information is requested (the latter leading to a re-assessment).
The process model in Fig.~\ref{fig:running-example} will be used as a running example in this article. For the sake of brevity, we will refer to the activities by the letters attached to the them, e.g, ``Check credit history'' will be referred to as ``A'' etc.
\begin{figure}[htbp]
\centering
\vspace{-0.5em}
\includegraphics[width=0.85\textwidth]{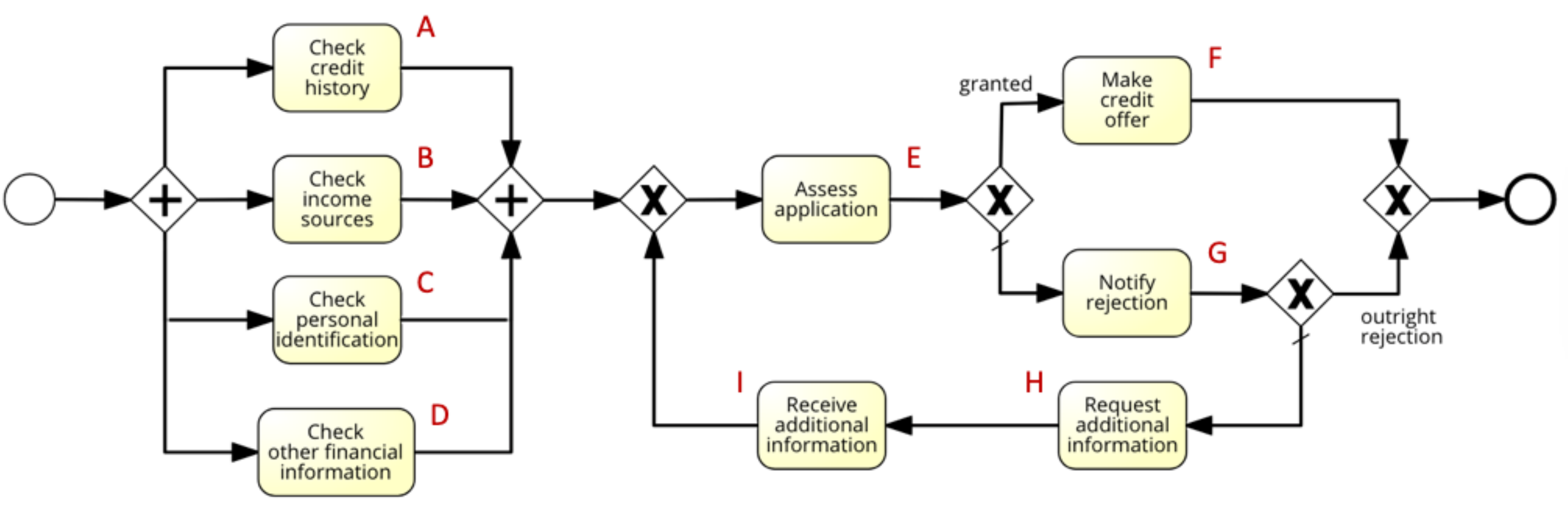}
\vspace{-1em}
\caption{Example loan application process model, adapted from \cite{GarciaL17}.}\label{fig:running-example}
\end{figure}

In the context of this work, business processes are represented as a particular family of Petri nets, namely \emph{labelled free-choice sound workflow nets}. This formalism uses transitions to represent activities, and places to represent resource containers. The formal definition of labelled Petri nets is given next.

\begin{definition}[Labelled Petri net]
A (labelled) \emph{Petri net} is the tuple $\lnet = (\places, \netTransitions, \netArcs, \netLabel)$, where $\places$ and $\netTransitions$ are disjoint sets of \emph{places} and \emph{transitions}, respectively (together called \emph{nodes}); $\netArcs \subseteq (\places \times \netTransitions) \cup (\netTransitions \times \places)$ is the flow relation, and $\netLabel : \netTransitions \to \labels$ is a labelling function mapping transitions to the set of task labels $\labels$ containing the special label $\tau$.
\end{definition}
Transitions labeled with $\tau$ describe invisible actions that are not recorded in the event log when executed. A node $x$ is in the preset of a node $y$ if there is a transition from $x$ to $y$ and, conversely, a node $z$ is in the postset of $y$ if there is a transition from $y$ to $z$. Then, the preset of a node $y$ is the set $\inTr{y} = \{x \in P \cup T | (x, y) \in F\}$ and the postset of $y$ is the set $\outTr{y} = \{z \in P \cup T | (y, z) \in F\}$.

Workflow nets \cite{WFNets} are Petri nets with two special places, an initial and a final place.

\begin{definition}[Labelled workflow net]\label{def:wNet}
A (labelled) \emph{workflow net} is a triplet $\wNet=(\lnet,i,o)$, where $\lnet = (\places, \netTransitions, \netArcs, \netLabel)$ is a labelled Petri net, $i\in\places$ is the initial and $o\in\places$ is the final place, and the following properties hold:
\begin{compactitem}
    \item The initial place $i$ has an empty present and the final place has an empty postset, i.e., $\inTr{i}=\outTr{o}=\varnothing$.
    \item If a transition $t^*$ were added from $o$ to $i$, such that $\inTr{i} = \outTr{o} = \{t^*\}$, then the resulting Petri net is strongly connected.
\end{compactitem}
\end{definition}
The execution semantics of a Petri net can be represented by means of markings. A marking $m: P \rightarrow \mathbb{N}_0$ is a function that associates places to natural numbers representing the amount of \emph{tokens} in each place at a given execution state. As we will later work with the so-called incidence matrix of a Petri net, we define the semantics already in terms of vectors over places. Fixing an order $\{p_1,\ldots,p_k\} = P$ over all places, we write a marking $m$ as a column vector $m = \langle m(p_1), \ldots, m(p_n) \rangle^\intercal$. We slightly abuse notation and write $m$ for both the function and the column vector; further we represent $m$ as the multiset of marked places in our examples. In vector notation, the pre-set $\inTr{t}$ of any transition $t$ defines a column-vector $N^-(t) = \langle x_1,\ldots,x_k \rangle^\intercal$ with $x_i = 1$ if $p_i \in \inTr{t}$, and $x_i = 0$ otherwise. Correspondingly, we define $N^+(t) = \langle z_1,\ldots,z_k\rangle^\intercal$ with $z_i = 1$ if $p_i \in \outTr{t}$, and $z_i = 0$ otherwise, for the post-set of $t$. We lift $+$, $-$, and $\leq$ to vectors by element-wise application.

A transition $t$ is \emph{enabled} at a marking $m$ if each pre-place of $t$ contains a token in $p$, i.e, $N^-(t) \leq m$. An enabled transition $t$ can \emph{fire} and yield a new marking $m' = m - N^-(t) + N^+(t)$ by consuming from all its pre-places ($N^-(t)$) and producing on all its post-places ($N^+(t)$). A marking $m$ is \emph{reachable} from another marking $m'$, if there exists a sequence of firing transitions $\sigma = \langle t_1,\dots t_n\rangle$ such that $\forall 1 \leq i < n : m_{i} = m_{i-1} - N^-(t_i) + N^+(t_i) \land N^-(t_i) \leq m_{i-1}$, where $m_0=m'$ and $m_n=m$. A marking $k$-bounded if every place at a marking $m$ has up to $k$ tokens, i.e., $m(p) \leq k$ for any $p \in P$. A Petri net equipped with an initial marking and a final marking is called a \emph{(Petri) system net}. 
The following definition for a system net refers specifically to workflow nets.

\begin{definition}[System net]
A \emph{System net} $\sysNet$ is a triplet $\sysNet = (\wNet, \initMarking, \finalMarkings)$, where $\wNet$ is a labelled workflow net, $\initMarking$ denotes the initial marking and $\finalMarkings$ denotes the final marking.
\end{definition}
A system net is $k$-bounded if every reachable marking in the workflow net is $k$-bounded. This work considers $1$-bounded system nets that are \emph{sound}~\cite{Verbeek01}, i.e., where from any marking $m$ reachable from $\initMarking$ we can always reach some $m_f \in \finalMarkings$, there is no reachable marking $m > m_f \in \finalMarkings$ that contains a final marking, and each transition is enabled in some reachable marking. Figure \ref{fig:runningExamplePetriNet} shows the system net representation for our running example.

\begin{figure}[htbp]
\centering
\includegraphics[width=0.85\textwidth]{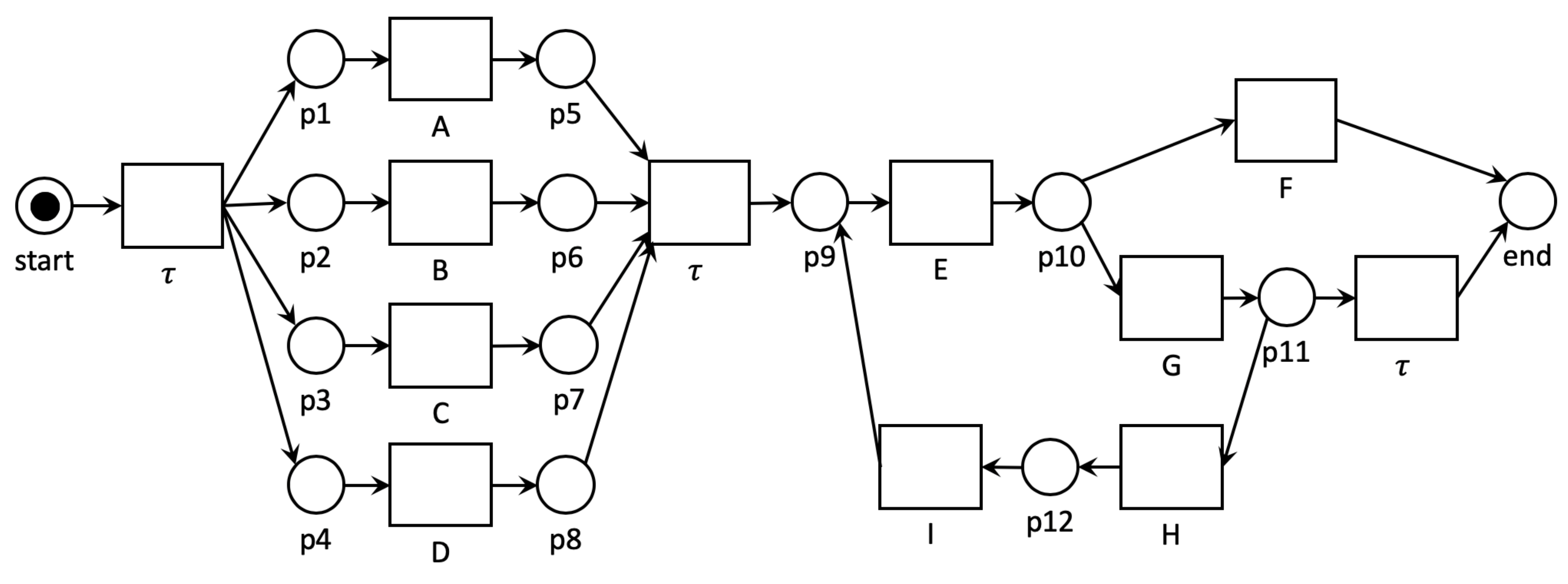}
\vspace{-1em}
\caption{System net representation of the running example of Fig.~\ref{fig:running-example}.}\label{fig:runningExamplePetriNet}
\end{figure}

The reachability graph~\cite{mayr84} of a system net $\sysNet$ contains all possible markings of $\sysNet$ -- denoted as $\markings$. Intuitively, a reachability graph is a non-deterministic FSM where states denote markings, and arcs denote the firing of a transition from one marking to another. The reachability graph for the running example is depicted in Fig. \ref{fig:runningExampleRG} showing markings as multi-sets of places. In this figure, every node contains the places with a token at each of the reachable markings. The complexity for constructing a reachability graph of a safe Petri net is $O(2^{\left|\places \cup \netTransitions\right|})$~\cite{lipton1976}. The formal definition of a reachability graph is presented next.

\begin{definition}[Reachability graph]
The \emph{reachability graph} of a System net $\sysNet = (\wNet, \initMarking, \finalMarkings)$ is a non-deterministic finite state machine $\reachGraph = (\markings, \arcs_{\reachGraph}, \initMarking, \finalMarkings)$, where $\markings$ is the set of reachable markings and $\arcs_{\reachGraph}$ is the set of arcs $\{(m_1, \netLabel(t), m_2) \in \markings \times \labels \times \markings \mid m_2 = m_1 - \inTr{t} + \outTr{t}\}$.
\end{definition}

\begin{figure}[h!]
\centering
\includegraphics[width=1\textwidth]{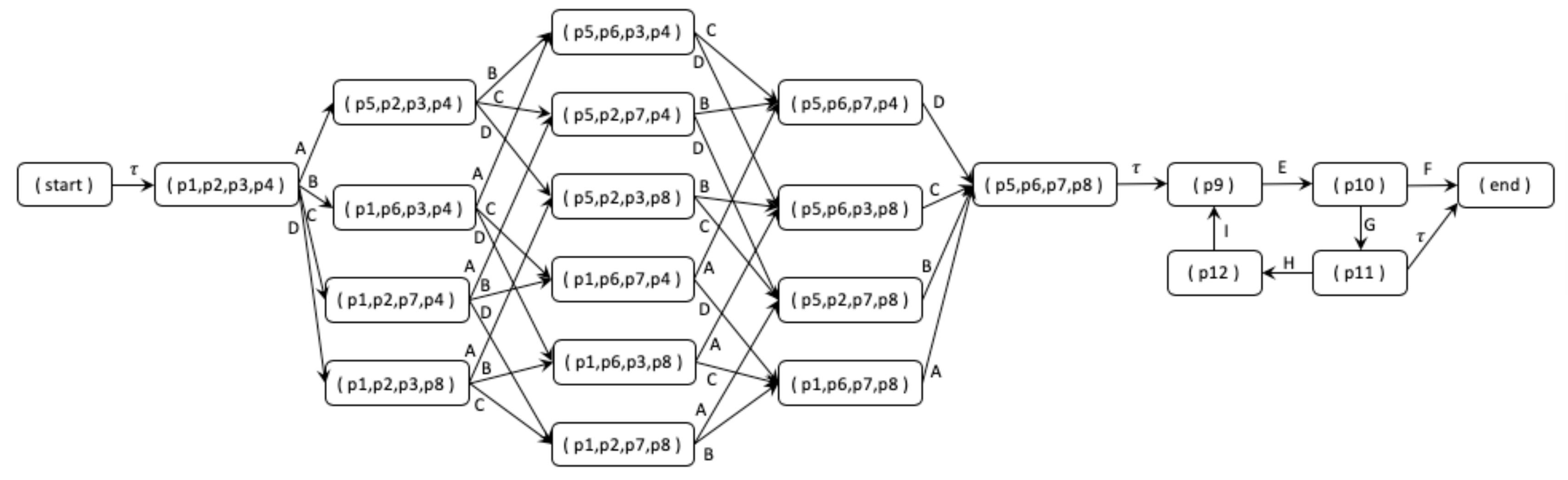}
\vspace{-1em}
\caption{Reachability graph of the running example.}\label{fig:runningExampleRG}
\end{figure}
\newpage
\subsection{Event logs}

\emph{Event logs}, or simply \emph{logs}, record the execution of activities in a business process.
These logs represent the executions of process instances as \emph{traces} -- sequences
of activity occurrences (\emph{a.k.a. events}). A trace can be represented as a sequence of
labels, such that each label signifies an event. Although an event log is a multiset
of traces containing several occurrences of the same trace, we are only interested in the distinct traces in the log and, therefore, we define a log as a set of traces. Figure~\ref{fig:runningExampleLog} depicts an example of a log containing activities of the loan application process in Fig.~\ref{fig:running-example}
We define the concept of a trace and an event log as follows:

\begin{definition}[Trace and event log]
Given a finite set of labels $L$, a \emph{trace} is a finite sequence of labels $\langle l_1, ..., l_n \rangle \in L^*$, such that $l_i \in L$ for any $1 \leq i \leq n$. An \emph{event log} $\logL$ is a set of traces.
\end{definition}

\begin{figure*}[htbp]
\centering
\resizebox{0.3\textwidth}{!}{
 \tikzstyle{block} = [draw, rectangle, fill=white, align=left, minimum height=5mm, text width={width("$t  C, A, B, E, H, I, E, F, G  t$")}, font=\footnotesize]
 \begin{tikzpicture}[>=stealth', node distance=-0.3pt]
  \node[block] (log) {\bf{Log}};
  \node[block, below=of log] (trace1) {$\langle B, D, C, E, G \rangle$};
  \node[block, below=of trace1] (trace2) {$\langle B, D, A, E, F, G \rangle$};
  \node[block, below=of trace2] (trace3) {$\langle C, A, B, E, E, G \rangle$};
  \node[block, below=of trace3] (trace4) {$\langle C, A, B, E, H, I, E, F, G \rangle$};
  \end{tikzpicture}
  }
  \vspace{-.5\baselineskip}
 \caption{Example log for our loan application process.}\label{fig:runningExampleLog}
 \vspace{-1.5\baselineskip}
\end{figure*}

\section{Automata-based conformance cheking}\label{sec:approach}

The objective of conformance checking is to identify a \emph{minimal} set of differences between the behavior of a given process model and a given log. As illustrated in Fig.~\ref{fig:approach}, the first approach proposed in this paper computes this minimal set of differences by constructing an error-correcting product, between the reachability graph of the model and an automaton-based representation of the log (called DAFSA).
An error-correcting product is an automaton that synchronizes the states and transitions of two automata. These synchronizations can represent both common and deviant behavior.
(1) First, the input process model is expanded into a reachability graph. (2) In parallel, the event log is compressed into a minimal, acyclic and deterministic FSM, a.k.a.\ DAFSA. The resulting reachability graph and DAFSA are then compared (3) to derive an error-correcting synchronized product automaton -- herein called a PSP. Each state in the PSP is a pair consisting of a state in the reachability graph and a state in the DAFSA.
A PSP represents a set of trace alignments that can be used for diagnosing behavioral difference statements via further analysis. The trace alignments contained in the PSP are minimal and deterministic (4).
The rest of this section starts by introducing some necessary concepts and is followed by a description of each of the steps including proofs of minimality and determinism of alignments.
\begin{figure}[htbp]
\centering
\includegraphics[scale=0.35]{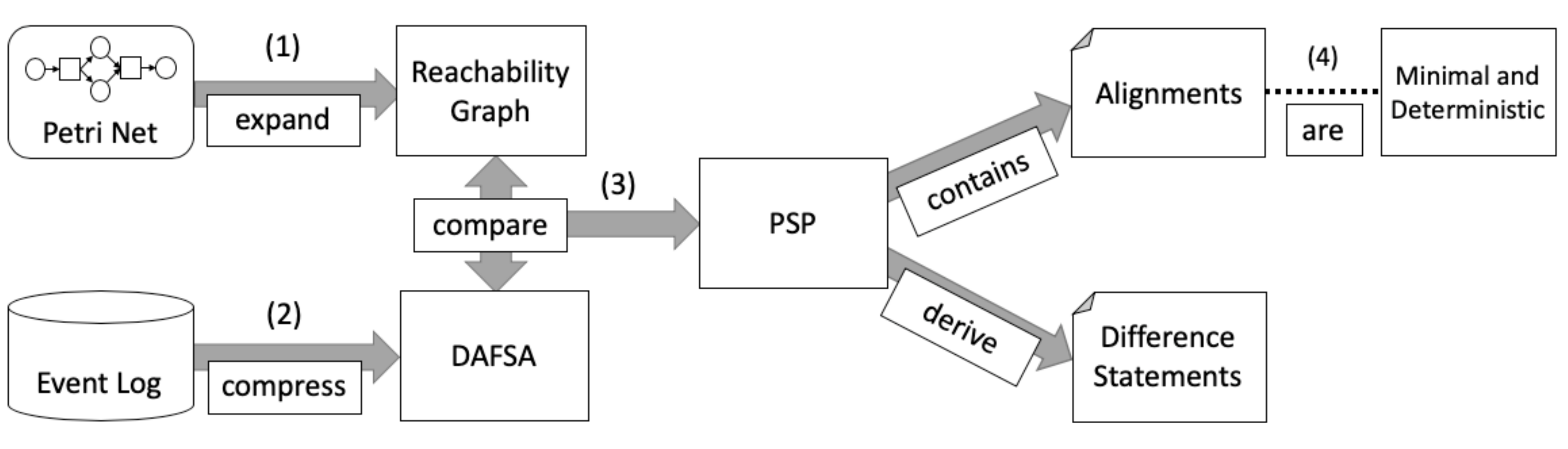}
\vspace*{-.5\baselineskip}
\caption{Overview of the automata-based approach.}\label{fig:approach}
\end{figure}

\subsection{Expanding a Petri net to its $\tau$-less reachability graph (1)}

\begin{algorithm}[h]{
  \SetKwInOut{Input}{input}
  \SetKwProg{Fn}{Function}{}{}
  \Input{Reachability Graph $\reachGraph$}
  $\open \leftarrow \langle \initMarking \rangle$; \tcp{List of markings to check}
  $\closed \leftarrow \{\initMarking\}$; \tcp{Markings checked}
  \While{$\open \neq \langle \rangle$}{\label{line:begin breadth-first}
   $\n \leftarrow \head{\open}$;\tcp{Remove marking m from the head the list}
   $\open \leftarrow \tail{\open}$\;
   $\INV \leftarrow \{a = (m_1, l, \n) \in \inT{\n} \mid l = \tau \land \n \notin \finalMarkings\}$;\tcp{Outgoing $\tau$ arcs at m}
  \For{$\inv \in \INV$}{
   $\replaceTau(\inv, \n, \{ \n\})$;\tcp{Replace $\tau$ arcs}
  }
  $\arcs_{\reachGraph} \leftarrow \arcs_{\reachGraph} \setminus \INV$;\tcp{Remove outgoing $\tau$ arcs}
  \For{$(\n, l, m_2) \in \outT{\n} \mid m_2 \notin \closed$}{\tcp{Insert the target marking of each outgoing arc into the auxiliary lists}
   $\open \leftarrow \open \tieconcat m_2$\;
   $\closed \leftarrow \closed \cup \{m_2\}$\;
  }\label{line:end breadth first}
 }
 $\remnodes \leftarrow \{\n \in \markings \mid (\inT{\n} = \varnothing \land \n \neq \initMarking) \lor (\outT{\n} = \varnothing \land \n \notin \finalMarkings)\}$;\tcp{Non-initial markings with no incoming arcs, and non-final markings with no outgoing arcs. These markings result from the deletion of $\tau$ arcs}
 \While{$ \remnodes \neq \varnothing$}{\label{line:begin remNodes}
  \For{$\n \in \remnodes$}{
   $\arcs \leftarrow \arcs \setminus (\inT{\n} \cup \outT{\n})$\tcp{Remove arcs from every marking in $\remnodes$}
  }
  $\markings \leftarrow \markings \setminus \remnodes$;\tcp{Remove all disconnected markings}
  $\remnodes \leftarrow \{\n \in \markings \mid (\inT{\n} = \varnothing \land \n \neq \initMarking) \lor (\outT{\n} = \varnothing \land \n \notin \finalMarkings)\}$;\tcp{Determine if any more markings became disconnected} \label{line:end remNodes}
 }
    \For{${a = (m_1,l,m_f) \in \inT{m_f} \mid l=\tau}$}
    {
        $\replaceTauBack(a)$;\tcp{Replace $\tau$ arcs that target final markings with a backwards replacement}
    }
 \Return{$\reachGraph$}\;
 \Fn{$\replaceTau((m_1, \tau, \n) \in \arcs, \n_t \in \markings, \Closed \in 2^\markings$) \tcp{Inputs are an arc, a marking and a closed list}}{\label{line:begin replaceTau}
  \For{$(\n_t, l, m_2) \in \outT{\n_t}$}{
   \uIf{$l \neq \tau \vee m_2 \in \finalMarkings$}{
   \tcp{Replace outgoing arcs of the input marking that is not $\tau$ and its target is not final}
    $\arcs_{\reachGraph} \leftarrow \arcs_{\reachGraph} \cup \{(m_1, l, m_2)\}$;  
   }
   \ElseIf{$m_2 \notin \Closed$} {
    $\Closed \leftarrow \Closed \cup \{m_2\}$;\tcp{Add $m_2$ to $\closed$ list if it has not been investigated yet.}
    $\replaceTau((m_1, \tau, \n), m_2, \Closed)$;\tcp{Try to replace the input $\tau$ arc from the new target marking $m_2$}
   }
  }\label{line:end replaceTau}
 }
 \Fn{$\replaceTauBack((m_1,\tau,m_f) \in \arcs_{\reachGraph})$ \tcp{Function for replacing a given $\tau$ arc backwards}}
 {
     \tcp{Replace incoming arcs of $m_1$ with arcs from the predecessors of $m_1$ to $m_f$}
     \For{$(m_2,l,m_1) \in \inT{m_1} \mid l \neq \tau$}
     {
        $\arcs_{\reachGraph} = \arcs_{\reachGraph} \cup \{(m_2,l,m_f)\}$;
     }
     $\arcs_{\reachGraph} = \arcs_{\reachGraph} \setminus \{(m_1,\tau,m_f)\}$; \tcp{Remove the $\tau$ arc from the reachability graph}
 }
 \caption{Remove Tau Transitions}\label{alg:remTau}
}
\end{algorithm}
\setlength{\textfloatsep}{\baselineskip}

Petri nets can contains $\tau$-transitions representing invisible steps that are not recorded in an event log and they are captured in the reachability graph of a Petri net. Our approach aims at matching events from the event log to activities in the process model, and thus $\tau$ transitions can never be matched. Other approaches like \cite{AdriansyahDA11,BVD-Alignment} handle these $\tau$-transitions as automatically matched events. We argue however that a pre-emptive removal of $\tau$ transitions from the reachability graph can speed up the presented technique since less steps need to be considered.
In principle, we assume that a Petri net has a minimal number of $\tau$-transitions, for instance, by applying structural reduction rules that preserve all visible behavior~\cite{murata}. However not all $\tau$-transitions can be removed by structural reduction of the Petri net. We therefore remove the remaining $\tau$-transitions through behavior preserving reduction rules on the reachability graph by the breadth-first search algorithm given in Alg.~\ref{alg:remTau}.
Intuitively, for every marking $m$ reached by a $\tau$-transition $a_1 = (m_1,\tau,m) \in \arcs_{\reachGraph}$ and any outgoing transition $a_2 = (m,l,m_2) \in A_R$, the algorithm replaces $a_1$ with $a_{12} = (m_1,l,m_2)$ (lines 6-8 and lines 19-21).
This replacement is repeated until all arcs representing $\tau$-transitions are removed. In case all incoming arcs of a state get replaced we also remove $m$ and its outgoing arcs (Lines 12-16). Function $\emph{replaceTau}$ also handles the case of another outgoing $\tau$-labeled transition $a_2 = (m,\tau,m_2)$ by a depth-first search along $\tau$-transitions in $A_R$ (lines 22-24).
The algorithm then removes each remaining $\tau$ transition $a = ( m_1, \tau, m_f )$ targeting the final marking while introducing new replacement arcs $a' = ( m_2, l, m_f )$ for each incoming arc of $m_1$, such that $(m_2,l,m_1) \in \arcs_{\reachGraph}$ (Line 17 and function \emph{replaceTauBackwards}).
The reachability graph returned by Alg.~\ref{alg:remTau} is now free of $\tau$ transitions.
Figure~\ref{fig:modelFSM} shows the $\tau$-less reachability graph of the loan application process.
Observe that the node $[p1,p2,p3,p4]$ is removed and its outgoing arcs are connected to the node $[start]$, and, similarly, node $[p5,p6,p7,p8]$ is removed and its incoming arcs now target the node $[p9]$ instead. 
In addition, the arc $([p11],\tau,[\text{end}])$ is replaced with the newly introduced arc $([p10],G,[\text{end} ])$.

\begin{figure}[htbp]
\centering
\includegraphics[width=0.9\textwidth]{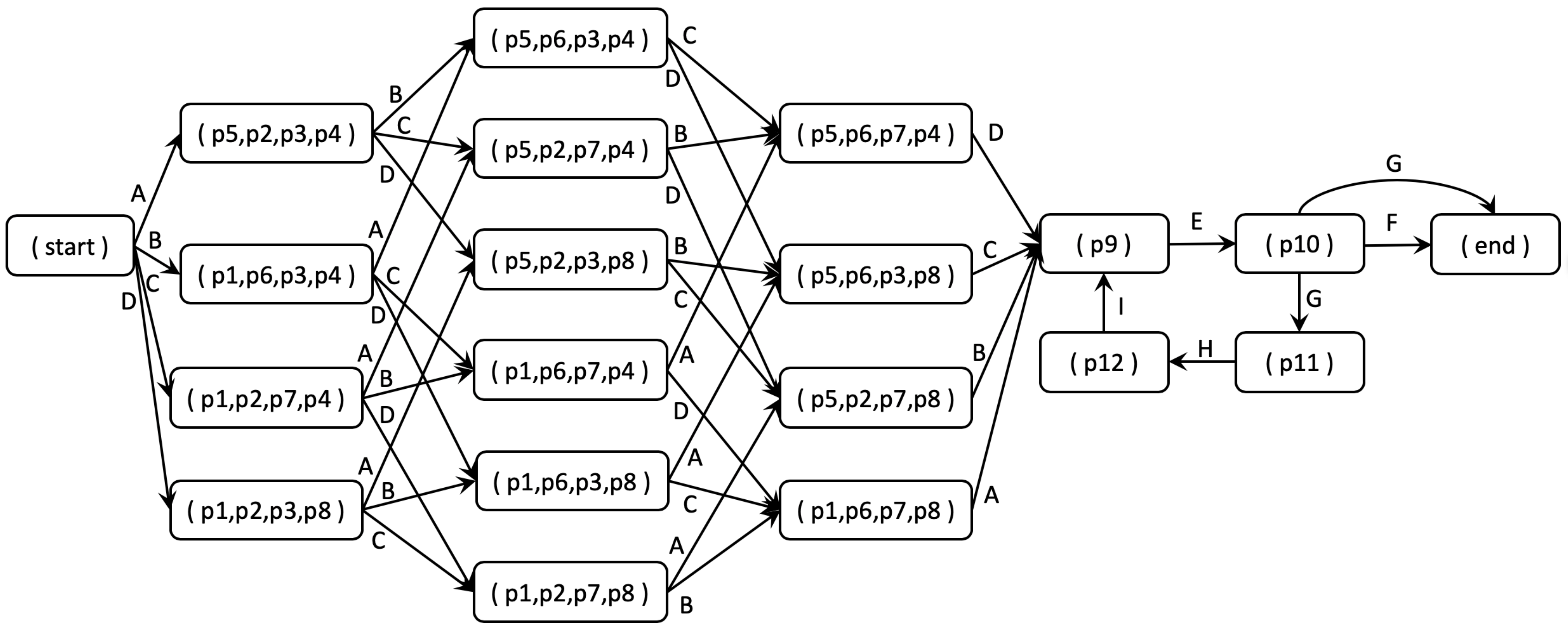}
\vspace{-.5\baselineskip}
\caption{Tau-less reachability graph of the running example.}\label{fig:modelFSM}
\end{figure}
\vspace{-1\baselineskip}

\subsection{Compressing an event log to its DAFSA representation (2)}\label{sub:dafsa}

Event logs can be represented as \emph{Deterministic Acyclic Finite State Automata} (DAFSA), which are acyclic and deterministic FSMs. A DAFSA can represent words, in our case \emph{traces}, in a
compact manner by exploiting prefix and suffix compression.

\begin{definition}[DAFSA]
Given a finite set of labels $L$, a DAFSA is an acyclic and deterministic finite state machine $\dafsa = (\nodes_{\dafsa}, \arcs_{\dafsa}, \initState_{\dafsa}, \finalStates_{\dafsa})$,
where $\nodes_{\dafsa}$ is a finite non-empty set of states, $\arcs_{\dafsa} \subseteq \nodes_{\dafsa} \times L \times \nodes_{\dafsa}$ is a
set of arcs, $\initState _{\dafsa}\in \nodes_{\dafsa}$ is the initial state, $\finalStates_{\dafsa} \subseteq \nodes_{\dafsa}$ is a set of final states. 
\end{definition}

Daciuk et al.~\cite{daciukJ00} present an efficient algorithm for constructing a DAFSA from a set of words. In the constructed algorithm every word is a path from the initial to a final state and, vice versa, every path from an initial to a final state is one of the given words.
We reuse this algorithm to construct a DAFSA from an event log, where the words are the set of traces. The complexity of building the DAFSA is $O(\left| L \right| \cdot \log{} n)$, where $L$ is the set of distinct event labels, and $n$ is the number of states in the DAFSA.


A \emph{prefix} of a state $n \in \nodes_{\dafsa}$ is a sequence of labels associated to the arcs on a path from the initial state to $n$ and, analogously, a \emph{suffix} of $n$ is a sequence of labels associated to the arcs on a path from $n$ to a final state. The prefix of the initial state and the suffix of a final state is $\{\langle\rangle\}$. A state $n$ can have several prefixes, which are denoted by $\prefixF{n} = \bigcup_{(n_s, l, n_t) \in \inT{n}} \{x \tieconcat l \mid x \in \prefixF{n_s}\}$,
where $\tieconcat$ denotes the concatenation operator. Similarly, the set of suffixes of $n$ is represented by
$\suffixF{n} = \bigcup_{(n_s, l, n_t) \in \outT{n}} \{l \tieconcat x \mid x \in \suffixF{n_t}\}$.
Prefixes and suffixes are said to be \emph{common} iff they are shared by more than one trace.

\begin{definition}[Common prefixes and suffixes]
Let $\dafsa = (\nodes_{\dafsa}, \arcs_{\dafsa}, \initState_{\dafsa}, \finalStates_{\dafsa})$ be a DAFSA. The set of common prefixes of $\dafsa$ is the set $\comPrefix = \{\prefixF{n} \mid n \in \nodes_{\dafsa} \wedge \left| \outT{n} \right| > 1\}$. The set of common suffixes of $\dafsa$ is the set $\comSuffix = \{\suffixF{n} \mid n \in \nodes_{\dafsa} \wedge \left| \inT{n} \right| > 1\}$.
\end{definition}
\begin{figure}[htbp]
\centering
\includegraphics[width=0.85\textwidth]{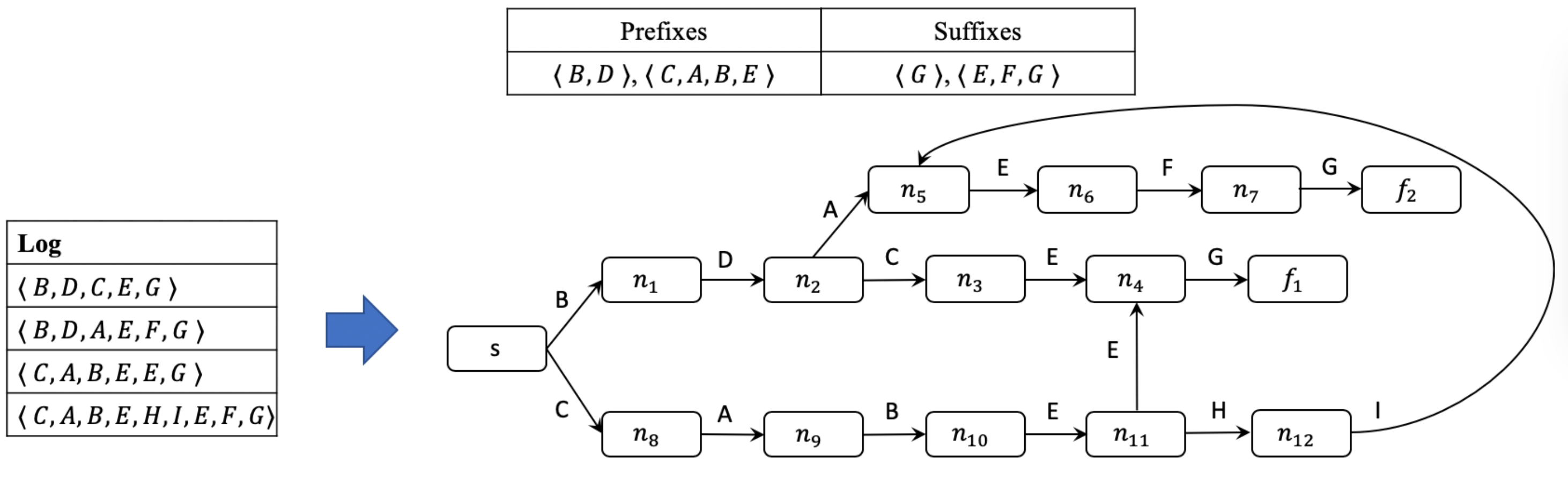}
\vspace{-1em}
\caption{DAFSA representation of the running example log.}\label{fig:runningExampleDAFSA}
\end{figure}


Figure \ref{fig:runningExampleDAFSA} depicts the DAFSA representation and its corresponding common prefixes and suffixes for the example event log in Fig.~\ref{fig:runningExampleLog}. In total, it summarizes 26 events with 16 arcs. All traces in the event log are paths from $s$ to one of the two final nodes $f_1$ or $f_2$. For instance, the trace $\langle B,D,C,E,G \rangle$ is represented by the path $\langle (s,B,n_1),(n_1,D,n_2),(n_2,C,n_3),(n_3,E,n_4),(n_4,G,f_1)\rangle$. In this example, the two common prefixes in nodes $n_2$ and $n_{11}$, as well as the common suffixes from nodes $n_4$ and $n_5$, are shared by two traces in the event log.


\subsection{Comparing a Reachability graph with a DAFSA to derive a PSP (3)}

The computation of similar and deviant behavior between an event log and a process model is based on an error-correcting synchronized product (PSP)~\cite{ArmasBDG2016}. Intuitively, the traces represented in the DAFSA are ``aligned'' with the executions of the model by means of three operations: 
\begin{inparaenum}[(1)]
	\item synchronized move ($\match$), the process model and the event log can execute the same task/event with respect to their label;
	\item log operation ($\lhide$), an event observed in the log cannot occur in the model; and
	\item model operation ($\rhide$), a task in the model can occur, but the corresponding event is missing in the log.
\end{inparaenum}
Both a trace in a log and an execution represented in a reachability graph are totally ordered sets of events (sequences). An alignment aims at \emph{matching} events from both sequences that represent the tasks with the same labels, such that the order between the matched events is preserved. An event that is not matched has to be hidden using the operation \emph{lhide} if it belongs to the log, or \emph{rhide} if it belongs to an execution in the model.
For example, given a trace $\langle B, D, C, E, G\rangle$ and an execution in a model $\langle B, D, C, A, E, G\rangle$, 
the first three activities $B, D$ and $C$ can be matched. Next, activity $A$ in the model cannot be matched, since it is not contained in the trace, and needs to be hidden with an $\rhide$ operation (so that the index in the model execution moves by one position). Last, the remaining activities $E$ and $G$ can be matched again and the matching for this example is complete.

In our context, the alignments are computed over a pair of FSMs, a DAFSA and a reachability graph, therefore the three operations: \emph{match}, \emph{lhide} and \emph{rhide}, are applied over the arcs of both FSMs. A \emph{match} is applied over a pair of arcs (one in the DAFSA and one in the reachability graph) whereas \emph{lhide} and \emph{rhide} are applied only over one arc. We record the type of operation and the involved arcs in a triplet called \emph{synchronization} where $\perp$ denotes the absence of an arc in case of $\lhide$ and $\rhide$.



\begin{definition}[Synchronization]
Let $\dafsa = (\nodes_{\dafsa}, \arcs_{\dafsa}, \initState_{\dafsa}, \finalStates_{\dafsa})$ and $\reachGraph = (\markings, \arcs_{\reachGraph}, \initMarking, \finalMarkings)$ be a DAFSA and a reachability graph, respectively. Then, the set of all possible synchronizations is defined as $\syncs(\dafsa, \reachGraph) = \{(\lhide, a_{\dafsa}, \perp) \mid a_{\dafsa} \in \arcs_{\dafsa}\} \cup \{(\rhide, \perp, a_{\reachGraph}) \mid a_{\reachGraph} \in \arcs_{\reachGraph}\} \cup \{(\match, a_{\dafsa}, a_{\reachGraph}) \mid a_{\dafsa} \in \arcs_{\dafsa} \land a_{\reachGraph} \in \arcs_{\reachGraph} \land \lambda(a_{\dafsa}) = \lambda(a_{\reachGraph})\}$.
\end{definition}

Given a synchronization $\beta = (\op,a_{\dafsa},a_{\reachGraph})$, let $\beta^\ell = \lambda(a_{\dafsa})$ if $a_{\dafsa} \neq \perp$ and $\beta^\ell = \lambda(a_{\reachGraph})$ if $a_{\reachGraph} \neq \perp$; this notation is well-defined as $\lambda(a_{\dafsa}) = \lambda(a_{\reachGraph})$ whenever $a_{\dafsa} \neq \perp \neq a_{\reachGraph}$.
Further, let $\beta^{\op} = \op, \logPr{\beta} = a_{\dafsa},$ and $\modPr{\beta} = a_{\reachGraph}$.
\newpage
All possible alignments between the traces represented in a DAFSA and the executions represented in a reachability graph can be inductively computed as follows. The construction starts by pairing the initial states of both FSMs and then applying the three defined operations over the arcs that can be taken in the DAFSA and in the reachability graph -- each application of the operations (synchronization) yield a new pairing of states. Note that the alignments between (partial) traces and executions are implicitly computed as sequences of synchronizations.

Given a sequence of synchronizations $\epsilon = \langle \beta_1, \dots, \beta_m \rangle$ with $\beta_i = (\mathit{op}_i,a_{i,\dafsa},a_{i,\reachGraph}), 1\leq i \leq m$, we define two projection operations $\epsilon\trPr{\dafsa}$ and $\epsilon\trPr{\reachGraph}$ that retrieve the sequence of arcs for the DAFSA and the reachability graph, respectively. The \emph{projection} onto $\dafsa$ is the sequence $\epsilon\trPr{\dafsa} = \langle a_{1,\dafsa},\ldots,a_{m,\dafsa} \rangle \trPr{A_\dafsa}$ of the $\dafsa$-entries in $\epsilon$ projected onto the arcs in $\dafsa$ (i.e., removing all $\perp$). Correspondingly, $\epsilon\trPr{\reachGraph} = \langle a_{1,\reachGraph},\ldots,a_{m,\reachGraph} \rangle \trPr{A_\reachGraph}$. Thus, $\epsilon\trPr{\dafsa}$ ($\epsilon\trPr{\reachGraph}$) contains the arcs of all $\mathit{match}$ and $\mathit{lhide}$ ($\mathit{rhide}$) triplets.
On top of that notation, we are interested in the sequence of labels represented by a sequence of arcs, shorthanded as $\lambda(\epsilon\trPr{\dafsa}) = \langle \lambda(a_1), \dots, \lambda(a_n) \rangle$.

\begin{definition}[(Proper) Alignment]\label{definition:alignment}
Given a DAFSA $\dafsa$, a reachability graph $\reachGraph$ and a trace $c\in\logL$, an alignment is defined as a sequence of synchronizations $\epsilon_c = \langle\beta_1, \beta_2, \dots \beta_m \rangle$, where $\beta_i \in \syncs(\dafsa, \reachGraph) \land 1 \leq i \leq \left|c\right| \leq m$. A \emph{proper alignment} for the trace $c$ fulfils two properties:
    \begin{compactenum}
        \item The sequence of synchronizations with $\lhide$ or $\match$ operations reflects the trace $c$, i.e. $\lambda(\epsilon_c\trPr{\dafsa}) = c$.
        \item The arcs of the reachability graph in the sequence of synchronizations with $\rhide$ or $\match$ operations forms a path in the reachability graph from the initial to the final marking, i.e. let $n=\left|\epsilon_c\trPr{\reachGraph}\right|$, then $\src{\epsilon_c\trPr{\reachGraph}[1]} = \initMarking \land \tgt{\epsilon_c\trPr{\reachGraph}[n]} = \finalMarkings \land \forall \leq i < n: tgt( \epsilon_c \trPr{\reachGraph} [i]) = src( \epsilon_c \trPr{\reachGraph} [i+1])$.
    \end{compactenum}
The set of all proper alignments for a given trace $c$ is denoted as $\configs(c,\reachGraph)$. We write $\epsilon_{c|op} = \{\beta = (op, a_{\dafsa},a_{\reachGraph}) \in \epsilon_c\}$ for the synchronizations of a particular operation $op$ in a given alignment $\epsilon_c$.
\end{definition}

Figure~\ref{fig:ExampleAlignment} illustrates the alignment for the matching example between the trace $\langle B, D, C, E, G\rangle$ and the model execution $\langle B, D, C, A, E, G\rangle$. Please note that this alignment is proper since all trace labels are present and the arcs of the reachability graph form a path in the automata displayed in Fig.~\ref{fig:modelFSM} to the final marking $[end]$.
\begin{figure}[htbp]
\centering
\includegraphics[width=1\textwidth]{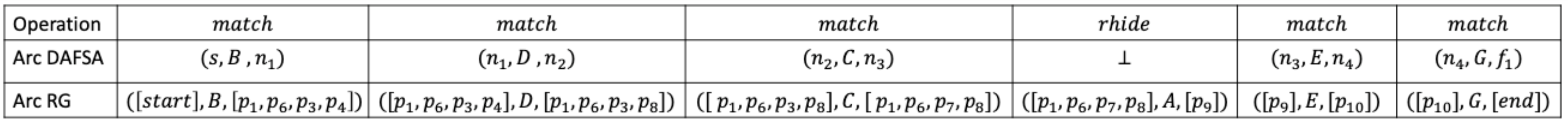}
\vspace{-1.5\baselineskip}
\caption{Proper alignment for the matching example.}\label{fig:ExampleAlignment}
\end{figure}

A cost can be associated to a proper alignment for a given trace. If an asynchronous $\lhide$ or $\rhide$ move is associated to a non-$\tau$ label then the cost increases. Assuming that the cost of hiding a non-$\tau$ transition is 1, the cost function is given as follows:

\begin{definition}[Alignment cost function]\label{definition:cost}
Given an alignment $\epsilon \in \configs(c,\reachGraph)$, the cost function $\cost : \configs(c,\reachGraph) \rightarrow \mathbb{N}$ for $\epsilon$ is defined as $\cost(\epsilon) = \left|\{\beta \in \epsilon \mid \beta^\ell \neq\tau \land \beta^\op \neq\match\}\right|$.
\end{definition}

All alignments can be collected in a finite state machine called PSP~\cite{ArmasBDG2016}. Every state in the PSP is a triplet $( n, m, \epsilon )$, where $n$ is a state in the DAFSA $\dafsa$, $m$ is a state in the reachability graph $\reachGraph$ and $\epsilon$ is the sequences of arcs taken in the $\dafsa$ and in $\reachGraph$ to reach $n$ and $m$; every arc of the PSP is a synchronization of $\dafsa$ and $\reachGraph$; the pairing of the initial states is the initial state of the PSP; and the finial states are those with no outgoing arcs.

\begin{definition}[PSP]\label{def:psp}
Given a DAFSA $\dafsa$ and a reachability graph $\reachGraph$, their PSP $\psp$ is a finite state machine $\psp = (\nodes_{\psp},\arcs_{\psp},\initState_{\psp},\finalStates_{\psp})$, where $\nodes_{\psp} \subseteq \nodes_\dafsa \times \markings \times \configs$ is the set of nodes, $\arcs_{\psp} = \nodes_{\psp}  \times\syncs(\dafsa, \reachGraph) \times \nodes_{\psp}$ is the set of arcs, $\initState_{\psp}= (\initState_\dafsa, \initMarking, \langle \rangle) \in \nodes_{\psp}$ is the initial node, and $\finalStates_{\psp} = \{\f \in \nodes_{\psp} \mid \outT{\f} =\varnothing\}$ is the set of final nodes.
\end{definition}

The PSP contains all possible alignments, however we are interested in the proper alignments with minimum cost. These alignments are called \emph{optimal}. The computation
of all possible alignments can become infeasible when the search space is too large. Thus, we use an $A^*$ algorithm~\cite{hart68} to consider the most promising paths in the PSP first, i.e., those minimizing the number of hides.
\newpage
The goal of our $A^*$-search is the same as other ``traditional'' techniques, such as \cite{AdriansyahDA11, BVD-Alignment}, i.e. to find alignments for each log trace while minimizing the number of mismatches.
However, the difference between the traditional approaches and the technique presented in this paper relies in the data structures used. Specifically, while the traditional approaches resemble the search for the shortest path through a synchronous net, our $A^*$ search resembles an error-correcting bi-simulation over two automata structures.
By applying the bi-simulation idea, we only need to compute the reachability graph of the workflow net once and do not need to compute a synchronous net for every unique trace in the event log.
We define the cost function for the $A^*$ as follows.

\begin{definition}[$A^*$-cost function]\label{def:cost}
Let $\logL$ and $\psp$ be a given event log and PSP, then for every trace $c \in \logL$ and every node $x = (n,m,\epsilon)\in\psp$ we define a cost function $\costF(x,c) = \curCost(x,c) + \futCost(x,c)$ that relies on the current cost function $\curCost$ and a heuristic function $\futCost$ for estimating future hides for a given trace. We define functions $\curCost$ and $\futCost$ as follows:
\begin{equation}
 \begin{split}
 \curCost(x,c) = \begin{cases}\cost(\epsilon(x)), & \textit{if } \epsilon(x)\trPr{\dafsa} = \mathit{for}(c,\left|\epsilon(x)\trPr{\dafsa}\right|) \\ \infty, & \mathit{otherwise} 
 \end{cases}
 \end{split}
 \end{equation}

 \centering$\futCost(x,c) = min \{\left|\futL(x,c) \setminus f_{Model} \right| + \left| f_{Model} \setminus \futL(x,c)\right| \mid f_{Model} \in \futM(x)\}$
\end{definition}


%
%

Function $\curCost$ returns the current cost for a given node $x$ in the PSP and a given trace $c$ to align. If the trace labels of the partial alignment of $x$, i.e. $\epsilon(x)\trPr{\dafsa}$, fully represent a prefix of $c$ then the cost of $\epsilon(x)$ is that of the $\cost$ function defined in Def. \ref{definition:cost}. Otherwise, node $x$ is not relevant to trace $c$ and the cost is set to $\infty$ to avoid considering this node in the search.
Function $\futCost$ relies on two functions $\futL$ and $\futM$. 
$\futL(x,c) = \mSet(c) \setminus \epsilon\trPr{\dafsa}$ denotes the multiset of future trace labels and $\futM$ is the set of multisets of future model labels.
The set of future model labels $\futM(x)$ is computed in a backwards breadth-first traversal over the strongly connected components of the reachability graph from each of its final markings. The multisets of task labels are collected during the traversal and stored in each node of the graph.
All labels from cyclic arcs inside strongly connected components are gathered during the traversal with a special symbol $\omega$ representing that the label can be repeated any number of times.
For the comparison of these labels to achieve an underestimating function, we set these labels to infinity for the term $\left|\futL \setminus \futm\right|$ and to 0 for the term $\left|\futm \setminus \futL \right|$, i.e. we assume that repeated task labels match all corresponding labels in the trace.
Observe that $\futCost$ assumes that all events with the same label in $\futL$ and $f_{Model}$ are matched, this is clearly an optimistic approximation, since some of the those matches might not be possible; then the optimistic approximation computed by $\futCost$ guarantees the optimality of the alignments; $\futCost$ is admissible.

\begin{algorithm}[!h]{
 \SetKwInOut{Input}{input}
 \SetKwProg{Fn}{Function}{}{}
 \Input{Event Log $\logL$, DAFSA $\dafsa$, Reachability Graph $\reachGraph$, }

 \For{$c \in \logL$}{
  $\open \leftarrow \{(\initState_{\psp}, \costF(\initState_{\psp},c))\}$\; \label{line:initAStar}
  $\maxcost \leftarrow \left| c \right| + $ minModelSkips\;
  \While{$\open \neq \varnothing$}{
   choose a tuple $(\actN=(n_{\dafsa}, m, \epsilon), \rho) \in \open$, such that $\nexists (n_{\psp}', \rho') \in \open : \rho > \rho'$\;
   $\open \leftarrow \open \setminus \{(\actN, \rho)\}$\;
   \uIf{$n_{\dafsa} \in \finalStates_{\dafsa} \land m \in M_f \land \epsilon\trPr{\dafsa} = c$}{
    \If{$\costF(\actN,c) < \maxcost$}{
     $\maxcost \leftarrow \costF(\actN,c)$\;
     $\optimal \leftarrow \varnothing$\;
     $\open \leftarrow \{(\node, \costF(\node,c)) \in \open \mid \costF(\node,c) \leq \maxcost\}$
    }
    $\optimal \leftarrow \optimal \cup \{\actN\}$\;
   }
   \Else{\label{line:begin rep2}
    $\potN \leftarrow \varnothing$\; \label{line:end rep2}\label{ln:beginSkip}
    \For{$\out_{\dafsa} = (n_{\dafsa},l_{\dafsa},n_t) \in \outT{n_{\dafsa}} \mid l_{\dafsa} = c(\left| \epsilon\trPr{\dafsa}\right| + 1)$}
    {
     $\potN \leftarrow \potN \cup \{(n_t,m,\epsilon \tieconcat(\lhide,\out_{\dafsa},\perp))\}$\;
     \For{$\out_{\reachGraph} = (m, l_{\reachGraph},m_t) \in \outT{m} \mid l_{\reachGraph} = l_{\dafsa}$}{
      $\potN \leftarrow \potN \cup \{(n_t, m_t,\epsilon \tieconcat (\match, \out_{\dafsa}, \out_{\reachGraph}))\}$
     }
    }
    \lFor{$\out_{\reachGraph} = (m,l_{\reachGraph},m_t) \in \outT{m}$}{
     $\potN \leftarrow \potN \cup \{n_{\dafsa},m_t,\epsilon \tieconcat (\rhide, \perp, \out_{\reachGraph}))\}$
    }
    $\open \leftarrow \open \cup \{(\nextN,\costF(\nextN,c)) \mid \nextN \in \potN \land \costF(\nextN,c) \leq \maxcost \}$\;\label{ln:endSkip}
   }
  }
  \lFor{$f \in \optimal$}{
   \textit{InsertIntoPSP}$(f,c,\psp)$
  }
 }
 \Return{$\psp$}\;
 \caption{Construct the PSP}\label{alg:constPSP}
}
\end{algorithm}

Algorithm \ref{alg:constPSP} shows the procedure to build the PSP, where an $A^*$ search is applied to find all optimal alignments for each trace in a log. The algorithm chooses a node with minimal cost $\costF$, such that if it pairs two final states (one in the DAFSA and one in the reachability graph) -- representing the alignment of a complete trace -- then it is marked as an optimal alignment. Otherwise, the search continues by applying $\lhide$, $\rhide$ and $\match$. 
As shown in~\cite{GarciaL17}, the complexity for constructing the PSP is in the order of $O(3^{|N_{\dafsa}|\cdot|M|})$ where $N_{\dafsa}$ is the set of states in the DAFSA and $M$ is the set of reachable markings of the Petri net.

\begin{algorithm}[!h]{
  \SetKwInOut{Input}{input}
  \SetKwProg{Fn}{Function}{}{}
  \vspace{0.5em}
  \begin{compactitem}
   \item[$\triangleright$] replace line \ref{line:initAStar} in Alg.~\ref{alg:constPSP} with the following block:
  \end{compactitem}
  \vspace{-0.8em}
  \begin{mdframed}[userdefinedwidth=0.9\textwidth,backgroundcolor=lightgray!20]
   \begin{algorithmic}[0]
    \begin{compactitem}
     \Indentp{-2em}
     \item[$\triangleright$] Reuse common prefix alignments
    \end{compactitem}
    \Indentp{-1em}
    \lFor{$i = 1 \rightarrow \left| c \right|$}{
     $\open \leftarrow \open \cup \{(\nextN, \costF(\nextN,c)) \mid \nextN \in \prefMem(c \textit{ for } i)\} $
    }
    \lIf{$\open =\varnothing$}{
     $\open \leftarrow \open \cup \{(\initState_{\psp}, \costF(\initState_{\psp},c)\}$
    }	
   \end{algorithmic}
  \end{mdframed}
  \vspace{-0.8em}
  \begin{compactitem}
   \item[$\triangleright$] replace line \ref{line:end rep2} in Alg.~\ref{alg:constPSP} with the following block:
  \end{compactitem}
  \vspace{-0.8em}
  \begin{mdframed}[userdefinedwidth=0.9\textwidth,backgroundcolor=lightgray!20]
   \begin{algorithmic}[0]
    \begin{compactitem}
     \Indentp{-2em}
     \item[$\triangleright$] Reuse common suffix alignments
    \end{compactitem}
    \Indentp{-1em}
      $\suffix_{act} \leftarrow c \textit{ after } \left| \{ \sync = (\op, a_{\dafsa}, a_{\reachGraph}) \in \epsilon \mid \op \neq \rhide \} \right|$\\;
     $\potN \leftarrow \{(f_{\dafsa},f_{\reachGraph}, \epsilon \tieconcat \conf_{\suffix}) \mid (f_{\dafsa},f_{\reachGraph},\conf_{\suffix}) \in \suffMem(n_{\dafsa},m,\suffix_{act})\}$\;
     $\open \leftarrow \open \cup \{(\nextN, \costF(\nextN,c)) \mid \nextN \in \potN\}$\;
     \lIf{$\potN \neq \varnothing$}{
      \bf continue\label{ln:continue}
     }
   \end{algorithmic}
  \end{mdframed}
 \caption{Construct the PSP with Prefix- and Suffix Memoization}\label{alg:constPSPwithPrefAndSuff}
 }
\end{algorithm}
In order to optimize the computation of the PSP, two memoization tables are used: prefix and suffix.
Both tables store partial trace alignments for common prefixes and suffixes that have been aligned previously.
The integration of these tables requires the modification of Alg.~\ref{alg:constPSP}, as shown in Alg.~\ref{alg:constPSPwithPrefAndSuff}. For each trace $c$, the algorithm starts by checking if there is a common prefix for $c$ in the prefix memoization table. If this is the case, the $A^*$ starts from the nodes stored in the memoization table for the partial trace alignments that have been previously observed.
In the case of common suffix memoization, the algorithm checks at each iteration whether the current pair of nodes and the current suffix is stored in the suffix memoization table. If this is the case, the algorithm appends nodes to the $A^*$ search for each pair of memoized final nodes and appends all partial suffix alignments to the current alignment
 \emph{instead} of continuing the regular search procedure.
Please note that the command {\bf continue} in line~\ref{ln:continue} of Alg.~\ref{alg:constPSPwithPrefAndSuff} refers to the while loop of lines~\ref{line:beginWhile}-\ref{line:endWhile} in Alg.~\ref{alg:constPSP}. This command skips lines~\ref{ln:beginSkip}-\ref{ln:endSkip} where we would usually offer new nodes to the open queue, and it is not necessary anymore because we had already inserted a node for the optimal suffix after the current node.
By reusing the information stored in these tables, the search space for the $A^*$ is reduced.

The approach illustrated so far produces a PSP containing all optimal alignments. Nevertheless, if only one optimal alignment is required, then the algorithm can be easily modified to stop as soon as the first alignment is found.
Overall, the complexity of the proposed approach is exponential in the worst case, i.e. $O(\left| \labels \right| \cdot \log{} n + 2^{\left| \places \cup \netTransitions \right|} + 3^{|N_{\dafsa}|\cdot|M|})$.

\begin{figure}[hbtp]
\centering
\includegraphics[width=\textwidth]{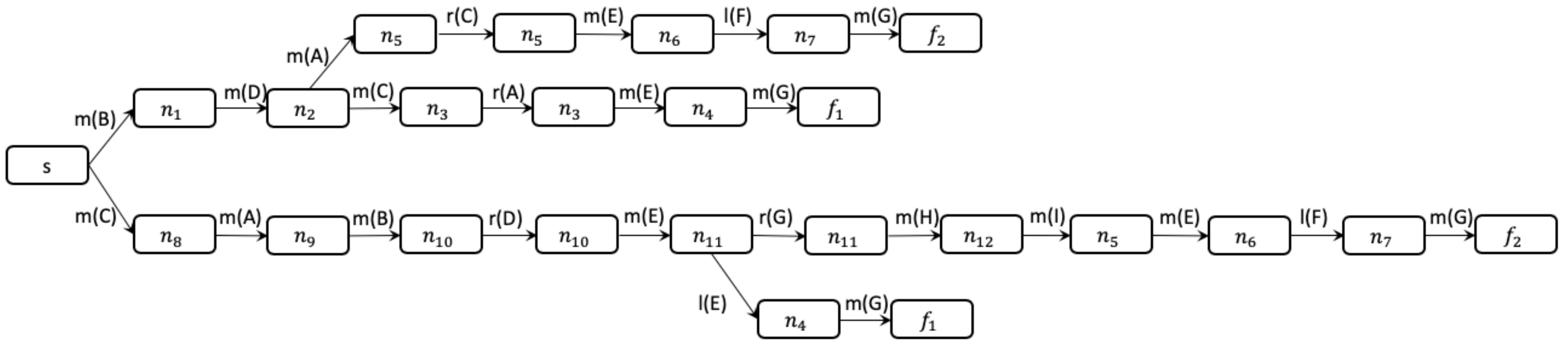}
\vspace{-1.5\baselineskip}
 \caption{\label{fig:PSP}The PSP for our loan application process example.}
\end{figure}

Figure~\ref{fig:PSP} shows an abbreviated PSP obtained by synchronizing the DAFSA of the loan application process in Fig.~\ref{fig:runningExampleDAFSA} and the $\tau$-less reachability graph of Fig.~\ref{fig:modelFSM}. The PSP shows the one-optimal alignments and abbreviates states in the PSP with only states in the DAFSA for readability purposes, and the name of the operations are shorthanded with the initial letter and the label of the activity, i.e., $(match, X) = m(X)$.
To understand its construction let us consider the sample trace $\langle B,D,C,E,G \rangle$. Starting from the source node of the PSP $\initState_{\psp}$,
the $A^*$ search explores the outgoing arc with label $B$ of the initial state of the DAFSA and all outgoing arcs of the initial marking of the reachability graph, i.e. arcs with the activity labels $A,B,C$ and $D$.
Thus, it will compute the cost of performing the following possible synchronizations: $(\lhide, B)$,  $(\rhide, A)$, $(\rhide, B)$, $(\rhide, C)$, $(\rhide, D)$ and also $(\match,B)$, because both, the DAFSA and the reachability graph, can execute $B$ at their current state.
Out of these six possibilities it will only explore $(\match, B)$\footnote{In case of $(\match, B)$ we have a current cost of zero since it is a match (i.e.\ $\curCost(\node) = 0$), and a future cost of one (i.e.\ $\futCost(\node, c) = \left|\{C,D,E,G\} \setminus \{A,C,D,E,G\}\right| + \left|\{A,C,D,E,G\} \setminus \{C,D,E,G\}\right| = 1$).} and $(\rhide, A)$ which have a cost of one. Other synchronizations like $(\rhide, B)$\footnote{In case of $(\rhide, B)$ we have a current cost of one since it is a hide (i.e.\ $\curCost(\node) = 1$), and a future cost of three (i.e.\ $\futCost(\node, c) = \left|\{B,C,D,E,G\} \setminus \{A,C,D,E,G\}\right| + \left|\{A,C,D,E,G\} \setminus \{B,C,D,E,G\}\right| = 2$).} will never be explored since they have a cost of three and there exist nodes with a lower cost. The $A^*$ search will continue exploring the possible synchronizations until all optimal alignments are discovered.

\subsection{Alignments are minimal and deterministic (4)}\label{sub:deterministicAlignments}

This section shows that the computed alignments are deterministic and minimal. 
In particular, we introduce rules to break ties between alignments with the same cost. While determinism is not a necessary property of alignments, it is a desired property for end-users that rely on the output of conformance checking techniques to, for instance, improve their process models. Newer studies, such as the \cite{21Conf}, deem deterministic results of conformance measures as desirable in process mining.

\textbf{Alignments are deterministic.}
A trace can have several optimal alignments, however, in order to have a deterministic computation of a single optimal alignment, we define an order on the construction of the PSP. This order is imposed on the operations, with the following precedence order: $\match > \rhide > \lhide$, and on the lexicographic order of the activity labels. We apply this precedence order at each iteration of the $A^*$-search on the set of candidate nodes of the queue that all have the lowest cost values w.r.t. $\rho$. In that way the $A^*$ search will still always explore the cheapest nodes first and guarantees to find an alignment with optimal cost. The precedence order merely provides a tool to deterministically select an optimal alignment from the set of optimal alignments with a specific order of operations and activity labels already during the exploration of the search space.

\begin{figure}[!h]
\centering
\includegraphics[scale=0.3]{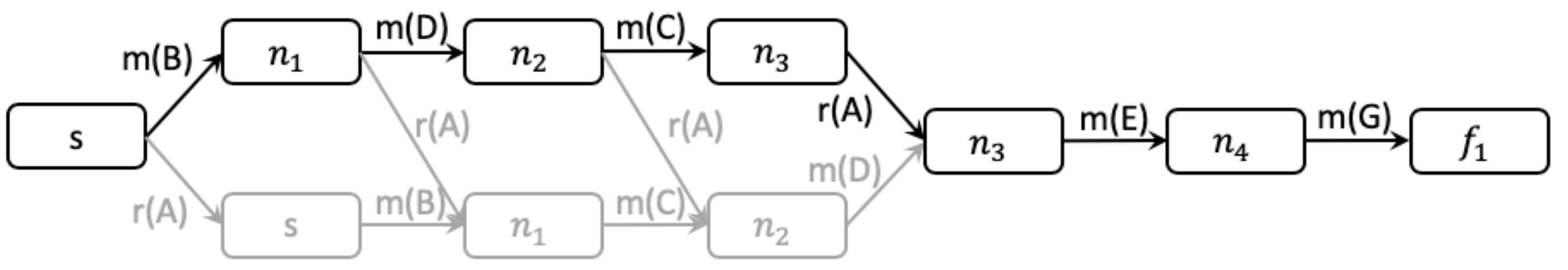}
 \caption{\label{fig:exampleDeterministic}Deterministic one-optimal alignment for trace $\langle B,D,C,E,G \rangle$.}
\end{figure}

We choose to prioritize $\rhide$ over $\lhide$ synchronizations in the preference order to increase the number of $\match$ synchronizations in the returned optimal alignment. We would like to remind the reader that an increase in $\match$ synchronizations does not change the cost function for an alignment as per Def.~\ref{def:cost}. An alignment with more $\match$ synchronizations, however, can link the observed trace more closely to the process model.
The following lemma shows that for optimal alignments, more $\rhide$ synchronizations lead to more $\match$ synchronizations.
Fig. \ref{fig:exampleDeterministic} demonstrates all 4 possible optimal alignments with the same cost for trace $\langle B,D,C,E,G\rangle$ of the loan application example with one mismatch (missing activity A in the parallel block). Out of these four, we select alignment $\langle m(B),m(D),m(C),r(A),m(E),m(G)\rangle$ according to the precedence order.

\begin{lemma}\label{theorem:rhideMatch}
Let $\epsilon_c$ be an optimal alignment for a trace $c$. For any other optimal alignment $\epsilon'_c$ for $c$, such that $\epsilon'_{c{|rhide}} < \epsilon_{c{|rhide}}$, then $\epsilon'_{c{|match}} < \epsilon_{c{|match}}$.
\end{lemma}

\begin{proof}
Given two optimal alignments $\epsilon_c,\epsilon'_c$, it holds that these two alignments have the same cost according to Def.~\ref{definition:cost}, i.e. $\cost(\epsilon_c)=\cost(\epsilon'_c)$, and these two alignments are proper according to Def.~\ref{definition:alignment}. Further, we assume that $\epsilon_c$ has more $\rhide$ synchronizations than $\epsilon'_c$, i.e. $\epsilon'_{c{|rhide}} < \epsilon_{c{|rhide}}$.
As a first step, we assume that $\epsilon_c$ has exactly one more $\rhide$ synchronization than $\epsilon'_c$, i.e. $\epsilon_{c{|rhide}} = \epsilon'_{c{|rhide}} + 1$.
The cost of an alignment is the number of $\rhide$ and $\lhide$ synchronizations disregarding all synchronizations involving $\tau$. Since we remove all $\tau$-labelled transitions in Alg.~\ref{alg:remTau}, the cost of an alignment equals exactly to the number of $\rhide$ and $\lhide$ synchronizations. By the assumptions, $\epsilon_c$ has one more $\rhide$ synchronization than, and the same cost as, $\epsilon'_c$ and so it follows that $\epsilon'_c$ has exactly one more $\lhide$ synchronization for a trace label $\ell$ than $\epsilon_c$, i.e. $(\lhide, \ell) \in \epsilon'_c \land (\lhide, \ell) \notin \epsilon_c \land \ell \in c$. 
Since both alignments properly represent the trace, the sum of their $\lhide$ and $\match$ synchronizations is equal to the size of the trace $\left|c\right|$. Therefore, $\epsilon_c$ needs to have one more $\match$ synchronizations than $\epsilon'_c$, in particular $(\match, \ell) \in \epsilon_c \land (\match, \ell) \notin \epsilon'_c \land \ell \in c$. The general case of multiple $\rhide$ synchronizations follows from inductive reasoning. If an optimal alignment $\epsilon_c$ has $x$ more $\rhide$ than another optimal alignment $\epsilon'_c$, then $\epsilon'_c$ must have $x$ more $\lhide$ than $\epsilon_c$ because they have the same cost. Similarly, $\epsilon_c$ must have $x$ more $\match$ synchronizations than $\epsilon'_c$ since the number of $\lhide$ and $\match$ synchronizations needs to equal to the size of the trace $\left|c\right|$. Hence, it holds for two optimal alignments $\epsilon_c, \epsilon'_c$ with $\epsilon'_{c{|rhide}} < \epsilon_{c{|rhide}}$ that $\epsilon'_{c{|match}} < \epsilon_{c{|match}}$ and thus the proof is complete.
\end{proof}

\textbf{Revising the construction of the PSP for one optimal alignments.}
Algorithm~\ref{alg:RevConstPSP} shows the modified procedure to construct a PSP containing one deterministic optimal alignment for a given trace $c$ which differs from Alg.~\ref{alg:constPSP} by using the deterministic selection criteria explained above (line 10), and terminating when the entire trace has been read and the final state in $\reachGraph$ has been reached (line 13).

\begin{algorithm}[!h]{
 \SetKwInOut{Input}{input}
 \SetKwProg{Fn}{Function}{}{}
 \Input{Event Log $\logL$, DAFSA $\dafsa$, Reachability Graph $\reachGraph$}

 \For{$c \in \logL$}{
  $f \leftarrow $ align$(c, \dafsa, \reachGraph)$\;
   \textit{InsertIntoPSP}$(f,c,\psp)$\;
 }
 \Return{$\psp$}\;
 \Fn{align$(c, \dafsa, \reachGraph$)}{\label{line:begin align}
  $\open \leftarrow \{(\initState_{\psp}, \costF(\initState_{\psp},c))\}$\;
  $\maxcost \leftarrow \mid c \mid + $ minModelSkips\;
  \While{$\open \neq \varnothing$}{
   $\optimal \leftarrow \{((n_{\dafsa}, m, \epsilon), \rho) \in \open$, such that $\nexists (n_{\psp}', \rho') \in \open : \rho > \rho'\}$\;\label{line:beginWhile}
   choose a tuple $\actN = ((n_{\dafsa}, m, \epsilon), \rho)  \in \optimal$ with the following priorities : $\op(\epsilon(\left|\epsilon\right|)) : \match > \rhide > \lhide$ and choosing $\lbl{\epsilon(\left|\epsilon\right|)}$ in lexicographical order\;\label{line:ChooseNode}
   $\open \leftarrow \open \setminus \{(\actN, \rho)\}$\;
   \uIf{$n_{\dafsa} \in \finalStates_{\dafsa} \land m \in \finalStates_{\reachGraph} \land \epsilon\trPr{\dafsa} = c$}{ \label{line:condOptimal}
    \Return{$\actN = (r_\dafsa,r_\reachGraph,\epsilon_c)$}\;
   }
   \Else{
    $\potN \leftarrow \varnothing$\;
    \For{$\out_{\dafsa} = (n_{\dafsa},l_{\dafsa},n_t) \in \outT{n_{\dafsa}} \mid l_{\dafsa} = c(\left| \epsilon\trPr{\dafsa}\right| + 1)$}
    {
     $\potN \leftarrow \potN \cup \{(n_t,m,\epsilon \tieconcat(\lhide,\out_{\dafsa},\perp))\}$\;\label{line:addLhide}
     \For{$\out_{\reachGraph} = (m, l_{\reachGraph},m_t) \in \outT{m} \mid l_{\reachGraph} = l_{\dafsa}$}{
      $\potN \leftarrow \potN \cup \{(n_t, m_t,\epsilon \tieconcat (\match, \out_{\dafsa}, \out_{\reachGraph}))\}$\label{line:addMatch}
     }
    }
    \lFor{$\out_{\reachGraph} = (m,l_{\reachGraph},m_t) \in \outT{m}$}{
     $\potN \leftarrow \potN \cup \{n_{\dafsa},m_t,\epsilon \tieconcat (\rhide, \perp, \out_{\reachGraph}))\}$\label{line:addRhide}
    }
    $\open \leftarrow \open \cup \{(\nextN,\costF(\nextN,c)) \mid \nextN \in \potN \land \costF(\nextN,c) \leq \maxcost \}$\;
   }
  }\label{line:endWhile}
  }\label{line:end align}

 \caption{Revised for one-optimal: Construct the PSP}\label{alg:RevConstPSP}
}
\end{algorithm}

\textbf{Alignments are proper and minimal.}
Note that the ``final'' node $(r_\dafsa, r_\reachGraph, \epsilon_c)$ returned in line 13 defines a sequence $\epsilon_c$ of synchronizations. Next, we show that $\epsilon_c$ is indeed a proper and optimal alignment of $c$ to $\reachGraph$. Let $\alignFunc(c,\psp) = \epsilon_c$ be a function that ``extracts'' $\epsilon_c$ out of the constructed PSP $\psp$ returned by Alg.~\ref{alg:RevConstPSP}.

\begin{lemma}\label{lem:global_alignment_is_alignment}
Let $\logL$, $\dafsa$ and $\reachGraph$ be an event log, a DAFSA and a reachability graph, respectively. For each trace $c \in \logL$ and $\psp = Alg4(\logL,\dafsa,\reachGraph)$, it holds that $\epsilon_c = \alignFunc(c,\psp)$ is a proper alignment of $c$ to $\reachGraph$, i.e. $\epsilon_c \in \configs(c,\reachGraph)$.
\end{lemma}

\begin{proof}[Proof (Sketch)]
In order to prove that $\epsilon_c$ is a proper alignment, we proceed to show that it fulfils the two properties in Def.~\ref{definition:alignment}.

(1) The projection on the DAFSA reflects the trace $\lambda(\epsilon_c\trPr{\dafsa})=c$. Recall that the projection of any proper alignment onto $\dafsa$ contains only $\mathit{match}$ or $\mathit{lhide}$ operations. Alg.\ref{alg:RevConstPSP} starts at the initial state of the DAFSA for every given trace, iterates over the trace (\ref{line:beginWhile}-\ref{line:endWhile}) and adds $\lhide$-operations (line \ref{line:addLhide}) and $\match$-operations (line \ref{line:addMatch}) for outgoing arcs with the next label of the trace. Every alignment $\epsilon_c$ returned by Alg.~\ref{alg:RevConstPSP} then fulfils this property by construction as it needs to fulfil the condition $\epsilon\trPr{\dafsa} = c$ 
in line \ref{line:condOptimal} 
for determining if a given alignment is final.
\newpage
(2) $\epsilon_c\trPr{\reachGraph}$ is a path form $m_0$ to a final marking $m_f \in M_f$. Recall that the projection of any proper alignment onto $\reachGraph$ contains only $\mathit{match}$ or $\mathit{rhide}$ operations.
The algorithm always starts to add arcs from the initial marking of the reachability graph. At every iteration of the main loop (\ref{line:beginWhile}-\ref{line:endWhile}) it either adds arcs with $\match$ operations in line \ref{line:addMatch} or with $\rhide$ operations in line \ref{line:addRhide} from the set of outgoing arcs of the current marking in the reachability graph. The algorithm then adds a new node to the queue that contains the target of the added arc. By lines 18 and 20, subsequent arcs are only added if they are outgoing arcs of the node $m$ reached in $\reachGraph$, and thus will always form a path in $\reachGraph$. This path will always start from the initial marking and end in a final marking as per the condition in line \ref{line:condOptimal} and thus it is a path through the reachability graph.
\end{proof}

\begin{lemma}\label{lem:global_alignment_is_optimal}
Let $\logL$, $\dafsa$ and $\reachGraph$ be an Event log, a DAFSA and a Reachability Graph, respectively. Then it holds for each trace $c \in \logL$ and $\psp = Alg4(\logL,\dafsa,\reachGraph)$ that the alignment $\epsilon_c = \alignFunc(c,\psp)$ is minimal w.r.t the cost function $\curCost(\epsilon_c)$, i.e. $\nexists\epsilon' \in \configs(c,\reachGraph) : \curCost(\epsilon') < \curCost(\epsilon_c)$.
\end{lemma}

\begin{proof}[Proof (Sketch)]
Algorithm \ref{alg:RevConstPSP} finds alignment $\epsilon_c$ inside the while loop in function align (\ref{line:begin align}-\ref{line:end align}). Potential alignments are inserted into a queue in lines \ref{line:addMatch}, \ref{line:addLhide} and \ref{line:addRhide}. In line \ref{line:ChooseNode}, a candidate alignment is chosen from the queue with a minimal cost function value with respect to $\rho$. In each iteration of the while loop, the active candidate alignment is checked for being final in line \ref{line:condOptimal}. Once a candidate alignment $\epsilon_c$ is found final, it is returned by the function. Since all candidate alignments $\epsilon$ in the queue are selected and then removed according to their cost function value $\rho(\epsilon)$ in increasing order, the first alignment that is a proper alignment for trace $c$ will have a minimal value for $\rho(\epsilon_c)$. If $\futCost(x)=0$ would hold, then the candidate alignment would always be picked according to the cost function $\curCost$ and trivially the first final alignment would also be optimal, since all alignments with smaller costs had been investigated.
\end{proof}

Observe that for all final states $f \in R_{\psp}$, $\futCost(f)=0$, since every final state in the PSP represents a proper alignment and a proper alignment fully represents the trace, i.e. $\futL = \emptyset$, and its projection on the reachability graph represents a path, i.e. $\futM = \emptyset$. It follows that $\epsilon_c$ is optimal w.r.t. $\rho$, when function $\futCost$ underestimates the cost to the optimal cost for any investigated node, which is in line with the optimality criterion of the $A^*$-search algorithm \cite{hart68}.

We show that our definition of function $\futCost$ fulfils this criterion by analyzing how it estimates future hides for any given node. Let node $x$ be a candidate node, function $\futCost$ compares the multiset of future log labels, determined by trace $c$ set minus the already aligned trace labels $\epsilon(x)\trPr{\dafsa}$, with every possible multiset of future model labels to all possible final markings.
The multisets of future task labels represent possible paths in the reachability graph to a final marking and a path to a final node in the DAFSA representing the suffix of trace $c$. By comparing multisets to find deviations, the context of task labels is dropped and $\futCost$ allows for a lower cost than $\curCost$. Repeated task labels are also assumed to be matched in these multisets and thus are not taken into account in the comparison. Finally, function $\futCost$ minimizes the difference of all multiset comparisons such that it always finds the closest final marking in terms of distance. Givent that the multisets represent possible paths, the value of $\futCost$ can only be as high as the true cost of a path and will underestimate the cost in case the abstractions obscure differences due to context or cyclic structures.
Thus, $\futCost$  underestimates the true cost to the closest final marking and thus the alignment $\epsilon_c$ is minimal with respect to $\rho$.

\vspace{-\baselineskip}
%
%
\section{Taming concurrency with S-Components}\label{sec:s-components}
%
%
Process models with significant amounts of parallelism can lead to exponentially large reachability graphs,  given that they need to represent all different interleavings. Large reachability graphs can negatively affect the performance of the automata-based technique presented in the previous section. More generally, the combinatorial state explosion inherent to models with paralellism has a direct impact on the size of the space of possible trace alignments.
To prevent this combinatorial explosion, this section presents a novel (approximate) divide-and-conquer approach based on the decomposition of the model into concurrency-free sub-models, known as S-Components. This approach improves the execution time for models with concurrency, at the expense of allowing for over-approximations compared to an optimal alignment, which as we will show later, are infrequent and minimal in practical scenarios.
The divide-and-conquer approach of this section is an alternative to the exact PSP-based approach presented in Sect.~\ref{sec:approach}.
Figure~\ref{fig:SCompApproach} outlines the proposed divide-and-conquer approach consisting of the following steps:
\begin{inparaenum}[(1)]
    \item divide the process model into S-Components,
    \item\label{step-2} expand each S-Component to its reachability graph,
    \item project the alphabet of each S-Component on the event log to derive sub logs,
    \item compress each sub-log into a DAFSA,
    \item compare the reachability graphs (see Step \ref{step-2}) and the corresponding DAFSAs to derive sub-PSPs, and
    \item recompose the related results into alignments that are
    \item proper and
    \item approximate, but empirically in most cases optimal.
 \end{inparaenum}

\begin{figure}[hbtp]
\centering
\includegraphics[width=\textwidth]{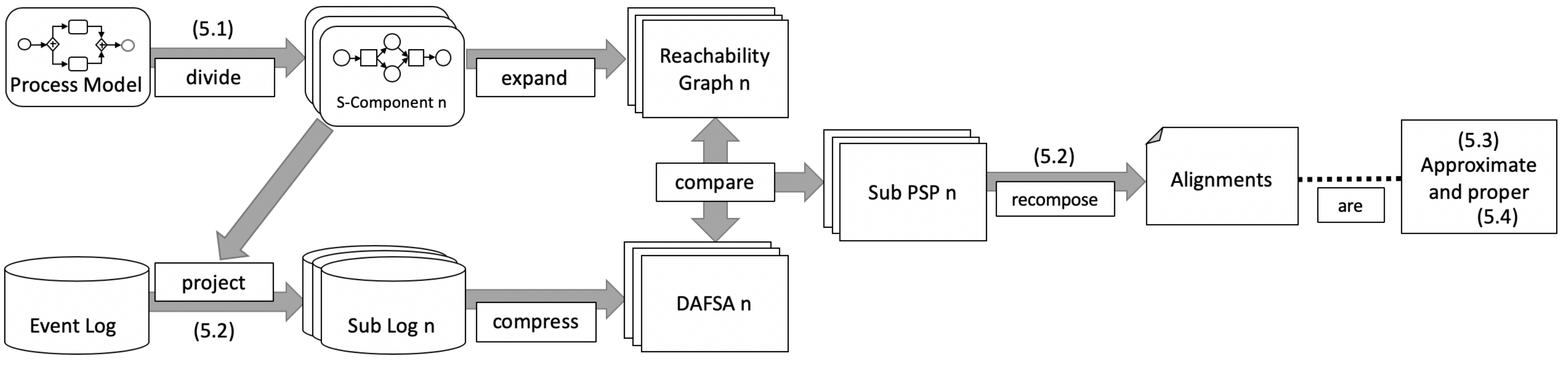}
\vspace{-1.5\baselineskip}
 \caption{\label{fig:SCompApproach}Overview of the S-Component approach.}
\end{figure}

Several steps of this approach have been already introduced in the previous section, such as the expansion of a process model to its reachability graph, the compression of a log into its DAFSA and the comparison of a reachability graph to a DAFSA to derive a PSP. In the following, we will introduce the remaining steps: namely, decomposing a process model into S-Components (Sect.~\ref{sub:scomps}), projection of logs onto S-Components and recomposing the results of conformance checking for each S-component into an alignment (Sect.~\ref{sec:s-component:recompose}). In Sect.~\ref{sec:s-components:properties}, we discuss when the alignments are approximate, and prove in Sect.~\ref{sec:s-components_invisible_label_conflicts} that they are proper.

%
%
\subsection{Dividing a process model into S-Components}\label{sub:scomps}
%
%

The decomposition approach considers uniquely-labelled sound free-choice workflow nets, a subclass of workflow nets \cite{WFNets,FreeChoice}. A workflow net  is uniquely labelled if every non-silent label is assigned to at most one transition. Soundness was defined in Sect.~\ref{sub:reachGraph}. A net is \emph{free-choice} iff whenever two transitions $t_1$ and $t_2$ share a common pre-place $s$, then $s$ is their only pre-place; in a free-choice net concurrency and choices are clearly separated. The formal definitions are given below.

\begin{definition}[Uniquely-labelled sound free-choice workflow net]
A labelled workflow net $\wNet=((P,T,F,\lambda),i,o)$ is \emph{free-choice} iff for any two transitions $t_1,t_2 \in T$: $s \in \inTr{t_1} \cap \inTr{t_2}$ implies $\inTr{t_1} = \inTr{t_2} = \{s\}$.  A workflow net is \emph{uniquely-labelled}, iff for any $t_1,t_2\in T_{\lnet},\lambda(t_1)=\lambda(t_2)\neq\tau\implies t_1=t_2$. A system net is uniquely-labelled, sound, and free-choice if the underlying workflow net is.
\end{definition}

An S-Component \cite{WFNets,FreeChoice} of a workflow net is a substructure, where every transition has one incoming and one outgoing arc (it does not contain parallelism). A sound free-choice workflow net is covered by S-Components and every place, arc and transition of the workflow net is contained in at least one S-Component, which is also a workflow net. Figure~\ref{fig:SCompExample} shows 4 different S-components of the running example of Fig.~\ref{fig:runningExamplePetriNet}. Each S-Component contains one of the four tasks $A$, $B$, $C$ or $D$ that can be executed in parallel. Note that S-components can overlap on non-concurrent parts of the workflow net as indicated by nodes with solid borders.

\begin{figure}[hbtp]
\centering
\includegraphics[width=\textwidth]{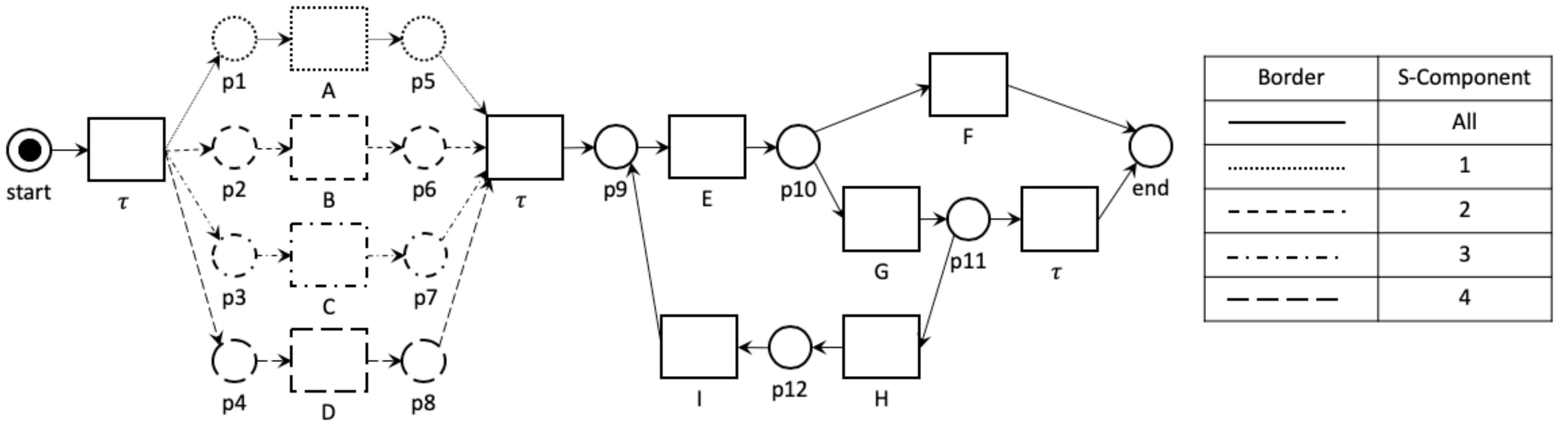}
 \caption{\label{fig:SCompExample}S-Component decomposition of the example loan application process model.}
\end{figure}
\newpage
Before we explain the decomposition of a workflow net into S-Components, we need to introduce the concept of the incidence matrix of a Petri-net. 
Recall from Sect.~\ref{sub:reachGraph} that a marking $m = \langle m(p_1),\ldots,m(p_k)\rangle^\intercal$ is a column vector over the places $P = \{p_1,\ldots,p_k\}$; and vectors $N^-(t)$ and $N^+(t)$ describe the tokens consumed and produced by $t$ on each $p \in P$. The resulting \emph{effect} of $t$ on $P$ is $N(t) = N^-(t) + N^+(t)$.
The \emph{incidence matrix} of a Petri net $N$ is the matrix $N = \langle N(t_1) \ldots N(t_r) \rangle $ of the effects of all transitions $T = \{t_1,\ldots,t_r\}$. Given a firing sequence $\sigma$ in $\wNet$ starting in $m_0$, let the row vector $y = \langle y_1,\ldots,y_r \rangle$ specify how often each $t_i, i=1,\ldots,r$ occurred in $\sigma$. For any such row vector, the \emph{marking equation} $m =  m_0 + N \cdot y$ yields the marking reached by firing $\sigma$. Figure \ref{fig:MEquation} shows how the marking equation of the Petri net of our sample loan application process in Fig.~\ref{fig:runningExamplePetriNet} gives a new marking from the initial marking.

\begin{figure}[hbtp]
\centering
\includegraphics[width=\textwidth]{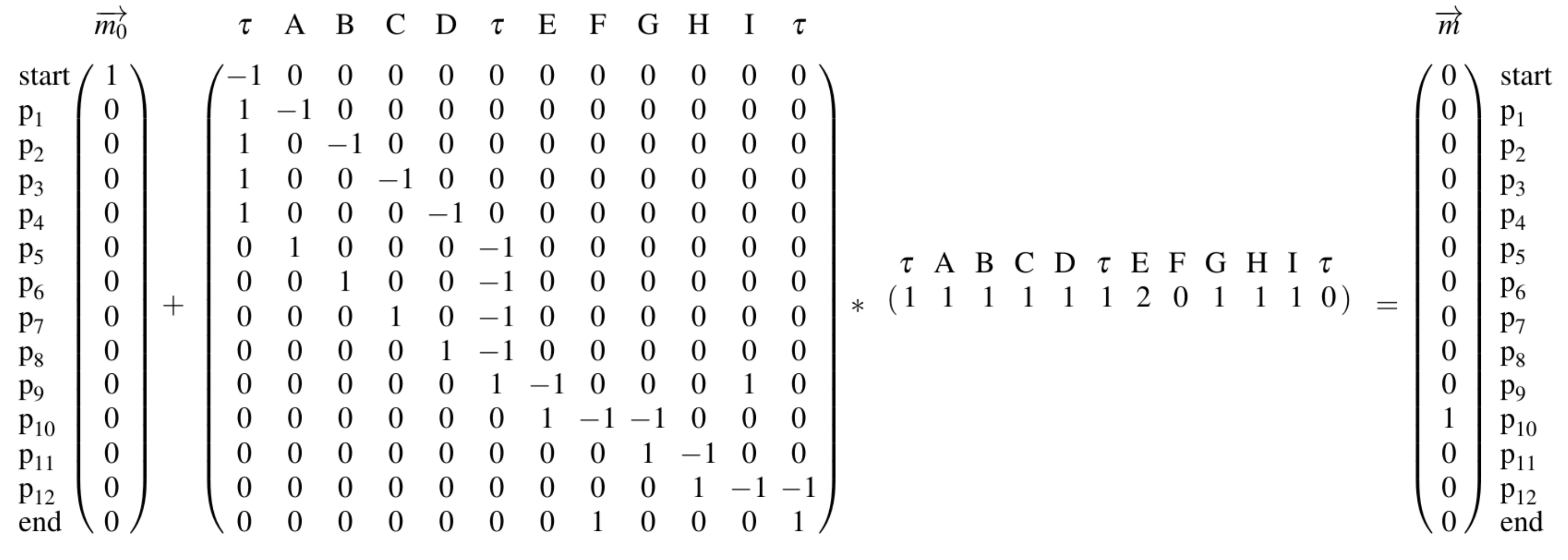}
 \vspace{-1.5\baselineskip}
 \caption{\label{fig:MEquation}Marking equation to reach marking $(p10)$ for our loan application example.}
\end{figure}

The decomposition of a sound free-choice Petri net into S-Components is based on its place invariants. A \emph{place invariant} is an integer vector indicating the number of tokens that are constant over all reachable markings. It can be determined as an integer solution $J$ to the marking equation $J \cdot N = 0$
, i.e., $J \cdot m_0 = J \cdot m$ for all reachable markings $m$ of $N$, because $J \cdot N \cdot y = 0$~\cite{FreeChoice}. The equation $J \cdot N = 0$ has an infinite number of solutions.
Place invariants can be \emph{non-trivial}, in the following denoted as $\minPI$, they are different from 0 and are \emph{minimal} (not linear combinations of other place invariants of $N$).
We are only interested in the unique set of \emph{non-trivial} place invariants $\minPI$, which can be obtained through standard linear-algebra techniques. Each minimal place invariant $J$ possibly defines an S-Component as a subnet of the workflow net consisting of the \emph{support} $\langle J \rangle$ of $J$~\cite{FreeChoice}. The workflow net can be decomposed into $n$ S-Component subnets, where $n$ is the number of minimal place invariants of the workflow net, i.e. $\left|\minPI\right|$. We next define a S-Component net and the decomposition of a workflow net.

\begin{definition}[S-Component, S-Component decomposition]
Let $\wNet=((\places_N, \netTransitions_N, \netArcs_N, \netLabel_N),i,o)$ be a sound, free-choice workflow net. Let $J$ be a minimal place invariant of $\wNet$. An \emph{S-Component} $\wNet_J$ is a non-empty, strongly connected labelled workflow net $\wNet_J=((\places_J, \netTransitions_J, \netArcs_J, \netLabel_J),i,o)$ with the following properties:
    \begin{itemize}
    \item $\places_J = \{p \in \places_N \mid p \in \left<J\right> \land \inTr{p} \subseteq \netTransitions_J \land \outTr{p} \subseteq \netTransitions_J \}$
    \item $\netTransitions_J = \{t \in \netTransitions_N \mid \left|\inTr{t} \cup \places_J\right| = 1 = \left|\outTr{t} \cup \ \places_J\right|\}$
    \item $\netArcs_J = \{(p,t) \in \netArcs_N \mid p \in \places_J \land t \in \netTransitions_J\} \cup \{(t,p) \in \netArcs_N \mid t \in \netTransitions_J \land p \in \places_J\}$
    \item $\netLabel_J = \{(t,l) \in \netLabel_N \mid t \in \netTransitions_J\}$
    \end{itemize}
For the set of all minimal place invariants $\minPI$ of $\wNet$, the \emph{S-Component decomposition} $\scomp$ is a non-empty set of S-Component workflow nets that cover $\wNet$, i.e. $\scomp=\{\wNet_J \mid J \in \minPI\}$.
\end{definition}

S-Components are concurrency-free, as the requirement $\left|\inTr{t} \cup \places_J\right| = 1 = \left|\outTr{t} \cup \ \places_J\right|$ allows only one input / output place per transition.
Applying the decomposition to our running example, four minimal place invariants are computed: (1 1 0 0 0 1 0 0 0 1 1 1 1 1), (1 0 1 0 0 0 1 0 0 1 1 1 1 1), (1 0 0 1 0 0 0 1 0 1 1 1 1 1) and (1 0 0 0 1 0 0 0 1 1 1 1 1 1). Figure \ref{fig:SCompExample} shows the derived four S-Component workflow nets; each S-Component contains one of the four tasks $A$, $B$, $C$ or $D$ that can be executed in parallel.

%
%
\input{approach_new_scomp_recompose.tex}
\subsection{Optimality is not guaranteed under recomposition}\label{sec:s-components:properties}
%
%

The recomposition of partial alignments in Alg.~\ref{alg:constPSPwithSComponents} is not necessarily optimal. Figure \ref{fig:CounterExample} shows a pair of S-Components, each representing a parallel activity $A$ or $B$ followed by a merging activity $C$ and a trace $\langle C,A,B\rangle$, where the merging activity is miss-allocated before the parallel activities. The two optimal projected alignments according to the sorting from subsection \ref{sub:deterministicAlignments} then each include a $\rhide$ synchronization for the parallel activity, a $\match$ synchronization for the merging activity $C$ and a $\lhide$ synchronization for the parallel activity. Note that both projected alignments are optimal in cost. Once the projected alignments are recomposed, the cost of the recomposed alignment is 4: $\langle r(A),r(B),m(C),l(A),l(B)\rangle$. However, there exists another proper alignment with a lower cost of 2: $\langle l(C), m(A),m(B),r(C)\rangle$.
The reason why the recomposed alignment is not optimal, while the projected alignments are optimal, is that the projected alignments choose one optimal alignments out of multiple possible optimal alignments with the same cost without considering which choices would globally minimize the cost when recomposing the projected alignments. In this example, the projected alignments with another kind of sorting could also be $\langle l(C),m(A),r(C)\rangle$ and $\langle l(C),m(B),r(C)\rangle$, which would recompose to the optimal alignment. 

\begin{figure}[hbtp]
\centering
\includegraphics[width=\textwidth]{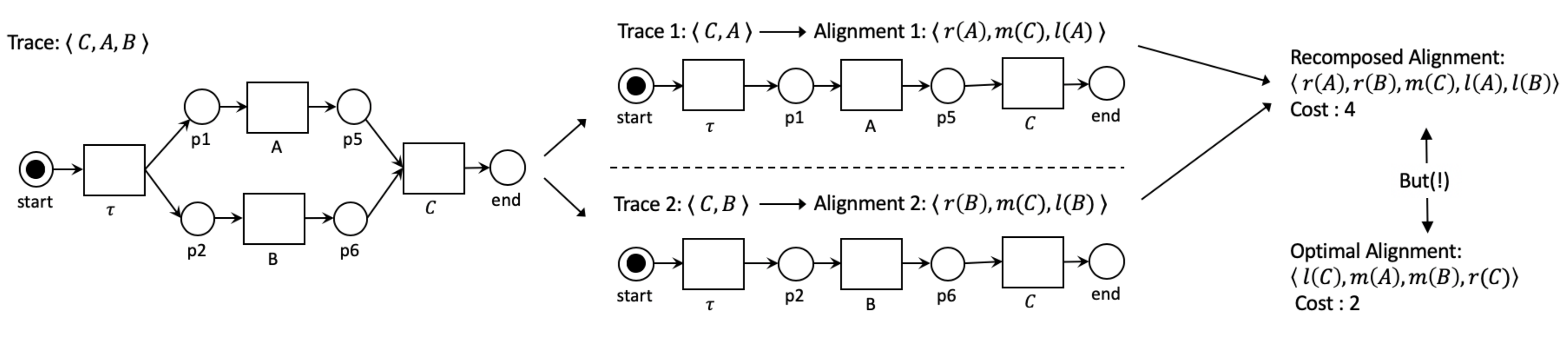}
\vspace{-2\baselineskip}
 \caption{\label{fig:CounterExample}Counter-example to optimality of a recomposed alignment.}
\end{figure}

With the current sorting introduced in subsection \ref{sub:deterministicAlignments}, we introduce an additional cost of one over the optimal cost per S-Component workflow net  for a task miss-allocation of a merging activity possibly multiple times, when the parallel block is enclosed in a cyclic structure.
Hence, the worst-case cost over-approximation of the proposed recomposition algorithm for a given trace $c$ is $k*\rep$, where $k$ is the size of the S-Component decomposition and $\rep$ is the number of maximal repetitions of a label in $c$ that is also contained in a parallel block in the process model.
Transforming the recomposition procedure into a minimization problem of selecting the best projected alignments for recomposition would however increase the calculation overhead exponentially since every trace can have exponentially many optimal alignments for each S-Component workflow net. Thus, selecting the best optimal projected alignments can be computationally more expensive than calculating only one-optimal alignments for the initial workflow net and event log. However, calculating the reachability graphs of workflow nets without parallel constructs is polynomial in size, speeding up the calculation of one optimal-projected alignments, and thus the proposed technique can provide significant speed-ups over the original technique on process models with parallelism.

Even though the presented approach computes non-optimal results, the evaluation shows that both the fraction of affected traces as well as the degree of over-approximation is rather low. The results obtained for the evaluation of this novel approach is oftentimes close to optimal.

%
%
\subsection{Proper alignments by addressing invisible label conflicts}\label{sec:s-components_invisible_label_conflicts}
%
%

The recomposition of synchronizations from the partial alignments of the S-components in Alg.~\ref{alg:constPSPwithSComponents} relies on the unique labeling. In this way, arcs in the reachability graphs of different S-components can safely be related to each other. However, if a uniquely labeled process model contains a $\tau$-labeled transition, Alg.~\ref{alg:remTau} reduces these $\tau$-labeled transitions by contraction with subsequent visible edges. This may lead to two arcs in the reachability graph carrying the same label $D$ but describing different effects, a hidden form of label duplication. Applying Alg.~\ref{alg:constPSPwithSComponents} on such a model may lead to two partial alignments where the composed synchronization agree on label $D$, but the underlying arcs in the reachability graphs disagree, leading to a ``hidden'' recomposition conflict not detected by Alg.~\ref{alg:constPSPwithSComponents}. The resulting $\epsilon_c$ would no longer form a path through the process model.

In the following, we illustrate the problem by an example and discuss a simple change to Alg.~\ref{alg:remTau} that ensures a unique labeling over all reachability graphs (global and projected). For such reachability graphs, Alg.~\ref{alg:constPSPwithSComponents} always returns an alignment, which we prove formally.
Figure \ref{fig:ConflictRecomposition} shows an example with trace $\langle A,B,D \rangle$ and a process model, where the parallel tasks $B$ and $C$ can be skipped. The process model is decomposed into two S-Component nets, one for each of the two parallel activities. When the trace is projected onto the S-Component with activity $C$, the obtained alignment matches both trace activities and skips activity $C$ with the $\tau$ transition. The sub-trace $\langle A, B, D \rangle$ can be fully matched on the other S-Component. The recomposed alignment is $\langle m(A),m(B),m(D) \rangle$.
However, $A,B,D$ is not a path through the reachability graph of this process model.

\begin{figure}[hbtp]
\centering
\vspace{-.5\baselineskip}
\includegraphics[width=\textwidth]{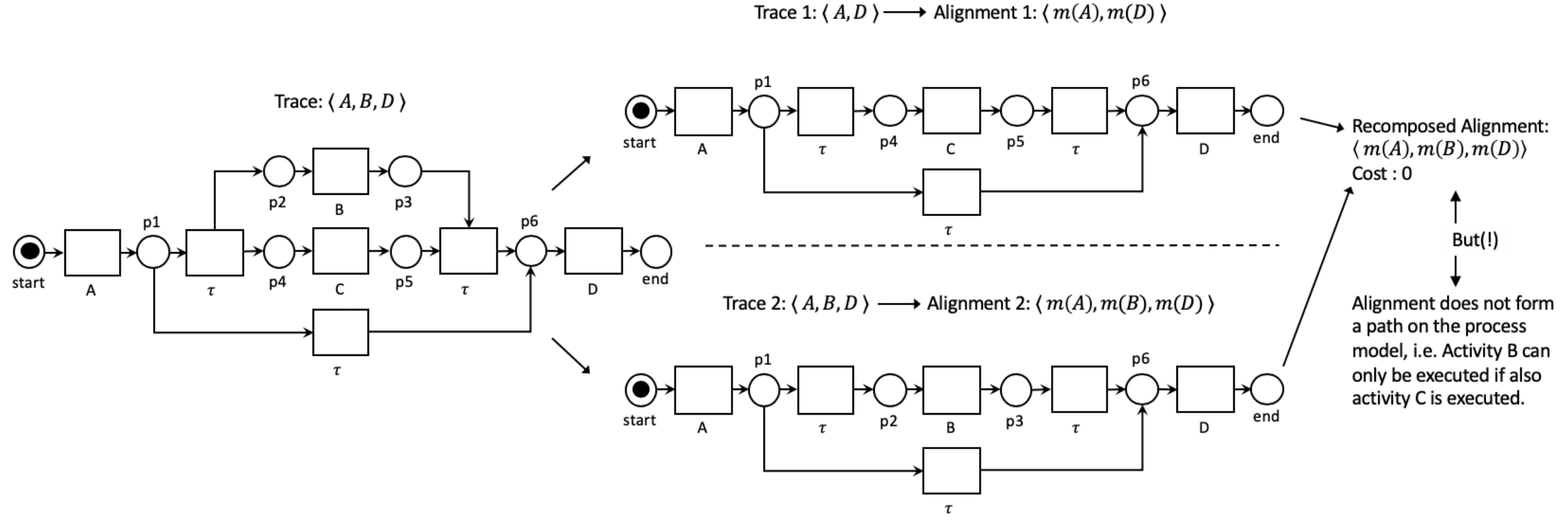}
\vspace{-1.5\baselineskip}
 \caption{\label{fig:ConflictRecomposition}Recomposed alignment that can not be replayed on the process model.}
\end{figure}

Note that reducing the reachability graph of the model in Fig.~\ref{fig:ConflictRecomposition} by Alg.~\ref{alg:remTau} leads to two $D$-labeled arcs: $([p1],D,[end])$ (by the skipping $\tau$-transition) and $([p3,p5],D,[end])$ (by the joining $\tau$-transition). The alignment for the first S-component uses the former whereas the alignment for the second S-component uses the latter, leading to the conflict described above.

\begin{figure}[hbtp]
\centering
\includegraphics[width=\textwidth]{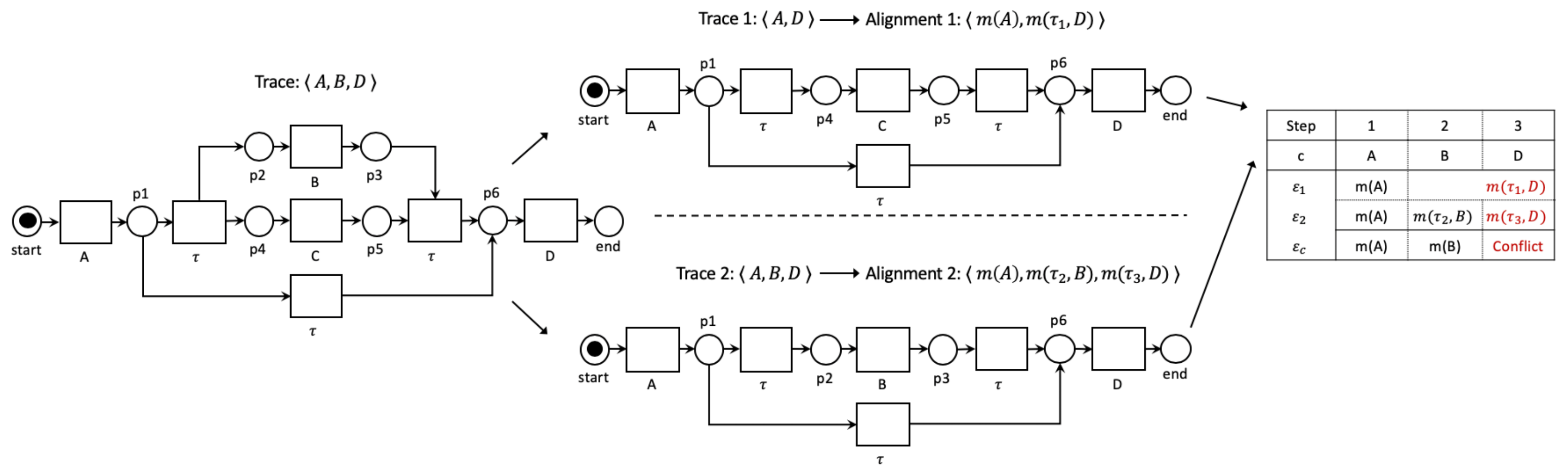}
\vspace{-1.5\baselineskip}
 \caption{\label{fig:ResolvingConflictRecomposition}Detecting the invisible label conflict.}
\end{figure}

Figure \ref{fig:ResolvingConflictRecomposition} illustrates shows how to relabel arcs in the reachability graph to avoid ``hidden'' label duplication. First, add to each $\tau$-transition a unique index at the start of Alg.~\ref{alg:remTau} so that all $\tau$-transitions are uniquely labeled. Second, we alter Alg.~\ref{alg:remTau} to maintain the identity of the removed $\tau$ transitions in the next visible transition . In particular, when replacing an arc $(n_1,\ell,n_2)$ for an arc with label $\tau_i$, we create an extended label $(\tau_i,\ell)$ for the replacement arc. Let Alg.~\ref{alg:remTau}$^*$ be this modification of Alg.~\ref{alg:remTau} and let Alg.~\ref{alg:constPSPwithSComponents}$^*$ which invokes Alg.~\ref{alg:remTau}$^*$ instead of Alg.~\ref{alg:remTau}. Alg.~\ref{alg:remTau}$^*$ and the extended labels are not used for the PSP construction Algs.~\ref{alg:constPSP} and \ref{alg:RevConstPSP}. We omit the technical details. The changes in Fig. \ref{fig:ResolvingConflictRecomposition} lead to the following differences: The transition $\tau_1$ can now be distinguished from transition $\tau_3$. During the recomposition, there will be a label conflict between the extended labels $(\tau_1,D)$ and $(\tau_3,D)$.
Since the recomposition can no longer find a viable path through the process model for these conflicting cases, these traces need to be aligned on the reachability graph of the input process model with Alg.~\ref{alg:RevConstPSP} in Section~\ref{sec:approach}. This ensures that alignments of traces with conflictsfound  during the recomposition also form a path through the process model according to Lemma~\ref{lem:global_alignment_is_alignment}.
While the quality of the results is maintained for these conflicting cases, this implies that calculation speed is lost since the sub-alignments are calculated for each S-Component and these results are invalid. Thus, this divide-and-conquer approach is specially faster when the number of conflicting cases during the recomposition is low.
Please note, however, that only cases with conflicts need to go though this procedure. Recomposed alignments without conflict do not need to be re-computed with Alg.~\ref{alg:RevConstPSP}.

\vspace{.5\baselineskip}
Next, we will investigate whether recomposed alignments without conflicts are proper, i.e. all trace labels are contained and forms a path through the reachability graph of the input process model.
The following theorem states that the recomposed alignment is a proper alignment.
\begin{theorem}\label{thm:recompose_alignment_optimal}
Let $\logL$ be an event log and $\sysNet = (\wNet, \initMarking, \finalMarkings)$ be a system net, where $\wNet$ is a uniquely labeled, sound, free-choice workflow net. Let $\psp' = Alg5^*(\logL,\dafsa,\sysNet)$. For every trace $c \in \logL$, $\epsilon_c = \epsilon(\psp',c)$ is a proper alignment. Specifically, the following two properties hold:
    \begin{compactenum}
        \item The sequence of synchronizations with $\lhide$ or $\match$ operations reflects the trace $c$, i.e. $\lambda(\epsilon_c \trPr{\dafsa}) = c$.
        \item The arcs of the reachability graph in the sequence of synchronizations with $\rhide$ or $\match$ operations forms a path in the reachability graph from the initial to the final marking, i.e. let $n=\left|\epsilon_c\trPr{\reachGraph}\right|$, then $\src{\epsilon_c \trPr{\reachGraph}[1]} = \initMarking \land \tgt{\epsilon_c \trPr{\reachGraph}(n)} = \finalMarkings \land \forall a_i, a_{i+1} \in \epsilon\trPr{\reachGraph} \mid 1\leq i < n : \tgt{a_i} = \src{a_{i+1}}$.
    \end{compactenum}

\end{theorem}
The formal proof by induction on the length of $c$ is given in~\ref{app:alg5_correct}. The core argument is to show that the markings and the transition firings of $\sysNet$ can be reconstructed from the vector $\underline{m}$ of markings of each S-component nets. Further, the arcs in the reachability graphs of the S-components nets are isomorphic to the transitions. As a result, the transition effect of the original transition in $\sysNet$ can be recomposed from the effects in the S-component nets. The latter argument requires the uniqueness of arcs in the reachability graphs provided by Alg.~\ref{alg:remTau}$^*$.
\vspace{.5\baselineskip}

%% file: approach_new_scomp_recompose.tex
\subsection{Conformance checking with S-Component decomposition}\label{sec:s-component:recompose}

This section introduces a novel divide-and-conquer approach to speed up the conformance checking between a system net and an event log. The division of the problem relies on the decomposition of the workflow net into S-Component workflow nets as introduced in Section~\ref{sub:scomps}.

The following definition introduces \emph{trace projection}, an operation that filters out the events with labels not contained in the alphabet of a particular S-Component.

\begin{definition}[Trace Projection]
A trace projection, denoted as $\trPr{\netLabel}$, is an operation over a trace $c = \langle l_1, l_2, ..., l_n \rangle$ that filters out all the labels not contained in $\netLabel$, i.e. $c \trPr{\{l_{i}, l_{j}, \dots, l_{k}\}} = \langle l_{i}, l_{j}, \dots, l_{k} \rangle$ such that $0 \leq i \leq j \leq, \dots, \leq k \leq n$.
\end{definition}

The novel divide-and-conquer approach decomposes the workflow nets into concurrency-free sub-workflow nets -- S-Components --, computing partial alignments between projected traces and S-components, and recomposing the partial alignments to create alignments for each trace in the log. Note that the alignments are partial because the projected traces are only parts of a complete trace. In the following, we explain the full procedure, illustrated in Fig.~\ref{fig:ExampleRecomposition} and defined in Alg.~\ref{alg:constPSPwithSComponents}, as we obtain and re-compose partial alignments for the trace $\langle B,D,A,E,F,G\rangle$ in our running example and the S-Component workflow nets (Fig.~\ref{fig:SCompExample}). Observe that in our running example there are four S-Component workflow nets, each representing the execution of one of the parallel activities $A, B, C$ and $D$.

\begin{algorithm}[!h]
{
    \SetKwInOut{Input}{input}
    \Input{Event Log $\logL$, DAFSA $\dafsa$, System net $\sysNet=(\wNet,m_0,M_f)$ with its S-Component decomposition $\scomp=\{\wNet_1,\ldots,\wNet_k\}$; }

    \tcp{compute decomposed reachability graph and DAFSA for each S-component}
    $\Sigma_L \leftarrow \textit{labels occuring in}\ \logL$\;
    $\Sigma_i \leftarrow \textit{labels occurring in } \wNet_i \textit{ for each } i=1,\ldots,k$\;
    $\reachGraph_i \leftarrow \textit{reachability graph of } \sysNet_i=(\wNet_i,m_0,M_f) \textit{ for each } i=1,\ldots,k$\;
    $\dafsa_i \leftarrow \textit{construct DAFSA}(\{ c\trPr{\Sigma_i} \mid c \in \logL\}) \textit{ for each } i=1,\ldots,k$\;

    $\underline{\Sigma} \leftarrow (\Sigma_1,\ldots,\Sigma_k)$\;

    \For{Trace $c \in \logL$}
    {
        \tcp{compute alignments locally for each S-component}
        $\underline{\epsilon} \leftarrow (\epsilon_1,\ldots,\epsilon_k)$ with $\epsilon_i \leftarrow align( c\trPr{\Sigma_i}, \dafsa_i, \reachGraph_i)$ with Alg.~4 for each $i=1,\ldots,k$\;
        
        \tcp{recompose local alignments $\underline{\epsilon}$ into global alignment}
        $\epsilon_c \leftarrow \langle\rangle$\;\label{line:beginRecompose}
        $\underline{n} \leftarrow (\nodeLog(\epsilon_1(0)),\ldots,\nodeLog(\epsilon_k(0)))$\;\label{line:initialDAFSA}
        $\underline{m} \leftarrow (\nodeMod(\epsilon_1(0)),\ldots,\nodeMod(\epsilon_k(0)))$\;\label{line:initialMarking}
        $\underline{pos} \leftarrow (pos_1,\ldots,pos_k)$ with $pos_i \leftarrow 0$ for each $i=1,\ldots,k$\;
        \For{$pos_c \leftarrow 1 : |c| + 1$}
        {
            \tcp{next log event, try to recompose one synchronization with $\ell$}
            $\ell \leftarrow c[pos_c]$\;\label{line:traceLabel}

            \tcp{but first local components may need to do model steps to reach $\ell$, try to recompose model steps}
            $(\mathit{couldRecompose}, \underline{pos}, \underline{m}, \epsilon_c) \leftarrow \mathit{recomposeModelStepsUntil}(\ell, \epsilon_c, \underline{\epsilon}, \underline{pos}, \underline{m}, \underline{\Sigma}, \Sigma)$ with Alg.~\ref{alg:constPSPwithSComponents:recompose_modelMoves}\;\label{line:recomposeModelMoves}

            \If{$\mathit{couldRecompose} = \mathit{false}$}
            {
                \tcp{Recomposition conflict: S-components locally aligned models steps in a different order, cannot recompose}
                jump to line \ref{line:conflict}\;\label{line:endSyncRhide}
            }

            \If{$(\forall i=1\ldots k \mid \ell \in \Sigma_i \implies \op(\epsilon_i(pos_i)) = \lhide) \land \ell \neq \perp$\label{line:beginTraceRelatedMove}\label{line:recLhide}}
            {
                \tcp{S-components having $\ell$ agree on log step, recompose their synchronizations}
                $(\epsilon_c, \underline{pos}, \underline{n}) \leftarrow \mathit{recomposeLogStep}(\ell, \epsilon_c, \underline{\epsilon}, \underline{pos}, \underline{n}, \underline{\Sigma})$ with Alg.~\ref{alg:constPSPwithSComponents:recompose_logMoves}\;\label{line:recomposeLogMoves}
            }
            \ElseIf{$(\forall i=1\ldots k \mid \ell \in \Sigma_i \implies \op(\epsilon_i(pos_i)) = \match) \land \ell \neq \perp$\label{line:recMatch}}
            {
                \tcp{S-components having $\ell$ agree on log and model step, recompose their synchronizations}
                $(\epsilon_c, \underline{pos}, \underline{n}, \underline{m}) \leftarrow \mathit{recomposeMatchStep}(\ell, \epsilon_c, \underline{\epsilon}, \underline{pos}, \underline{n}, \underline{m}, \underline{\Sigma})$ with Alg.~\ref{alg:constPSPwithSComponents:recompose_syncMoves}\;\label{line:recomposeSyncMoves}
            }
            \ElseIf{$\ell \neq \perp$\label{line:condLhideConflict}}
            {
                \tcp{Operation conflict: S-components locally disagree on operation, cannot recompose}
                jump to line \ref{line:conflict}\;
            }\label{line:endTraceRelatedMove}
        }
        \lIf{Any conflict occurred\label{line:beginFallback}}{$\epsilon_c \leftarrow align(c, \dafsa, \reachGraph)$ with Algorithm \ref{alg:RevConstPSP}}\label{line:conflict}
        \textit{InsertIntoPSP}$(\epsilon_c,c,\psp)$\;\label{line:endRecompose}\label{line:endFallback}
    }
    \Return{$\psp$}\;
     \caption{Construct PSP by Recomposing S-Component Alignments}\label{alg:constPSPwithSComponents}
}
\end{algorithm}
\vspace{-1.5\baselineskip}

\subsubsection*{Algorithmic idea.} Algorithm~\ref{alg:constPSPwithSComponents} is given the S-Components $\sysNet_1,\ldots,\sysNet_k$ of $\sysNet$ and the DAFSA $\dafsa$ of the complete log $\logL$ which has alphabet $\Sigma_L$. Each S-component $\sysNet_i$ has alphabet $\Sigma_i$. We first compute for each S-component $\sysNet_i$ its reachability graph and the DAFSAs $\dafsa_i$ of the projected log with alphabet $\Sigma_i$ (see Lines 2-5). From that point on, Alg.~\ref{alg:constPSPwithSComponents} computes and recomposes along the $k$ S-components and stores all information in \emph{vectors} of length $k$, i.e.,  $\underline{\Sigma} = (\Sigma_1,\ldots,\Sigma_k)$ (see Line 6).

We continue by taking each trace $c$ in the log projecting it onto the alphabet of each S-component $c \trPr{\Sigma_i}$. For our example trace $\langle B,D,A,E,F,G\rangle$ four partial traces are created: $\langle B,E,F,G\rangle$, $\langle D,E,F,G\rangle$, $\langle E,F,G\rangle$ and $\langle A,E,F,G\rangle$. The traces share the subsequence $\langle E,F,G \rangle$ as the corresponding transitions are in the sequential part of the workflow net and hence in all S-components in Fig.~\ref{fig:SCompExample}.

Then, we compute the deterministic optimal alignment $\epsilon_i$ of each projected trace $c \trPr{\Sigma_i}$ to its S-component $\sysNet_i$ (by calling Alg.~\ref{alg:RevConstPSP} in line 7 of Alg.~\ref{alg:constPSPwithSComponents}); we call each $\epsilon_i$ a \emph{projected} alignment. Figure~\ref{fig:ExampleRecomposition} shows the four optimal projected alignments $\epsilon_1$-$\epsilon_4$ retrieved by Alg.~\ref{alg:RevConstPSP} for our running example. Note that because each $\sysNet_i$ is sequential, the reachability graph of each $\sysNet_i$ has the same size as $\sysNet_i$ itself. Thus the $k$ projected alignment problems are exponentially smaller than the alignment problem on the reachability graph of the original $\sysNet$.

Once the projected alignments have been computed, we iterate over the original trace $c$ and \emph{compose} the projected alignments $\underline{\epsilon} = (\epsilon_1,\ldots,\epsilon_k)$ (between the $\dafsa_i$ and $\reachGraph_i$ of $\sysNet_i$) into a global alignment $\epsilon_c$ between $\dafsa$ and $\reachGraph$ of $\sysNet$. 
We explain the idea of this recomposition by recomposing the projected alignments $\epsilon_1$-$\epsilon_4$ of our example of Fig.~\ref{fig:ExampleRecomposition}.
The recomposition technically ``replays'' all alignments ($\epsilon_1,\epsilon_2,\epsilon_3,\epsilon_4)$ along trace $c = \langle B,D,A,E,F,G\rangle$ in parallel. For each next event $\ell$ of the log, Alg.~\ref{alg:constPSPwithSComponents} determines which alignments together can replay $\ell$ (and make a joint step in their DAFSAs and in their reachability graphs). Initially, the 4 S-components of Fig.~\ref{fig:SCompExample} are locally in their initial markings $\underline{m} = ([\mathit{start}],[\mathit{start}],[\mathit{start}],[\mathit{start}])$ (see Alg.~\ref{alg:constPSPwithSComponents}, line~\ref{line:initialMarking}). We iterate over trace $c = \langle B,D,A,E,F,G\rangle$ as shown in Fig.~\ref{fig:ExampleRecomposition}.
\begin{figure}[hbtp]
\centering
\vspace{-.5\baselineskip}
\includegraphics[width=0.5\textwidth]{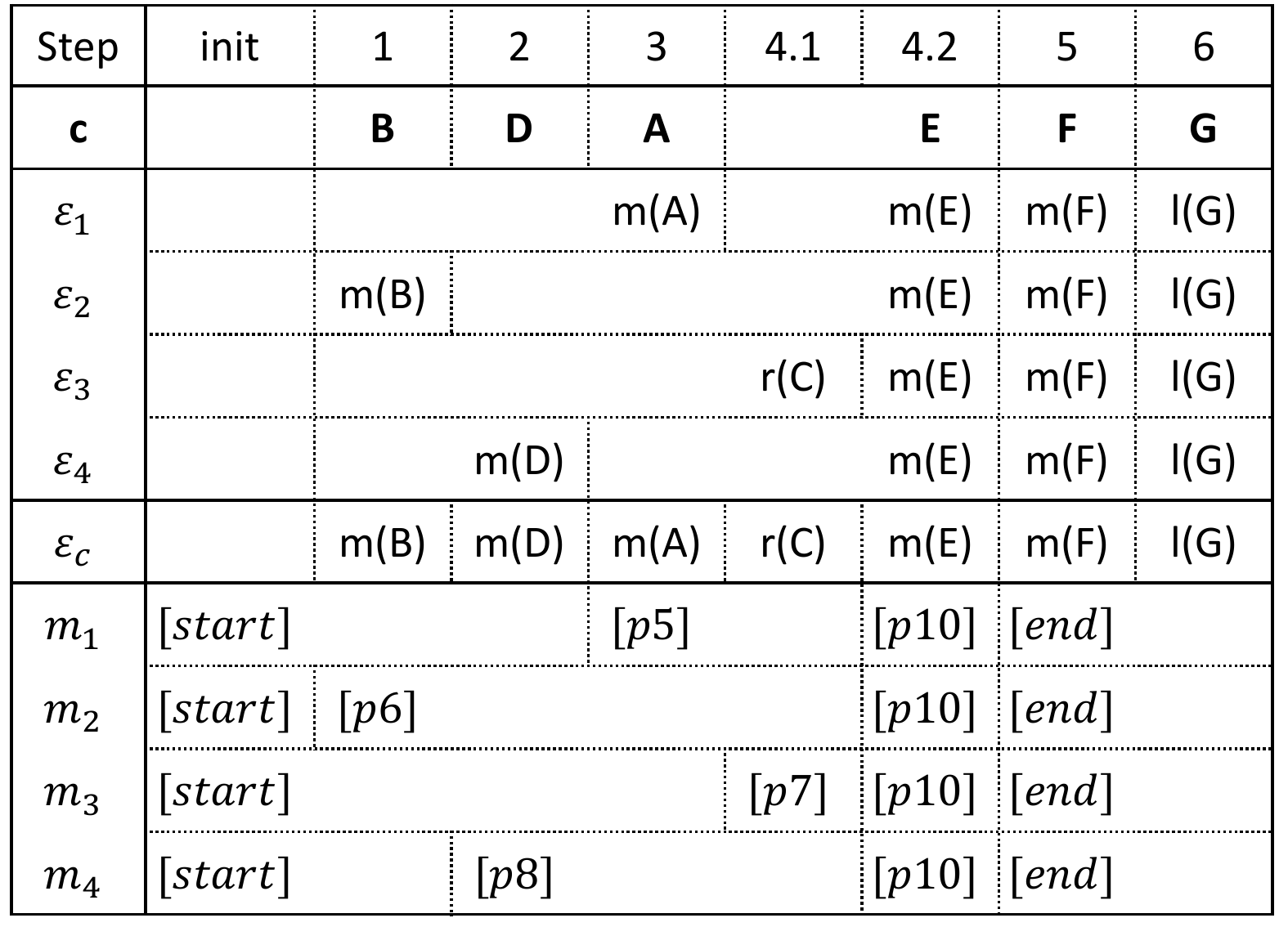}
 \caption{\label{fig:ExampleRecomposition}Example for applying Alg.~\ref{alg:constPSPwithSComponents} to trace $\langle B,D,A,E,F,G\rangle$ to our running example.}
 \vspace{-.5\baselineskip}
\end{figure}
\begin{itemize}
\item[(1)] For the first event $B$, all S-components involving $B$ (which is the single S-component 2 in this case) have as their next synchronization $\epsilon_i(pos_i)$ a match synchronization $m(B)$ for $B$. We add a match step $m(B)$ for the composed alignment $\epsilon_c$ and S-component 2 reaches $m_2 = [p6]$, while S-components 1,3, and 4 remain in $[\mathit{start}]$, respectively. Recall that $\tau$-transitions were removed from the reachability graph, allowing this single step. In Alg.~\ref{alg:constPSPwithSComponents}, the condition in line~\ref{line:recMatch} and $\mathit{recomposeMatchStep}$ line~\ref{line:recomposeSyncMoves} define this composition step explained below.
\item[(2-3)] Similarly, for the next events $D$ and $A$ of $c$, $m(D)$ and $m(A)$ synchronizations from $\epsilon_1$ and $\epsilon_2$ are added to $\epsilon_c$ by $\mathit{recomposeMatchStep}$, which corresponds to reaching $[p8]$ and $[p5]$ in S-components 1 and 4.
\item[(4.1)] The fourth event $E$ occurs in all S-components, but $\epsilon_3$ has at its current position $pos_3 = 0$ an $\rhide$ synchronization for $C$ which is only later followed by an $E$-synchronization. In Fig.~\ref{fig:SCompExample}, this corresponds to S-component 3 still being in marking $[start]$ and transition $E$ not being enabled yet. In order to reach $E$ and to ``catch up'' with all other S-components, S-component 3 can locally replay $\rhide$ synchronizations of any label that is not the current label $E$. In Alg.~\ref{alg:constPSPwithSComponents}, $\mathit{recomposeModelStepsUntil}$ line~\ref{line:recomposeModelMoves} define this composition step explained below. In our example, $r(C)$ is added to $\epsilon_c$ and S-component 3 reaches marking $[p7]$.
\item[(4.2)] Now all S-components agree on $m(E)$ synchronizations, thus $m(E)$ is added to $\epsilon_c$ and all S-components reach marking $[p10]$. Note that \emph{all} S-components having the $m(E)$ make a joint step, expressed in lines~\ref{line:recMatch}-\ref{line:recomposeSyncMoves} of Alg.~\ref{alg:constPSPwithSComponents} by an update on several components of the vector $\underline{m}$.
\item[(5)] The algorithm proceeds similarly in step 5 by adding a $\match$ synchronization for label $F$. 
\item[(6)] For label $G$, all S-components agree on an $\lhide$ synchronization which is added to $\epsilon_c$ as $l(G)$. This is the third case in the loop of Alg.~\ref{alg:constPSPwithSComponents} by the condition in line~\label{line:recLhide} and the composition due to $\mathit{recomposeLogSteps}$ in line~\ref{line:recomposeLogMoves} explained below.
\end{itemize}
The resulting sequence $\epsilon_c$ of synchronizations is a proper alignment according to Def.~\ref{definition:alignment} and is added to the PSP (line~\ref{line:endRecompose}). Alg.~\ref{alg:constPSPwithSComponents} may fail recomposing the S-components in case the projected alignments locally disagree on the next synchronization to compose, i.e., lines ~\ref{line:endSyncRhide} and \ref{line:endTraceRelatedMove}. In this case, we revert back to computing a global alignment without decomposition (line \ref{line:conflict}). Next explain the technical details of this recomposition and formalize the three recomposition cases $\mathit{recomposeModelStepsUntil}$, $\mathit{recomposeLogStep}$, and $\mathit{recomposeMatchStep}$.


\subsubsection*{Composing partial alignments by synchronizing $k$ FSMs.}

Recall that an alignment is a sequence of synchronizations. Each synchronization $\beta_i = (\mathit{op}_i,b_i,a_i)$ of a projected alignment $\epsilon_i$ refers to an arc $b_i = (n_i,\ell,n_i')$ of DAFSA $\dafsa_i$ and/or an arc $a_i = (m_i,\ell,m_i')$ of the reachability graph $\reachGraph_i$ of $\sysNet_i$; $\dafsa_i$ and $\reachGraph_i$ are FSMs. Each arc has a source and target state of the projected DAFSA/reachability graph. The task of Alg.~\ref{alg:constPSPwithSComponents} is to compose from the nodes and arcs of the projected $\dafsa_i$ and $\reachGraph_i$ along trace $c$ a \emph{sequence $\epsilon_c$ of synchronizations} of nodes and arcs of the global DAFSA $\dafsa$ and $\reachGraph$. This sequence $\epsilon_c$ has to form a path through $\dafsa$ (i.e., describe the trace $c$) and a path through $\reachGraph$ (i.e., describe a run). Then $\epsilon_c$ is an alignment of $c$ to $\sysNet$ by Def.~\ref{definition:alignment}.

Technically, we construct the nodes and arcs of $\dafsa$ and $\reachGraph$ as \emph{vectors} of the nodes and arcs of the projected $\dafsa_i$ and $\reachGraph_i$. We represent the \emph{composed DAFSA node} in $\epsilon_c$ as a vector $\underline{n} = (n_1,\ldots,n_k),n_i \in N_{\dafsa_i}$ and the composed state of the reachability graph as a vector $\underline{m} = (m_1,\ldots,m_k),m_i \in N_{\reachGraph_i}$. Both vectors are initialized in lines \ref{line:initialDAFSA} and \ref{line:initialMarking} from the projected initial states, respectively. We represent the arcs of the composed DAFSA as a \emph{partial} vector $\underline{b} = (b_1,\ldots,b_k)$, where a component $b_i = \perp$ may be undefined; a partial vector $\underline{a} = (a_1,\ldots,a_k)$ denotes an arc of a the recomposed reachability graph. For instance, let $a_2 = (m_2,\ell_2,m_2'), a_4 = (m_4,\ell_4,m_4')$ be arcs of the reachability graph of S-components 2 and 4. The vector $\underline{b} = (\perp,b_2,\perp,b_4)$ describes that S-components 2 and 4 synchronize in $\underline{b}$ while S-components 1 and 3 do not participate. We may only synchronize arcs of different S-components if they agree on the label, e.g., if $\ell(b_2) = \ell(b_4)$. The synchronized arc $\underline{b} = (\perp,b_2,\perp,b_4)$ then has the label $\ell(\underline{b}) = \ell(b_2) = \ell(b_4)$ and takes S-component 2 from $m_2$ to $m_2'$ and S-component 4 from $m_4$ to $m_4'$.

As we iterate over the trace $c$ in Alg.~\ref{alg:constPSPwithSComponents}, our composition has to include in the partial vector $\underline{b}$ \emph{all} arcs that agree on the current label $\ell$. First, we give some technical notation for constructing the partial vectors from the available $\dafsa_i$ and $\reachGraph_i$ arcs, and then we explain the loop for the composition.

Suppose we are at the composed marking $\underline{m} = \langle m_1,\ldots,m_k \rangle$ of all the $\reachGraph_i$; the next arcs we can follow in the $\reachGraph_i$ are $a_1,\ldots,a_k, a_i = (m_i,\ell_i,m_i')$. We may follow only those arcs \emph{together} that share the same label. The \emph{partial composition} of these arcs for some label $\ell$ is the vector $\underline{a} = \langle \hat{a}_1,\ldots,\hat{a}_k\rangle$ with $\hat{a}_i = a_i$ if $\ell(a_i) = \ell$ and $\hat{a}_i = \perp$ otherwise. For any component $i$ where $\hat{a}_i \neq \perp$, the state changes from $m_i$ to $m_i'$ and all other components remain in their state. Technically, we write $m_i \blacktriangleright (m_i,\ell,m_i') = m_i'$ and $m_i \blacktriangleright \perp = m_i$ which we lift to $\underline{m} \blacktriangleright \underline{a} = (m_1 \blacktriangleright \hat{a}_1,\ldots,m_k \blacktriangleright \hat{a}_k)$. Thus, $\underline{a}$ traverses all those arcs with label $\ell$ and $\underline{m} \blacktriangleright \underline{a}$ is the composed successor marking reached by this partial synchronization. These definitions equally apply for composing arcs of the $\dafsa_i$.


We now can explain how we compose the projected alignments $\epsilon_i$ into $\epsilon_c$ by composing the arcs of the $\dafsa_i$ and $\reachGraph_i$ in the order in which they occur in the $\epsilon_i,i=1,\ldots,k$. We ``replay'' trace $c$ starting from an empty composed alignment, all projected alignments are at $\mathit{pos}_i = 0$, and at the initial composed nodes $\underline{n}$ and $\underline{m}$ for the $\dafsa_i$ and $\reachGraph_i$ (lines 8-11 in Alg.~\ref{alg:constPSPwithSComponents}).

The next event to replay is $\ell = c(\mathit{pos}_c)$ (line 13 of Alg.~\ref{alg:constPSPwithSComponents}). The next projected synchronizations are $\beta_i = \epsilon_i(\mathit{pos}_i),i=1,\ldots,k$ with $\beta_i = (\mathit{op}_i,a_i,b_i)$. Two cases may arise.
\newpage
Case 1: For all S-components $i$ that have $\ell \in \Sigma_i$ in their alphabet, their next synchronization $\beta_i$ involves arcs labeled with $\ell = \ell(\beta_i)$; lines (17-20 in Alg.~\ref{alg:constPSPwithSComponents}). In this case, all S-components ``agree'' and we can synchronize the $\dafsa_i$ arcs and the $\reachGraph_i$ arcs in the $\beta_i$ of those S-components into a synchronization for $\ell$ in $\epsilon_c$. Again, three cases may arise.
    \begin{enumerate}
    \item All synchronizations $\beta_i$ labeled with $\ell$ agree on the operation $\lhide$ (line 17 in Alg.~\ref{alg:constPSPwithSComponents}). We obtain the next synchronization for $\epsilon_c$ by composing the $\dafsa_i$ arcs with $\ell$ and updated the node $\underline{n}$ for the composed DAFSA as defined by $\mathit{recomposeLogStep}$ in Alg.~\ref{alg:constPSPwithSComponents:recompose_logMoves}.
        The partially composed arc $\underline{a}' = (\underline{n},\ell,\underline{n} \blacktriangleright \underline{a})$ of the $\dafsa_i$ in the new synchronization $(\lhide,\underline{a}',\perp)$ describes that all S-components make a $\lhide$ step together (i.e., no S-component fires a transition for event $\ell$). The new synchronization is appended to $\epsilon_c$ and we advance the position $\mathit{pos}_i$ for all S-components involved in this composition.
\vspace{-.5\baselineskip}        
        \begin{algorithm}[!h]
    \SetKwInOut{Input}{input}
    \Input{ Next log label $\ell$, current partial recomposed alignment $\epsilon_c$, local alignments $\underline{\epsilon}$, positions in local alignments $\underline{pos}$, local DAFSA states $\underline{n}$, labels in each component $\underline{\Sigma}$ }

    $b_i \leftarrow \left\{ \begin{array}{ll}
        \logPr{\epsilon_i(pos_i)} & \mathit{ if } \ell \in \Sigma_i\\
        \perp & \mathit{ otherwise} \end{array}\right.$, for each $i=1,\ldots,k$\;
    $\epsilon_c \leftarrow \epsilon_c \oplus (\lhide, (\underline{n},\ell,\underline{n} \blacktriangleright (b_1,\ldots,b_k)),\perp)$\;
    $\underline{n} \leftarrow \underline{n} \blacktriangleright (b_1,\ldots,b_k)$\;
    $pos_i \leftarrow pos_i+1$ for each $i=1,\ldots,k$ where $\ell \in \Sigma_i$\;

    \Return{$(\epsilon_c, \underline{pos}, \underline{n})$}\;
    \caption{$\mathit{recomposeLogStep}$, recompose log steps for current event $\ell$}\label{alg:constPSPwithSComponents:recompose_logMoves}
\end{algorithm}
\vspace{-1\baselineskip}
    \item  All synchronizations $\beta_i$ labeled with $\ell$ agree on the operation $\match$ (line 19 in Alg.~\ref{alg:constPSPwithSComponents}). We append to $\epsilon_c$ a new $\match$ synchronization with partially composed $\dafsa_i$ arcs and $\reachGraph_i$ arcs, describing that all involved S-components make a $\match$ step together, and update nodes $\underline{n}$ and $\underline{m}$ of the DAFSA and the reachability graph; see function $\mathit{recomposeMatchStep}$ in in Alg.~\ref{alg:constPSPwithSComponents:recompose_syncMoves}.
\vspace{-.5\baselineskip}    
    \begin{algorithm}[!h]
    \SetKwInOut{Input}{input}
    \Input{ Next log label $\ell$, current partial recomposed alignment $\epsilon_c$, local alignments $\underline{\epsilon}$, positions in local alignments $\underline{pos}$, local DAFSA states $\underline{n}$, local model states $\underline{m}$, labels in each component $\underline{\Sigma}$ }

    $b_i \leftarrow \left\{ \begin{array}{ll}
        \logPr{\epsilon_i(pos_i)} & \mathit{ if } \ell \in \Sigma_i\\
        \perp & \mathit{ otherwise} \end{array}\right.$, for each $i=1,\ldots,k$\;
    $a_i \leftarrow \left\{ \begin{array}{ll}
        \modPr{\epsilon_i(pos_i)} & \mathit{ if } \ell \in \Sigma_i\\
        \perp & \mathit{ otherwise} \end{array}\right.$, for each $i=1,\ldots,k$\;
    $\epsilon_c \leftarrow \epsilon_c \oplus (\match, (\underline{n},\ell,\underline{n} \blacktriangleright (b_1,\ldots,b_k)),(\underline{m},\ell,\underline{m} \blacktriangleright (a_1,\ldots,a_k)))$\;
    $\underline{n} \leftarrow \underline{n} \blacktriangleright (b_1,\ldots,b_k)$\;
    $\underline{m} \leftarrow \underline{m} \blacktriangleright (a_1,\ldots,a_k)$\;
    $pos_i \leftarrow pos_i+1$ for each $i=1,\ldots,k$ where $\ell \in \Sigma_i$\;

    \Return{$(\epsilon_c, \underline{pos}, \underline{n}, \underline{m})$}\;
    \caption{$\mathit{recomposeMatchStep}$, recompose matching steps for current event $\ell$}\label{alg:constPSPwithSComponents:recompose_syncMoves}
\end{algorithm}
\vspace{-1\baselineskip}
    \item The partial alignments of some S-components disagree on the operation, i.e., we have conflicting partial solutions (lines 21-22). In this case we fall back to computing a global alignment without decomposition (line 23).
    \end{enumerate}

Case 2: There are S-components $i$ that have $\ell \in \Sigma_i$ in their alphabet, but the next synchronization $\beta_i$ is not labeled with $\ell \neq \ell(\beta_i)$. The set $\configs_\ell$ defined in line 1 of function $\mathit{recomposeModelStepsUntil}$ in Alg.~\ref{alg:constPSPwithSComponents:recompose_modelMoves} contains all these S-components. These S-components have to ``catch up'' with $\rhide$ synchronizations to reach a state where they can participate in a $\lhide$ or $\match$ synchronization over $\ell$ (lines 2-13 of Alg.~\ref{alg:constPSPwithSComponents:recompose_modelMoves}). However, such S-components may only catch up together: Suppose there is an S-component $i$ having as next synchronization an $\rhide$ over $\ell(\beta_i) = x \neq \ell$, then \emph{all} S-components with $x$ in their alphabet (set $\mathit{lab}_x$ in line 4) must \emph{also have} an $\rhide$ synchronization on $x$ as their next synchronization (set $\mathit{sync}_x$ in line 3). If we find such a set $\mathit{sync}_x$ (line 5), then we can compose a $\rhide$ synchronization from the $\reachGraph_i$ arcs in $\mathit{sync}_x$ and append it to $\epsilon_c$ (lines 6-9). This step may have to be repeated if there is another S-component that still has to catch up. If the projected alignments disagree on the next $\rhide$, we have conflicting partial solutions and fall back to computing a global alignment without decomposition (lines 11-12). Note that $\mathit{recomposeModelStepsUntil}$ is called in Alg.~\ref{alg:constPSPwithSComponents} (line 14) for each new trace label $\ell$ and lets all S-components catch up before attempting to synchronize on $\ell$.

\begin{algorithm}
    \SetKwInOut{Input}{input}
    \Input{ Next log label $\ell$, current partial recomposed alignment $\epsilon_c$, local alignments $\underline{\epsilon}$, positions in local alignments $\underline{pos}$, local model states $\underline{m}$, labels in each component $\underline{\Sigma}$, labels in the full log $\Sigma_L$ }
    $\configs_{\ell} \leftarrow \{ (i,\epsilon_i) \mid i=1\ldots k \land pos_i \leq |\epsilon_i| \land ((\ell\in \Sigma_i \land \ell(\epsilon_i(pos_i)) \neq \ell) \lor (\ell = \perp))\}$\;\label{line:configsEll} 
    \While{$\configs_{\ell} \neq \varnothing$\label{line:beginSyncRhide}}
    {
        $\mathit{sync}_x \leftarrow \{ (i,\epsilon_i) \mid \ell(\epsilon_i(pos_i)) = x \land \op(\epsilon_i(pos_i)) = \rhide\}$ for each label $x \in \Sigma_L$\;
        $\mathit{lab}_x \leftarrow \{ (i,\epsilon_i) \mid x \in \Sigma_i \}$ for each label $x \in \Sigma_L$\;
        \If{$\exists x \in \Sigma: \mathit{sync}_x = \mathit{lab}_x$\label{line:condRhideSync}}
            {
                $a_i \leftarrow \left\{ \begin{array}{ll}
                    \modPr{\epsilon_i(pos_i)} & \mathit{ if } (i,\epsilon_i) \in \mathit{sync}_x\\
                    \perp & \mathit{ otherwise} \end{array}\right.$, for each $i=1,\ldots,k$\;
                $\epsilon_c \leftarrow \epsilon_c \oplus (\rhide, \perp, (\underline{m},x,\underline{m} \blacktriangleright (a_1,\dots,a_k)))$\;
                $\underline{m} \leftarrow \underline{m} \blacktriangleright (a_1,\dots,a_k)$\;
                $pos_i \leftarrow pos_i + 1$ for each $(i,\epsilon_i) \in \mathit{sync}_x$\;
                $\configs_{\ell} \leftarrow \{ (i,\epsilon_i) \mid i=1\ldots k \land pos_i \leq |\epsilon_i| \land ((\ell\in \Sigma_i \land \ell(\epsilon_i(pos_i)) \neq \ell) \lor (\ell = \perp))\}$\; 
            }

        \Else
        {
            \tcp{S-components locally aligned models steps in a different order, cannot recompose}
            \Return{$(\mathit{false}, \underline{pos}, \underline{m}, \epsilon_c)$}\label{line:condRhideConflict}
        }
    }
    \Return{$(\mathit{true}, \underline{pos}, \underline{m}, \epsilon_c)$}\;
    \caption{$\mathit{recomposeModelStepsUntil}$, recompose model steps until enabling next log event $\ell$}\label{alg:constPSPwithSComponents:recompose_modelMoves}
\end{algorithm}

%
In this way, we consecutively construct two paths: one through the composition of the $\dafsa_i$ (by the $\underline{a}' = (\underline{n},\ell,\underline{n} \blacktriangleright \underline{a})$ arcs) and one through the composition of the $\reachGraph_i$. In Sect.~\ref{sec:s-components_invisible_label_conflicts} we formally state that these paths correspond to paths through $\dafsa$ and $\reachGraph$ and thus $\epsilon_c$ is an alignment; the proof is given in~\ref{app:alg5_correct}.

%% file: Extension_eval.tex
\section{Evaluation}\label{sec:evaluation}
We implemented our approach in a standalone open-source tool\footnote{Tool available at \url{https://apromore.org/platform/tools}. Source code available at \url{https://github.com/apromore/DAFSABasedConformance}} as part of the Apromore software environment. Given an event log in XES format and a process model in BPMN or PNML (the latter is the serialization format of Petri nets), the tool returns several conformance statistics such as fitness and raw fitness cost. Optionally, the tool can also return a list of one-optimal alignments for each unique trace as well as their individual alignment statistics. The tool implements both the Automata-based approach described in Section \ref{sec:approach} as well as the extended approach with the S-Components improvement described in Section \ref{sec:s-components}.
\newpage
Using this tool, we conducted a series of experiments to measure the time performance and the quality of the results obtained by both our approaches against two exact techniques for computing alignments, which are state-of-the-art, and one approximate technique: 1) the newest version of the one-optimal alignment with the ILP marking equation, first presented in \cite{ILP-Alignment} and
implemented in ProM in the PNetReplayer package (ILP Alignments); 2) the one-optimal alignment approach using the extended marking equation presented in \cite{BVD-Alignment} and implemented in ProM in the Alignment package (MEQ Alignments); and 3) the approximate alignment approach for large instances (ALI)~\cite{ALI} using local search.\footnote{Tool available at \url{https://www.cs.upc.edu/~taymouri/tool.html}} ALI is implemented in Python and we used it in conjunction with the commercial LP solver Gurobi to conduct the experiments.
We implemented multi-threading for each unique trace, and in the S-Components variant, also for each S-Component.

The two baseline implementations for optimal alignments computation use optimized data structures and efficient hashcodes \cite{AlignmentTechnicalImprovements}. Accordingly, we optimized our software implementation using similar techniques, so as to achieve results that are as comparable as possible. Specifically, we optimized the queueing mechanism by improving the selection of suitable solutions, merging overlapping solutions and prioritizing longer solutions with the same cost to find an optimal solution more efficiently. 
While the approximate approach ALI is only implemented as a Python prototype, the authors  previously compared it with the two baselines for optimal alignments computation, and showed to outperform these on a synthetic dataset \cite{ALI}.

\vspace{-.5\baselineskip}
\subsection{Setup}
Using our tool and the reference tools for the three baselines, we conducted a series of experiments to measure the execution time of all four approaches, in milliseconds. Each experiment was run five times and we report the average results of runs \#2 to \#4 to avoid influence of the Java class loader and reduce variance. Given that the complexity of the alignment problem is worst-time exponential, we decided to apply a reasonable time bound of ten minutes to each experiment. We note that previous experiments reported that in certain cases the computation of an alignment may take over a dozen hours \cite{Munoz-GamaCA14}.

Measuring time performance was the primary focus of our experiments. Besides time performance, we measured the quality of alignment. We did so in terms of alignment cost (Def.~\ref{def:cost}) per trace. The alignment cost is a secondary factor of comparison to measure the degree of optimality of the results. We chose to report the alignment cost over other conformance measures such as fitness as it allows one to more precisely pinpoint the over-approximation of the results, if any. The experiments were run multi-threaded with a fixed amount of 16 threads for each approach to achieve a comparable computation setup.
The experiments were conducted on a 22-core Intel Xeon CPU E5-2699 v4 with 2.30GHz, with 128GB of RAM running JVM 8. This machine can execute up to 44 threads per socket. 

Given that a different number of threads can lead to different time performance, we repeated our tests in a single-thread setting. The results of this latter experiment are reported in \ref{app:time_performance_st} and are consistent with those obtained in the multi-thread setting reported in this section.

\vspace{-.5\baselineskip}
\subsection{Datasets}
We used two datasets of log-model pairs from a recent benchmark on automated process discovery \cite{PD-Discovery-BM} 
The first dataset consists of twelve public event logs. These logs in turn originate from the 4TU Centre for Research Data\footnote{\url{https://data.4tu.nl/repository/collection:event_logs_real}}. They include the logs of the Business Process Intelligence Challenge (BPIC) series, BPIC12 \cite{BPIC12}, BPIC13\textsubscript{cp} \cite{BPIC13cp}, BPIC13\textsubscript{inc} \cite{BPIC13inc}, BPIC14 \cite{BPIC14}, BPIC15 \cite{BPIC15}, BPIC17 \cite{BPIC17}, the Road Traffic Fines Management process log (RTFMP) \cite{RTFMP} and the SEPSIS Cases log (SEPSIS) \cite{SEPSIS}. These logs record process executions from different domains such as finance, healthcare, government and IT service management. The BPIC logs from  years 2011 and 2016 (BPIC11 and BPIC16) were excluded since they do not represent real business processes.
The second dataset is composed of eight proprietary logs sourced from several organizations around the world, including healthcare, banking, insurance and software vendors.
In the benchmark, some of the public logs (marked with ``$_f$'') were filtered in order to remove infrequent behavior, using the technique in \cite{Noise-filtering}. The reason for this filtering was that the majority of automated discovery techniques used in the benchmark could not discover a model from the unfiltered log with the allotted memory (i.e.\ they ran into a state-space explosion). In order to guarantee compatibility with the benchmark, we decided to keep these logs as is, i.e.\ we did not remove the filter. 

Each of the two datasets (public and private) comes with four process models per log, that have been discovered using four state-of-the-art automated discovery methods in the benchmark in \cite{PD-Discovery-BM}, namely: Inductive Miner \cite{InductiveMiner}, Split Miner \cite{SplitMiner}, Structured Heuristics Miner \cite{StructuredHeuristicsMiner} and Fodina \cite{Fodina}. We discarded the process models discovered by the latter two methods for our experiments since they may lead to process models with transitions with duplicate labels (and in some cases also to unsound models), which our S-Components extension does not handle. This resulted in a total of 40 log-model pairs for our evaluation.
Table \ref{tb:log_statistics} reports the log characteristics. There are logs of different sizes in terms of total traces (681--787,667) or total number of events (6,660--1,808,706). The difficulty of the conformance checking problem, however, is more related to the percentage of distinct traces (0.01\%--97.5\%), the number of distinct events (7--82) and the trace length (avg. 1--32). These logs thus feature a wide range of characteristics, and include both simple and complex logs. For reference, we made the public logs and the corresponding models, together with all the results of our experiments, available online \cite{reissner_2019}.

\begin{table}[ht]
{\footnotesize{
\setlength{\tabcolsep}{6pt}
\centering{
\begin{tabular}{|c|c c c c|c c c|}
\hline\textbf{Log} & \textbf{Total} & \textbf{Distinct} & \textbf{Total} & \textbf{Distinct}  & \multicolumn{3}{ c |}{\textbf{Trace Length}} \\ \cline{6-8}
\textbf{Name} & \textbf{Traces} & \textbf{Traces} (\%) & \textbf{Events} & \textbf{Events} & \textbf{min} & \textbf{avg} & \textbf{max}\\ \hline
BPIC12 & 13,087 & 33.4 & 262,200 & 36 & 3 & 20 & 175\\ \hline
BPIC13\textsubscript{cp} & 1,487 & 12.3 & 6,660 & 7 & 1 & 4 & 35\\ \hline
BPIC13\textsubscript{inc} & 7,554 & 20.0 & 65,533 & 13 & 1 & 9 & 123\\ \hline
BPIC14\textsubscript{f} & 41,353 & 36.1 & 369,485 & 9 & 3 & 9 & 167\\ \hline
BPIC15\textsubscript{1f} & 902 & 32.7 & 21,656 & 70 & 5 & 24 & 50\\ \hline
BPIC15\textsubscript{2f} & 681 & 61.7 & 24,678 & 82 & 4 & 36 & 63\\ \hline
BPIC15\textsubscript{3f} & 1,369 & 60.3 & 43,786 & 62 & 4 & 32 & 54\\ \hline
BPIC15\textsubscript{4f} & 860 & 52.4 & 29,403 & 65 & 5 & 34 & 54\\ \hline
BPIC15\textsubscript{5f} & 975 & 45.7 & 30,030 & 74 & 4 & 31 & 61\\ \hline
BPIC17\textsubscript{f} & 21,861 & 40.1 & 714,198 & 41 & 11 & 33 & 113\\ \hline
RTFMP & 150,370 & 0.2 & 561,470 & 11 & 2 & 4 & 20\\ \hline
SEPSIS & 1,050 & 80.6 & 15,214 & 16 & 3 & 14 & 185\\ \hline\hline
PRT1 & 12,720 & 8.1 & 75,353 & 9 & 2 & 5 & 64\\ \hline
PRT2 & 1,182 & 97.5 & 46,282 & 9 & 12 & 39 & 276\\ \hline
PRT3 & 1,600 & 19.9 & 13,720 & 15 & 6 & 8 & 9\\ \hline
PRT4 & 20,000 & 29.7 & 166,282 & 11 & 6 & 8 & 36 \\ \hline
PRT6 & 744 & 22.4 & 6,011 & 9 & 7 & 8 & 21\\ \hline
PRT7 & 2,000 & 6.4 & 16,353 & 13 & 8 & 8 & 11\\ \hline
PRT9 & 787,657 & 0.01 & 1,808,706 & 8 & 1 & 2 & 58\\ \hline
PRT10 & 43,514 & 0.01 & 78,864 & 19 & 1 & 1 & 15\\ \hline
\end{tabular}
\vspace{-.5\baselineskip}
}\caption{Descriptive statistics of the event logs in the public and private datasets}\label{tb:log_statistics}
\vspace{-.5\baselineskip}
}}
\end{table}

Table \ref{tb:model_statistics} reports the statistics of the process models obtained with Inductive (IM) and Split Miner (SM), for each log in our evaluation. Specifically, this table reports size (number of places, transitions and arcs), number of transitions, number gateways (XOR-splits, AND-splits) and size of the resulting reachability graph from the Petri net (in case of a BPMN model, it is the Petri net obtained from this model). In addition, if a Petri net has at least one AND-split, we also report on the number of S-Components and for each of them the following statistics: their average Petri net size, average number of transitions, average number of XOR-splits and average size of the resulting reachability graph.

Inductive Miner is designed to discover highly-fitting models. As a result, the models often exhibit a large reachability graph as the models need to cater for a large variety of executions present in the logs. Split Miner strikes a trade-off between fitness and precision by filtering the directly-follows graph of the log before discovering the model. That leads to process models with a smaller state space, but with a possibly higher number of fitness mismatches. Altogether, these models present two different scenarios for conformance checking: the models discovered by Inductive Miner require a large state space to be traversed with a low to medium number of mismatches per trace, while the models of Split Miner have a smaller state space with a medium to high number of mismatches per trace.

The S-Component decomposition can drastically reduce the size of the state space of the model. This becomes apparent when comparing the size of the reachability graph with that of the S-Component reachability graphs, e.g.\ BPIC12 (IM) reduces from 1,997 nodes and arcs to a total of 583 nodes and arcs and BPIC14\textsubscript{f} from 4,383 to a total of 261 nodes and arcs. This reduction depends on the internal structure of the model, i.e.\ the number of S-components and the nesting of XOR-splits and AND-splits. Sometimes, this reduction will not lead to a smaller state space, e.g.\ for BPIC15\textsubscript{3f} (IM) the size reduces from 875 to 191.5 per S-Component, which leads to a total state space of 1,532 nodes and arcs for all S-Components, which is larger than the size of the original model.

\begin{table}[!h]
{\footnotesize{
\setlength{\tabcolsep}{3pt}
\centering{
\begin{tabular}{|c|c|c|c c c c c|c c c c c|}
\hline
\textbf{Miner} & \textbf{Domain} & \textbf{Dataset} & \textbf{Size} & \textbf{Trns}  & \textbf{XOR} & \textbf{AND} & \textbf{RG Size} & \textbf{\#SComp} & \textbf{$\varnothing$ Size} & \textbf{$\varnothing$ Trns} & \textbf{$\varnothing$ XOR} & \textbf{$\varnothing$ RG Size} \\ \hline
\multirow{20}{*}{IM} & \multirow{12}{*}{public} & BPIC12 & 177 & 45 & 16 & 2 & 1,997 & 10 & 130.9 & 36.3 & 12.3 & 58.3\\ \cline{3-13}
 & & BPIC13\textsubscript{cp} & 31 & 8 & 2 & 0 & 9 & 1 & - & - & - & -\\ \cline{3-13}		
 & & BPIC13\textsubscript{inc} & 56 & 13 & 3 & 1 & 121 & 3 & 28 & 7 & 1 & 14\\ \cline{3-13}
 & & BPIC14\textsubscript{f} & 124 & 29 & 8 & 2 & 4,383 & 10 & 53.9 & 13.9 & 2.7 & 26.1\\ \cline{3-13}
 & & BPIC15\textsubscript{1f} & 449 & 127 & 48 & 0 & 719 & 1 & - & - & - & - \\ \cline{3-13}	
 & & BPIC15\textsubscript{2f} & 537 & 150 & 55 & 1 & 1,019 & 2 & 530 & 149 & 55 & 232 \\ \cline{3-13}
 & & BPIC15\textsubscript{3f} & 464 & 128 & 47 & 3 & 875 & 8 & 438.5 & 123.5 & 45.5 & 191.5\\ \cline{3-13}
 & & BPIC15\textsubscript{4f} & 469 & 131 & 51 & 1 & 1,019 & 2 & 462 & 130 & 51 & 202\\ \cline{3-13}
 & & BPIC15\textsubscript{5f} & 381 & 111 & 31 & 0 & 429 & 1 & - & - & - & -\\ \cline{3-13}		
 & & BPIC17\textsubscript{f} & 121 & 33 & 8 & 0 & 59 & 1 & - & - & - & - \\ \cline{3-13}
 & & RTFMP & 111 & 26 & 9 & 2 & 2,394 & 6 & 52.7 & 13.8 & 3.7 & 25 \\ \cline{3-13}
 & & SEPSIS & 145 & 37 & 13 & 3 & 2,274 & 8 & 99 & 27.5 & 9 & 44 \\ \cline{2-13}
 & \multirow{8}{*}{private} & PRT1 & 70 & 16 & 4 & 1 & 195 & 4 & 39.3 & 10 & 1.8 & 19.3\\ \cline{3-13}
 & & PRT2 & 175 & 43 & 16 & 1 & 5,515,357 & 7 & 49 & 13 & 4 & 23\\ \cline{3-13}
 & & PRT3 & 111 & 27 & 8 & 2 & 167 & 8 & 71 & 19 & 4 & 33\\ \cline{3-13}
 & & PRT4 & 91 & 21 & 5 & 2 & 154 & 8 & 54 & 14 & 2 & 26\\ \cline{3-13}
 & & PRT6 & 86 & 20 & 4 & 2 & 65 & 6 & 59 & 15 & 2 & 29\\ \cline{3-13}
 & & PRT7 & 99 & 23 & 5 & 2 & 158 & 8 & 62 & 16 & 2 & 30\\ \cline{3-13}
 & & PRT9 & 96 & 21 & 7 & 2 & 9,121 & 7 & 27.9 & 7 & 1.1 & 13.9\\ \cline{3-13}
 & & PRT10 & 124 & 35 & 8 & 1 & 184 & 2 & 115.5 & 33.5 & 7.5 & 48.5\\ \hline\hline
\multirow{20}{*}{SM} & \multirow{12}{*}{public} & BPIC12 & 315 & 85 & 29 & 1 & 95 & 2 & 308 & 84 & 29 & 140 \\ \cline{3-13}
 & & BPIC13\textsubscript{cp} & 49 & 13 & 4 & 0 & 13 & 1 & - & - & - & - \\ \cline{3-13}		
 & & BPIC13\textsubscript{inc} &  56 & 15 & 5 & 0 & 17 & 1 & - & - & - & - \\ \cline{3-13}		
 & & BPIC14\textsubscript{f} &  88 & 24 & 9 & 0 & 24 & 1 & - & - & - & - \\ \cline{3-13}
 & & BPIC15\textsubscript{1f} & 368 & 98 & 25 & 0 & 156 & 1 & - & - & - & - \\ \cline{3-13}		
 & & BPIC15\textsubscript{2f} & 444 & 117 & 25 & 0 & 186 & 1 & - & - & - & - \\ \cline{3-13}		
 & & BPIC15\textsubscript{3f} & 296 & 78 & 17 & 0 & 136 & 1 & - & - & - & - \\ \cline{3-13}
 & & BPIC15\textsubscript{4f} & 323 & 85 & 18 & 0 & 141 & 1 & - & - & - & - \\ \cline{3-13}
 & & BPIC15\textsubscript{5f} & 359 & 94 & 18 & 0 & 159 & 1 & - & - & - & - \\ \cline{3-13}
 & & BPIC17\textsubscript{f} & 149 & 40 & 12 & 0 & 54 & 1 & - & - & - & - \\ \cline{3-13}
 & & RTFMP & 102 & 28 & 11 & 0 & 37 & 1 & - & - & - & - \\ \cline{3-13}				
 & & SEPSIS & 162 & 44 & 15 & 0 & 41 & 1 & - & - & - & - \\ \cline{2-13}				
 & \multirow{8}{*}{private} & PRT1 & 104 & 28 & 9 & 0 & 28 & 1 & - & - & - & - \\ \cline{3-13}
 & & PRT2 & 166 & 45 & 15 & 0 & 37 & 1 & - & - & - & - \\ \cline{3-13}
 & & PRT3 & 96 & 25 & 8 & 1 & 34 & 2 & 89 & 24 & 8 & 41 \\ \cline{3-13}
 & & PRT4 & 126 & 33 & 10 & 1 & 34 & 2 & 119 & 32 & 10 & 55 \\ \cline{3-13}
 & & PRT6 & 46 & 11 & 2 & 1 & 20 & 2 & 39 & 10 & 2 & 19 \\ \cline{3-13}
 & & PRT7 & 86 & 19 & 3 & 5 & 39 & 6 & 57 & 14.7 & 2.7 & 27.7 \\ \cline{3-13}
 & & PRT9 & 107 & 29 & 10 & 0 & 32 & 1 & - & - & - & - \\ \cline{3-13}
 & & PRT10 & 327 & 90 & 34 & 0 & 92 & 1 & - & - & - & - \\ \hline
 \end{tabular}
 }\vspace{.5\baselineskip}
\caption{Descriptive statistics of the process models obtained by IM and SM from the datasets in Table \ref{tb:log_statistics}}\label{tb:model_statistics}
\vspace{1.5\baselineskip}
}}
\end{table}

\subsection{Results}\label{sec:evaluation:results}
Table \ref{tb:time_performance} reports the running times in milliseconds for each approach against each of the 40 log-model pairs, using a fixed number of 16 threads per approach. The best execution time for each experiment is highlighted in bold.

\vspace{.5\baselineskip}
\textbf{Analyzing the overall performance.}
The S-Components approach outperforms the other approaches in eight out of 40 log-model pairs; the Automata-based technique performs best in 28 out of 40 cases, and the extended MEQ Alignments approach outperforms in three out of 40 cases. 
The total time spent for all 40 datasets, excluding any timeouts, is comparable between the Automata-based approach and the S-Component approach (circa 230 seconds with the S-Components approach gaining 17 seconds). Both our approaches improve on ILP Alignments by 500 seconds and on the extended MEQ Alignments and ALI by over 1,000 seconds.

\vspace{.5\baselineskip}
\textbf{Investigating timeout cases.}
In total, the S-Components approach times out (``t/out'' in the table) in two cases, the Automata-based approach in three, ILP Alignments in one case and the extended MEQ alignments approach times out in six cases. ALI is the only approach that never times out in the datasets used in our experiments, though, as discussed later in the section, this approach always over-approximates. While all other approaches timed out on PRT2 (IM), which has a huge state space of 5,515,357 nodes and arcs in the reachability graph, ALI managed to compute approximate alignments for this dataset. The S-Components approach actually manages to compute alignments quickly for this log-model pair since the S-Component reachability graphs are very small, but times out when some traces conflict with each other in the recomposition algorithm and need to be aligned on the original reachability graph, which is much larger. The S-Components approach manages to compute alignments for the BPIC14\textsubscript{f} (IM), which was not possible for the Automata-based approach within the 10-minute bound of the timeout.

\vspace{.5\baselineskip}
\textbf{Improvements of our approaches.}
The Automata-based approach performs better than both state-of-the-art approaches ILP and MEQ Alignments by one-two orders of magnitude. For example, for the BPIC17\textsubscript{f} (IM) it takes 680 ms against 20.7 sec of ILP Alignments or for BPIC12 (SM) it takes 4,578 ms vs 188,489 ms of ILP Alignments. When the state space reduction of the S-Components is effective, it shows the potential to improve over other approaches by at least one order of magnitude, e.g.\ for BPIC12 (IM) it improves from 121,845 ms (Automata-based) to 44,301 ms and for BPIC14\textsubscript{f} (IM) from 84,102 ms (ILP) to 7,789 ms.
In total the Automata-based approach improves over all baseline approaches by one order of magnitude in ten datasets and the S-Component approach in three datasets. The S-Component extension improves over the automata-based approach by one order of magnitude in five cases.

\vspace{.5\baselineskip}
\textbf{Problematic cases of the S-Component approach.}
The process models discovered by Split Miner do not feature parallel constructs except the model discovered from the BPIC12 log. Thus, in these logs, the performance of the S-Component extension is the same as that of the automata-based approach. In the BPIC12 (SM) case, the Automata-based technique outperforms the S-Component approach because it exploits the parallel constructs in the model. This is due to the combined state-space of the S-Component reachability graphs being larger than the original size of the reachability graph of the process model. This can already happen for process models with a small amount of parallel behavior in comparison to models with a large amount of other behavior, i.e.\ the model from the BPIC12 log has one parallel block with two parallel transitions against a total of 85 transitions. The log-model pair  BPIC15\textsubscript{3f} of IM exhibits similar problems.

\vspace{.5\baselineskip}
\textbf{Hybrid approach: definition and performance.}
Since the advantages of the S-Components decomposition are limited to a specific type of process models (those with large state spaces due to a high degree of parallelism), we derived an empirical rule to decide when to use the S-Components improvement on top of our Automata-based approach. Accordingly, we apply this improvement if the sum of the reachability graph sizes of all S-Components is smaller than that of the original reachability graph of the process model. We added the execution times for this hybrid approach to Table \ref{tb:time_performance}. 

In total, the hybrid approach gains 100 seconds over the Automata-based approach and  improves the results by ALI and the extended MEQ Alignments by one order of magnitude.
In detail, the hybrid approach manages to outperform all other approaches in 30 out of 40 cases and performs second best in five more cases. 
We note that the reported execution time of the hybrid approach does not include the time required to decide whether or not to apply the S-Components improvement. If we end up selecting the S-Components, we do not actually need additional time, since the reachability graphs for the S-Component nets are computed anyways as part of the decomposition approach. If we select the base approach, this leads to two cases: the model does not have parallelism or it does. If it does not, we detect this case by checking all transitions of the Petri net, which is a linear operation, so the time is negligible. If the model has parallelism, we need to calculate the reachability graphs for every S-Component net. In practice, this time was always negligible in our experiments, but there can be very large process models for which this operation may be expensive. However, in these cases, it is likely that we would select the S-Component approach anyway.

\input{tableTimePerformanceMT.tex}

\input{tableCostReduced}

\vspace{.5\baselineskip}
\textbf{Analyzing cost over-approximations.}
Table \ref{tb:cost_approximation} shows the optimal costs for a subset of datasets. In these log-model pairs, the S-Components approach over-approximates the optimal cost of the alignments, i.e.\ in 6 out of 40 cases. For completeness the full table with optimal costs for all datasets can be found in \ref{app:cost_comparison}. The difference between the S-Component approach and all other approaches with optimal costs ranges from 0.002 to 0.052 per trace. We further broke down the over-approximation into two columns: the fraction of traces in the log that were affected by an over-approximation, which ranges from 0.2 to 5.2\%, and the average fitness-cost that was over-approximated in the affected traces, which ranges from 1 to 2 mismatches more than the optimal number. We observe that the approach never under-approximates and always returns proper alignments.
In comparison, ALI over-approximates in all datasets and the degree of over-approximation is much larger. It ranges from 0.83 up to 26.22. 
The fraction of approximated traces is also much larger, ranging from 9.3 to 100\% of aligned traces.
By design, the Automata-based approach always has the same cost as the ILP or the MEQ Alignments and thus is always optimal.

\textbf{Investigating causes of over-approximations.}
One example of over-approximation can be observed in the SEPSIS dataset (IM) for the trace $\langle$CRP, Leucocytes, LacticAcid, ER Registration, ER Triage, ER Sepsis Triage, IV Antibiotics, IV Liquid$\rangle$. The optimal alignment for this trace, retrieved with ILP-Alignments, is $\langle(\rhide, $ER Registration$), (\match, $CRP$), (\match, $Leucocytes$), (\match, $Lactic -Acid$), (\lhide, $ER Registration$), (\match, $ER Triage$), (\match, $ER Sepsis Triage$), (\match, $IV Antibiotics$), (\match, $IV Liquid$)\rangle$ with a cost of 2, because task ER Registration is misplaced after the parallel block. The S-Components approach finds instead the following alignment: $\langle(\lhide, $CRP$), (\lhide, $Leucocytes$), (\lhide, $LacticAcid$), (\match, $ER Registration$), (\match, $ER Triage$), (\match, $ER Sepsis Triage$), (\match, $IV Antibiotics$), (\match, $IV Liquid$)\rangle$ with a cost of 3. As shown in Figure~\ref{fig:SepsisIM}, in the process model, task ER Registration appears before the parallel block, while in the trace this occurs after the activities in a parallel block. As a result, the S-Component approach will hide all the activities in the parallel block, i.e.\ CRP, Leucocytes and LacticAcid, and then match the activity ER Registration. When recomposing the projected alignments, however, the added alignment cost will be 3 instead of 2. Note that the alignment of the S-Components approach is still a proper alignment, i.e.\ it represents the trace and forms a path through the process model.

\begin{figure}[h!]
\centering
\includegraphics[width=1\textwidth]{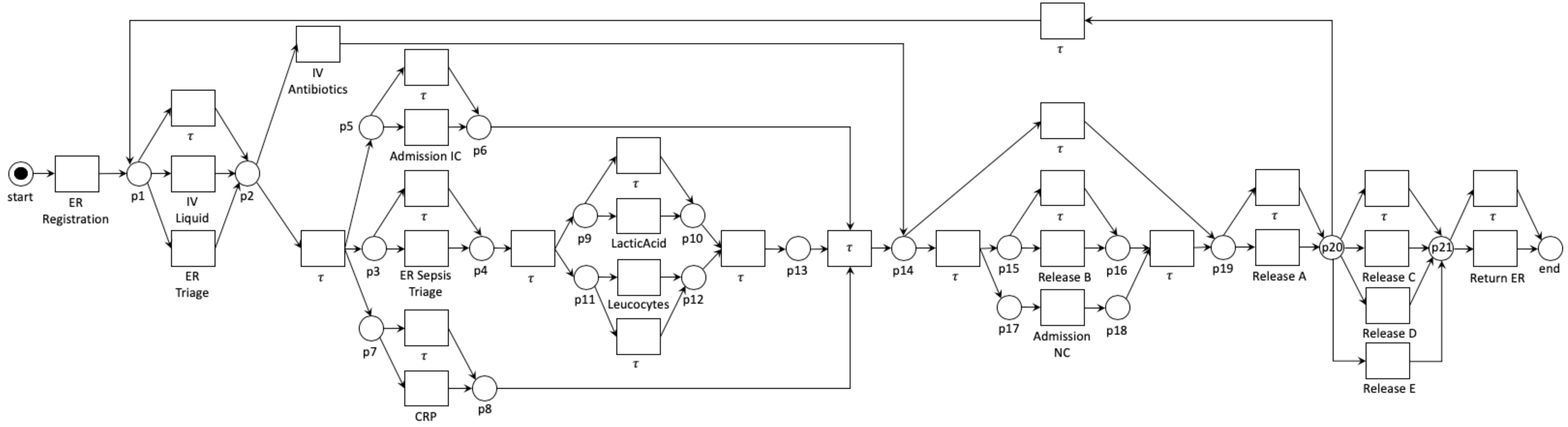}
\vspace{-1em}
\caption{Sepsis Inductive Miner process model.}\label{fig:SepsisIM}
\vspace{-1em}
\end{figure}

\subsection{Threats to validity}

A potential threat to validity is the selection of datasets. We decided to use two datasets of real-life log-model pairs from a recent discovery benchmark \cite{PD-Discovery-BM}. These datasets exhibit a wide range of structural characteristics and originate from different industry domains, so they provide a good representation of reality. However, the models discovered by Split Miner did not contain a lot of parallel structures and were thus not highlighting the strengths of the S-Components decomposition. This calls for further experiments with models with a higher degree of parallelism, and more in general, with very large real-life log-model pairs. Such datasets are not publicly available at the time of writing. An alternative, is to use artificial datasets as in \cite{SyntheticDataset}.

The selection of benchmark approaches is another threat to validity. The approaches presented in this article are only applicable to a subclass of Petri nets, namely 1-safe sound workflow nets, while the two exact approaches are applicable to a wider class of Petri nets, namely easy sound Petri nets, and the
approximate technique is applicable to sound Petri nets.
To the best of our knowledge, however, there are no conformance checking techniques available that target this specific subclass of Petri nets for a better comparison.
In addition, this specific class of Petri nets has relevance to the the field of Process Mining since BPMN models can be translated to this class and several mining algorithms such as Split Miner \cite{SplitMiner}, Inductive Miner \cite{InductiveMiner} or Fodina \cite{Fodina} produce Petri nets of this class.

A final threat to validity is posed by the number of methods used for automated process discovery (two). Potentially we could have chosen a larger number of methods. The choice of Split Miner and Inductive Miner was determined by both pragmatic reasons (other methods such as Structured Heuristics Miner return models with duplicate labels which we cannot handle, or led to models for which fitness could not be computed) as well as by the need to test two extreme cases: models with large state spaces versus models with large degrees of parallelism. Moreover, they are the best performing automated discovery methods according to the benchmark in \cite{PD-Discovery-BM}. So, all considered, they constitute a sufficiently representative set of discovery methods. 

%% file: tableTimePerformanceMT.tex
 \begin{table}[htbp!]                 
 {\footnotesize{                 
 \setlength{\tabcolsep}{3pt}                 
 \centering{                 
 \begin{tabular}{|c|c|c|c c c|c c|c|}                 
 \hline                 
  &  &  & \multicolumn{3}{c|}{\textbf{Baselines}} & \multicolumn{3}{c|}{\textbf{Our Approaches}}   \\ \cline{4-9}
 \textbf{Miner}&  \textbf{Domain}&  \textbf{Dataset}&  \textbf{ILP}&  \textbf{extended MEQ}&  \textbf{ALI}&  \textbf{Automata-based}&  \textbf{S-Component}&  \textbf{Hybrid}\\
  &  & &  \textbf{Alignments}&  \textbf{Alignments} &  &  \textbf{Approach}&  \textbf{Approach}&  \textbf{Approach}\\ \hline
\multirow{20}{*}{ IM} & \multirow{12}{*}{ public}& BPIC12 & 129,268 & 464,849 &  49,018  & 121,845 & \textbf{44,301}& \textbf{44,301}\\ \cline{3-9}
  &  & BPIC13\textsubscript{cp} & 97 & 1,362 &  3,836  &  \textbf{45}& 52 &  \textbf{45}\\ \cline{3-9}
  &  & BPIC13\textsubscript{inc} & 1,514 & 17,857 &  27,850  & 4,733 & \textbf{155}& \textbf{155}\\ \cline{3-9}
  &  & BPIC14\textsubscript{f} & 84,102 & 143,006 &  100,351  & t/out & \textbf{7,789}& \textbf{7,789}\\ \cline{3-9}
  &  & BPIC15\textsubscript{1f} & 3,323 & \textbf{1,312}&  19,651  & t/out & t/out & t/out \\ \cline{3-9}
  &  & BPIC15\textsubscript{2f} & 19,663 & 4,300 &  43,863  & \textbf{2,423}& 9,541 & 9,541 \\ \cline{3-9}
  &  & BPIC15\textsubscript{3f} & 30,690 & 15,644 &  47,525  & \textbf{2,042}& 22,191 & \textbf{2,042}\\ \cline{3-9}
  &  & BPIC15\textsubscript{4f} & 11,696 & 6,574 &  33,958  & \textbf{1,023}& 4,915 & 4,915 \\ \cline{3-9}
  &  & BPIC15\textsubscript{5f} & 6,830 & \textbf{2,859}&  23,239  & 30,648 & 17,577 & 30,648 \\ \cline{3-9}
  &  & BPIC17\textsubscript{f} & 20,745 & 63,335 &  60,621  & \textbf{680}& 3,530 & \textbf{680}\\ \cline{3-9}
  &  & RTFMP & 1,846 & 108,818 &  5,884  & \textbf{334}& 507 & 507 \\ \cline{3-9}
  &  & SEPSIS & 6,592 & 4,347 &  23,842  & 31,113 & \textbf{1,575}& \textbf{1,575}\\ \cline{2-9}
  &\multirow{8}{*}{ private}& PRT1 & 868 & 8,860 &  10,812  & 774 & \textbf{170}& \textbf{170}\\ \cline{3-9}
  &  & PRT2 & t/out & t/out & \textbf{ 17,296 }& t/out & t/out & t/out \\ \cline{3-9}
  &  & PRT3 & 203 & 1,176 &  6,500  & \textbf{60}& 106 & \textbf{60}\\ \cline{3-9}
  &  & PRT4 & 4,041 & 14,396 &  32,811  & \textbf{1,021}& 1,329 & \textbf{1,021}\\ \cline{3-9}
  &  & PRT6 & 131 & 582 &  6,223  & \textbf{32}& 72 & \textbf{32}\\ \cline{3-9}
  &  & PRT7 & 106 & 1,456 &  5,898  & \textbf{34}& 70 & \textbf{34}\\ \cline{3-9}
  &  & PRT9 & 30,772 & t/out &  10,738  & 12,686 & \textbf{4,505}& \textbf{4,505}\\ \cline{3-9}
  &  & PRT10 & 562 & 33,813 &  8,677  & \textbf{63}& 166 & 166 \\ \hline\hline
\multirow{20}{*}{ SM} & \multirow{12}{*}{ public}& BPIC12 & 188,489 & 487,878 &  118,053  & \textbf{4,578}& 68,246 & \textbf{4,578}\\ \cline{3-9}
  &  & BPIC13\textsubscript{cp} & 128 & 1,934 &  4,351  & \textbf{28}& 42 & \textbf{28}\\ \cline{3-9}
  &  & BPIC13\textsubscript{inc} & 1,120 & 23,497 &  18,906  & \textbf{140}& 290 & \textbf{140}\\ \cline{3-9}
  &  & BPIC14\textsubscript{f} & 37,987 & t/out &  157,733  & \textbf{3,185}& 7,792 & \textbf{3,185}\\ \cline{3-9}
  &  & BPIC15\textsubscript{1f} & 2,972 & \textbf{847}&  12,986  & 1,528 & 1,397 & 1,528 \\ \cline{3-9}
  &  & BPIC15\textsubscript{2f} & 9,128 & 1,352 &  26,114  & 1,036 & \textbf{1,026}& 1,036 \\ \cline{3-9}
  &  & BPIC15\textsubscript{3f} & 7,282 & 2,911 &  25,067  & \textbf{609}& 934 & \textbf{609}\\ \cline{3-9}
  &  & BPIC15\textsubscript{4f} & 7,258 & 1,804 &  19,893  & \textbf{585}& 655 & \textbf{585}\\ \cline{3-9}
  &  & BPIC15\textsubscript{5f} & 8,805 & 1,423 &  24,608  & 765 & \textbf{660}& 765 \\ \cline{3-9}
  &  & BPIC17\textsubscript{f} & 27,011 & 22,572 &  265,169  & \textbf{727}& 3,669 & \textbf{727}\\ \cline{3-9}
  &  & RTFMP & 1,836 & 112,328 &  7,020  & \textbf{110}& 345 & \textbf{110}\\ \cline{3-9}
  &  & SEPSIS & 4,912 & t/out &  18,821  & \textbf{276}& 343 & \textbf{276}\\ \cline{2-9}
  &\multirow{8}{*}{ private}& PRT1 & 1,954 & 21,642 &  10,788  & \textbf{135}& 379 & \textbf{135}\\ \cline{3-9}
  &  & PRT2 & 37,556 & t/out &  91,186  & \textbf{3,830}& 3,836 & \textbf{3,830}\\ \cline{3-9}
  &  & PRT3 & 181 & 1,524 &  6,439  & \textbf{55}& 89 & \textbf{55}\\ \cline{3-9}
  &  & PRT4 & 7,433 & 49,980 &  47,776  & \textbf{475}& 1,794 & \textbf{475}\\ \cline{3-9}
  &  & PRT6 & 75 & 603 &  5,997  & \textbf{33}& 52 & \textbf{33}\\ \cline{3-9}
  &  & PRT7 & 96 & 1,439 &  5,887  & \textbf{26}& 103 & \textbf{26}\\ \cline{3-9}
  &  & PRT9 & 30,835 & t/out &  18,165  & \textbf{1,336}& 1,776 & \textbf{1,336}\\ \cline{3-9}
  &  & PRT10 & 828 & 36,458 &  9,374  & \textbf{78}& 92 & \textbf{78}\\ \hline \hline
\multicolumn{3}{|c|}{ Total time spent (ms): }   & 728,934 & 1,662,737 & 1,432,927 & 229,090 & 212,069 & \textbf{127,721} \\ \hline
\multicolumn{3}{|c|}{ Total Outperforming: }   & 0 & 3 & 1 & 28 & 8 & \textbf{30} \\ \hline
\multicolumn{3}{|c|}{ Total second: }   & \textbf{5} & 3 & 2 & 3 & 26 & \textbf{5} \\ \hline
\multicolumn{3}{|c|}{ \#Timeouts: }   & 1 & 6 & \textbf{0} & 3 & 2 & 2 \\ \hline
 \end{tabular}                 
}                 
 \vspace{.5\baselineskip}                 
 \caption{Time performance in milliseconds}\label{tb:time_performance}                 
 \vspace{.5\baselineskip}                 
}}                 
\end{table}                 

%% file: tableCostReduced.tex
  \begin{table}[htbp!]                           
 {\footnotesize{                           
 \setlength{\tabcolsep}{3pt}                           
 \centering{                           
 \begin{tabular}{|c|c|c|c c c c c|c c|c c|c c|}                           
 \hline                           
  &  &  &\multicolumn{5}{|c|}{\textbf{ Cost }} &\multicolumn{2}{|c|}{\textbf{ Over-Apprx. }} & \multicolumn{2}{|c|}{\textbf{ \%Apprx traces }} &  \multicolumn{2}{|c|}{\textbf{ Avg Over Apprx. }}  \\ \cline{4-14}
\textbf{ Miner }&\textbf{ Domain }&\textbf{ Dataset }& \textbf{ ILP }& \textbf{ eMEQ }& \multirow{2}{*}{\textbf{ ALI }}& \textbf{ Aut. }& \textbf{ S-Comp}& \multirow{2}{*}{\textbf{ ALI }}& \textbf{ S-Comp}& \multirow{2}{*}{\textbf{ ALI }}& \textbf{ S-Comp}& \multirow{2}{*}{\textbf{ ALI }}& \textbf{ S-Comp}\\
  &  &  & \textbf{ Align.}& \textbf{ Align. } &  & \textbf{ Appr. }& \textbf{ Appr. } &  & \textbf{ Appr. } &   & \textbf{ Appr. } &  & \textbf{ Appr. }\\ \hline
\multirow{4}{*}{ IM } & \multirow{2}{*}{ public }& BPIC15\textsubscript{2f} & 2.02 & 2.02 & 26.73 & 2.02 & 2.07 & 24.71 & 0.05 & 100.0 \% & 5.2 \% & 24.71 & 1 \\ \cline{3-14}
  &  & SEPSIS & 0.12 & 0.12 & 12.83 & 0.12 & 0.12 & 12.69 & 0.00 & 100.0 \% & 0.2 \% & 12.69 & 1 \\ \cline{2-14}
  & \multirow{2}{*}{private} & PRT6 & 0.09 & 0.09 & 2.31 & 0.09 & 0.12 & 2.22 & 0.03 & 69.9 \% & 2.7 \% & 3.17 & 1.05 \\ \cline{3-14}
  &  & PRT9 & 0.38 & t/out & 2.14 & 0.38 & 0.39 & 1.79 & 0.01 & 98.9 \% & 0.9 \% & 1.80 & 1.50 \\ \hline\hline
\multirow{2}{*}{ SM }& public & BPIC12 & 1.29 & 1.29 & 17.87 & 1.29 & 1.30 & 16.58 & 0.02 & 73.6 \% & 0.8 \% & 22.55 & 2 \\ \cline{2-14}
  & private & PRT7 & 1.40 & 1.40 & 2.22 & 1.40 & 1.41 & 0.83 & 0.01 & 37.0 \% & 1.2 \% & 2.23 & 1 \\ \hline
 \end{tabular}                           
 }                           
 \vspace{.5\baselineskip}                           
 \caption{Cost comparison and order of approximation for those log-model pairs where S-Components over-approximate}\label{tb:cost_approximation}                           
 \vspace{.5\baselineskip}                           
 }}                           
 \end{table}                           

%% file: conclusion.tex
\section{Conclusion}\label{sec:conclusion}

This article presented two contributions to the field of conformance checking of event logs against process models.
First, the article showed that the problem of conformance checking can be mapped to that of computing a minimal partially synchronised product between an automaton representing the event log (its minimal DAFSA) and an automaton representing the process model (its reachability graph). The resulting product automaton can be used to extract optimal alignments between each trace in the log and the closest corresponding trace of the model.

The use of a DAFSA to represent the event log allows us to reuse the optimal alignment computations for prefixes or suffixes that are shared by multiple traces. This is a distinctive feature of the proposal with respect to existing trace alignment techniques, which compute an optimal alignment for each trace in the log separately, without any reuse across traces. 
The empirical evaluation shows that this approach outperforms state-of-the-art trace alignment techniques in a clear majority of cases. 

However, like other techniques for computing optimal trace alignments, the proposed automata-based technique suffers from degraded performance when the process model has a high degree of concurrency. This issue stems from the combinatorial state explosion inherent to such process models. This state explosion, in turn, leads to larger search spaces of possible alignments.
To address this shortcoming, the article presented a divide-and-conquer approach, wherein the process model is decomposed into a collection of concurrency-free components, namely S-components. Each of these S-components (which corresponds to an automaton) is then aligned separately against a projected version of the log, leading to one product automaton per S-component. The article spells out criteria under which the resulting product automata can be recomposed into a correct (although not necessarily optimal) product automata of the original event log and process model. When two or more S-components cannot be re-composed due to a conflict, these S-components are merged and a product automaton is computed for the merged S-component. 
The evaluation showed that this decomposition-based approach achieves lower execution times than the monolithic automata-based approach when the number of S-components is high, in part thanks to the fact that the decomposition-based approach lends itself to parallel computation.
The evaluation also showed that the decomposition-based approach computes optimal alignments in the majority of cases. In those model-log pairs where it does not find the optimal (minimal) alignments, the over-approximation is small (one or a handful of moves) and it only occurs for a small percentage of traces (5\% or less).

To benefit from the advantages of both approaches (basic automata-based and decomposition-based approach), the article also presented a hybrid approach where either the automata-based approach (without decomposition) or the S-components-decomposition approach is used depending on the size of the reachabililty graph versus the total size of the S-components. The empirical evauluation showed that this hybrid approach outperforms all other techniques in a clear majority of cases (first in 30 out of 40 logs and second in 5 of the remaining 10 cases).




The proposed technique still fails to perform satisfactorily on a handful of the event logs used in the evaluation. Further improvements may be achieved by designing tighter heuristic functions to guide the A* algorithm, particularly to handle process models with nested loops.


In this article, we combined the S-components decomposition with an automata-based conformance checking approach. This combination is natural since each S-component corresponds to a concurrency-free slice of the process model, which can be seen as an automaton. However, the idea of using an S-component decomposition for conformance checking has broader applicability. We foresee that the S-components decomposition approach could also be used in conjunction with other trace alignment approaches. In particular, adapting the S-components approach to work with approaches that align one trace at a time, such as those of Adriansyah et al.~\cite{AdriansyahDA11} or Van Dongen~\cite{BVD-Alignment}, is a possible avenue for future work.


This article addressed the problem of identifying unfitting log behavior. However, the ideas investigated in this paper could also be applied to the related problem of identifying additional model behavior, that is, behavior captured in the process model but not observed in the event log~\cite{GarciaL17}. This latter problem is related to that of measuring the precision of a process model relative to an event log, which is an open problem in the field of process mining~\cite{PrecisionMeasures}. Another direction for future work is to investigate the application of the S-components decomposition approach to the problem of identifying and measuring additional model behavior.




\medskip
\noindent \emph{Acknowledgments.} This research is partly funded by the Australian Research Council (grant DP180102839), the Estonian Research Council (grant IUT20-55), and the European Research Council (PIX project).

%% file: Appendix.tex
\newpage
\section{Complete cost comparison and order of approximation}\label{app:cost_comparison}

\input{tableCostFull.tex}
\pagebreak
\section{Time performance -- single threaded}\label{app:time_performance_st}

\input{tableTimeST.tex}

\section{Recomposing partial alignments is correct}\label{app:alg5_correct}

\input{Extension_recomposing_proof} 

%% file: tableCostFull.tex
  \begin{table}[htbp!]                            
 {\footnotesize{                            
 \setlength{\tabcolsep}{3pt}                            
 \centering{                            
 \begin{tabular}{|c|c|c|c c c c c|c c|c c|c c|}                            
 \hline                            
 & & &\multicolumn{5}{|c|}{\textbf{ Cost }} &\multicolumn{2}{|c|}{\textbf{ Over-Apprx. }} & \multicolumn{2}{|c|}{\textbf{ \%Apprx traces }} & \multicolumn{2}{|c|}{\textbf{ Avg Over Apprx. }} \\ \cline{4-14}                           
 \textbf{ Miner }&\textbf{ Domain }&\textbf{ Dataset }& \textbf{ ILP }& \textbf{ eMEQ }& \multirow{2}{*}{\textbf{ ALI }}& \textbf{ Aut. }& \textbf{ S-Comp}& \multirow{2}{*}{\textbf{ ALI }}& \textbf{ S-Comp}& \multirow{2}{*}{\textbf{ ALI }}& \textbf{ S-Comp}& \multirow{2}{*}{\textbf{ ALI }}& \textbf{ S-Comp}\\                           
 & & & \textbf{ Align.}& \textbf{ Align. } & & \textbf{ Appr. }& \textbf{ Appr. } & & \textbf{ Appr. } & & \textbf{ Appr. } & & \textbf{ Appr. }\\ \hline                           
\multirow{20}{*}{ IM }&\multirow{12}{*}{ public }& BPIC12 & 0.87 & 0.87 & 3.38 & 0.87 & 0.87 & 3.05 & 0.00 & 100.0\% &  & 3.05 &  \\ \cline{3-14}
  &  & BPIC13\textsubscript{cp} & 1.46 & 1.46 & 2.36 & 1.46 & 1.46 & 0.89 & 0.00 & 81.0\% &  & 1.10 &  \\ \cline{3-14}
  &  & BPIC13\textsubscript{inc} & 0.84 & 0.84 & 4.14 & 0.84 & 0.84 & 3.31 & 0.00 & 98.6\% &  & 3.35 &  \\ \cline{3-14}
  &  & BPIC14\textsubscript{f} & 1.94 & 1.94 & 3.58 & t/out & 1.94 & 3.24 & 0.00 & 100.0\% &  & 3.24 &  \\ \cline{3-14}
  &  & BPIC15\textsubscript{1f} & 0.50 & 0.50 & 15.01 & t/out & t/out & 14.51 & t/out & 100.0\% &  & 14.51 &  \\ \cline{3-14}
  &  & BPIC15\textsubscript{2f} & 2.02 & 2.02 & 26.73 & 2.02 & 2.07 & 24.71 & 0.05 & 100.0\% & 5.2\% & 24.71 & 1 \\ \cline{3-14}
  &  & BPIC15\textsubscript{3f} & 1.70 & 1.70 & 25.67 & 1.70 & 1.70 & 23.97 & 0.00 & 100.0\% &  & 23.97 &  \\ \cline{3-14}
  &  & BPIC15\textsubscript{4f} & 1.14 & 1.14 & 22.84 & 1.14 & 1.14 & 21.69 & 0.00 & 100.0\% &  & 21.69 &  \\ \cline{3-14}
  &  & BPIC15\textsubscript{5f} & 1.20 & 1.20 & 24.90 & 1.20 & 1.20 & 23.70 & 0.00 & 100.0\% &  & 23.70 &  \\ \cline{3-14}
  &  & BPIC17\textsubscript{f} & 0.83 & 0.83 & 11.53 & 0.83 & 0.83 & 10.93 & 0.00 & 100.0\% &  & 10.93 &  \\ \cline{3-14}
  &  & RTFMP & 0.06 & 0.06 & 1.77 & 0.06 & 0.06 & 1.71 & 0.00 & 100.0\% &  & 1.71 &  \\ \cline{3-14}
  &  & SEPSIS & 0.12 & 0.12 & 12.83 & 0.12 & 0.12 & 12.69 & 0.00 & 100.0\% & 0.2\% & 12.69 & 1 \\ \cline{2-14}
  &\multirow{8}{*}{ private }& PRT1 & 1.43 & 1.43 & 4.57 & 1.43 & 1.43 & 3.15 & 0.00 & 98.7\% &  & 3.19 &  \\ \cline{3-14}
  &  & PRT2 & t/out & t/out & 38.32 & t/out & t/out & 0.00 & 0.00 & 0.0\% &  & 0.00 &  \\ \cline{3-14}
  &  & PRT3 & 0.23 & 0.23 & 3.84 & 0.23 & 0.23 & 3.61 & 0.00 & 100.0\% &  & 3.61 &  \\ \cline{3-14}
  &  & PRT4 & 1.22 & 1.22 & 4.79 & 1.22 & 1.22 & 3.56 & 0.00 & 98.5\% &  & 3.61 &  \\ \cline{3-14}
  &  & PRT6 & 0.09 & 0.09 & 2.31 & 0.09 & 0.12 & 2.22 & 0.03 & 69.9\% & 2.7\% & 3.17 & 1.05 \\ \cline{3-14}
  &  & PRT7 & 0.00 & 0.00 & 2.10 & 0.00 & 0.00 & 2.10 & 0.00 & 99.5\% &  & 2.11 &  \\ \cline{3-14}
  &  & PRT9 & 0.38 & t/out & 2.14 & 0.38 & 0.39 & 1.79 & 0.01 & 98.9\% & 0.9\% & 1.80 & 1.50 \\ \cline{3-14}
  &  & PRT10 & 0.06 & 0.06 & 2.81 & 0.06 & 0.06 & 2.75 & 0.00 & 99.9\% &  & 2.76 &  \\ \hline\hline
\multirow{20}{*}{ SM }&\multirow{12}{*}{ public }& BPIC12 & 1.29 & 1.29 & 17.87 & 1.29 & 1.30 & 16.58 & 0.02 & 73.6\% & 0.8\% & 22.55 & 2 \\ \cline{3-14}
  &  & BPIC13\textsubscript{cp} & 0.09 & 0.09 & 2.19 & 0.09 & 0.09 & 2.10 & 0.00 & 97.8\% &  & 2.15 &  \\ \cline{3-14}
  &  & BPIC13\textsubscript{inc} & 0.24 & 0.24 & 3.77 & 0.24 & 0.24 & 3.54 & 0.00 & 99.7\% &  & 3.55 &  \\ \cline{3-14}
  &  & BPIC14\textsubscript{f} & 2.91 & t/out & 5.83 & 2.91 & 2.91 & 2.92 & 0.00 & 67.3\% &  & 4.34 &  \\ \cline{2-14}
  &  & BPIC15\textsubscript{1f} & 3.20 & 3.20 & 19.38 & 3.20 & 3.20 & 16.18 & 0.00 & 93.2\% &  & 17.36 &  \\ \cline{3-14}
  &  & BPIC15\textsubscript{2f} & 10.22 & 10.22 & 27.44 & 10.22 & 10.22 & 17.21 & 0.00 & 80.8\% &  & 21.31 &  \\ \cline{3-14}
  &  & BPIC15\textsubscript{3f} & 9.70 & 9.70 & 13.75 & 9.70 & 9.70 & 4.06 & 0.00 & 99.3\% &  & 4.08 &  \\ \cline{3-14}
  &  & BPIC15\textsubscript{4f} & 10.42 & 10.42 & 22.23 & 10.42 & 10.42 & 11.81 & 0.00 & 76.9\% &  & 15.37 &  \\ \cline{3-14}
  &  & BPIC15\textsubscript{5f} & 8.10 & 8.10 & 17.04 & 8.10 & 8.10 & 8.94 & 0.00 & 71.5\% &  & 12.51 &  \\ \cline{3-14}
  &  & BPIC17\textsubscript{f} & 1.47 & 1.47 & 18.94 & 1.47 & 1.47 & 17.47 & 0.00 & 100.0\% &  & 17.47 &  \\ \cline{3-14}
  &  & RTFMP & 0.03 & 0.03 & 2.36 & 0.03 & 0.03 & 2.32 & 0.00 & 86.3\% &  & 2.69 &  \\ \cline{3-14}
  &  & SEPSIS & 4.72 & t/out & 12.23 & 4.72 & 4.72 & 7.51 & 0.00 & 93.5\% &  & 8.03 &  \\ \cline{2-14}
  &\multirow{8}{*}{ private }& PRT1 & 0.29 & 0.29 & 3.95 & 0.29 & 0.29 & 3.66 & 0.00 & 99.4\% &  & 3.68 &  \\ \cline{3-14}
  &  & PRT2 & 8.32 & t/out & 34.53 & 8.32 & 8.32 & 26.22 & 0.00 & 100.0\% &  & 26.22 &  \\ \cline{3-14}
  &  & PRT3 & 2.54 & 2.54 & 6.50 & 2.54 & 2.54 & 3.96 & 0.00 & 99.4\% &  & 3.98 &  \\ \cline{3-14}
  &  & PRT4 & 1.91 & 1.91 & 3.59 & 1.91 & 1.91 & 1.68 & 0.00 & 64.3\% &  & 2.61 &  \\ \cline{3-14}
  &  & PRT6 & 1.08 & 1.08 & 1.26 & 1.08 & 1.08 & 0.19 & 0.00 & 9.3\% &  & 2.00 &  \\ \cline{3-14}
  &  & PRT7 & 1.40 & 1.40 & 2.22 & 1.40 & 1.41 & 0.83 & 0.01 & 37.0\% & 1.2\% & 2.23 & 1 \\ \cline{3-14}
  &  & PRT9 & 0.35 & t/out & 1.43 & 0.35 & 0.35 & 1.08 & 0.00 & 28.7\% &  & 3.77 &  \\ \cline{3-14}
  &  & PRT10 & 0.10 & 0.10 & 2.76 & 0.10 & 0.10 & 2.65 & 0.00 & 97.2\% &  & 2.73 &  \\ \hline
 \end{tabular}                            
 }                            
 \vspace{.5\baselineskip}                            
 \caption{Cost comparison and order of approximation for those log-model pairs where S-Components over-approximate}\label{tb:cost_approximation}                            
 \vspace{.5\baselineskip}                            
 }}                            
 \end{table}                                                      

%% file: tableTimeST.tex
 \begin{table}[htbp!]                 
 {\footnotesize{                 
 \setlength{\tabcolsep}{3pt}                 
 \centering{                 
 \begin{tabular}{|c|c|c|c c c|c c|c|}                 
 \hline                 
  &  &  & \multicolumn{3}{c|}{\textbf{ Baselines }}     & \multicolumn{3}{c|}{\textbf{ Our Approaches }}   \\ \cline{4-9}
\textbf{ Miner }& \textbf{ Domain }& \textbf{ Dataset }& \textbf{ ILP }& \textbf{ extended MEQ }& \textbf{ ALI }& \textbf{ Automata-based }& \textbf{ S-Component }& \textbf{ Hybrid }\\
  &  &  & \textbf{ Alignments }& \textbf{ Alignments } &  & \textbf{ Approach }& \textbf{ Approach }& \textbf{ Approach }\\ \hline
\multirow{20}{*}{ IM } & \multirow{12}{*}{ public }& BPIC12 &\textbf{ 125,948 }& t/out & 539,214 & 428,166 & 160,070 & 160,070 \\ \cline{3-9}
  &  & BPIC13\textsubscript{cp} &\textbf{ 295 }& 2,137 & 4,469 & 330 & 317 & 330 \\ \cline{3-9}
  &  & BPIC13\textsubscript{inc} & 2,040 & 79,597 & 29,657 & 24,917 &\textbf{ 504 }&\textbf{ 504 }\\ \cline{3-9}
  &  & BPIC14\textsubscript{f} & 98,667 & 592,238 & 130,455 & t/out &\textbf{ 9,844 }&\textbf{ 9,844 }\\ \cline{3-9}
  &  & BPIC15\textsubscript{1f} &\textbf{ 3,691 }& 6,015 & 20,203 & t/out & t/out & t/out \\ \cline{3-9}
  &  & BPIC15\textsubscript{2f} & 19,805 & 25,185 & 44,657 &\textbf{ 8,061 }& 16,418 & 16,418 \\ \cline{3-9}
  &  & BPIC15\textsubscript{3f} & 31,075 & 91,996 & 49,142 &\textbf{ 12,218 }& 69,432 &\textbf{ 12,218 }\\ \cline{3-9}
  &  & BPIC15\textsubscript{4f} & 12,093 & 38,816 & 35,125 &\textbf{ 3,837 }& 8,064 & 8,064 \\ \cline{3-9}
  &  & BPIC15\textsubscript{5f} &\textbf{ 8,174 }& 15,553 & 24,688 & 32,413 & 18,582 & 32,413 \\ \cline{3-9}
  &  & BPIC17\textsubscript{f} & 22,805 & 257,705 & 283,112 &\textbf{ 7,127 }& 8,996 &\textbf{ 7,127 }\\ \cline{3-9}
  &  & RTFMP & 3,924 & 147,135 & 6,891 & 1,779 &\textbf{ 1,024 }&\textbf{ 1,024 }\\ \cline{3-9}
  &  & SEPSIS & 6,964 & 17,785 & 29,281 & 58,571 &\textbf{ 2,984 }&\textbf{ 2,984 }\\ \cline{2-9}
  &\multirow{8}{*}{ private }& PRT1 & 1,379 & 12,454 & 10,891 & 2,527 &\textbf{ 693 }&\textbf{ 693 }\\ \cline{3-9}
  &  & PRT2 & t/out & t/out & 17,203 & t/out & t/out & t/out \\ \cline{3-9}
  &  & PRT3 & 425 & 2,175 & 6,438 &\textbf{ 392 }& 404 &\textbf{ 392 }\\ \cline{3-9}
  &  & PRT4 & 3,945 & 22,515 & 32,359 & 6,069 &\textbf{ 1,954 }& 6,069 \\ \cline{3-9}
  &  & PRT6 & 335 & 1,168 & 6,125 &\textbf{ 315 }& 330 &\textbf{ 315 }\\ \cline{3-9}
  &  & PRT7 &\textbf{ 300 }& 2,560 & 5,844 & 301 & 323 & 301 \\ \cline{3-9}
  &  & PRT9 & 29,538 & t/out & 10,843 & 68,047 &\textbf{ 1,937 }&\textbf{ 1,937 }\\ \cline{3-9}
  &  & PRT10 & 970 & 49,160 & 8,750 &\textbf{ 350 }& 724 & 724 \\ \hline \hline
\multirow{20}{*}{ SM } & \multirow{12}{*}{ public }& BPIC12 & 189,423 & t/out & 118,053 &\textbf{ 32,035 }& 111,256 &\textbf{ 32,035 }\\ \cline{3-9}
  &  & BPIC13\textsubscript{cp} & 350 & 3,481 & 4,351 & 258 &\textbf{ 232 }& 258 \\ \cline{3-9}
  &  & BPIC13\textsubscript{inc} & 1,764 & 96,106 & 18,906 &\textbf{ 766 }& 918 &\textbf{ 766 }\\ \cline{3-9}
  &  & BPIC14\textsubscript{f} & 40,161 & t/out & 157,733 &\textbf{ 26,605 }& 31,743 &\textbf{ 26,605 }\\ \cline{3-9}
  &  & BPIC15\textsubscript{1f} & 3,239 & 2,941 & 12,986 & 2,266 &\textbf{ 2,048 }& 2,266 \\ \cline{3-9}
  &  & BPIC15\textsubscript{2f} & 9,534 & 7,025 & 26,114 &\textbf{ 3,568 }& 3,678 &\textbf{ 3,568 }\\ \cline{3-9}
  &  & BPIC15\textsubscript{3f} & 7,541 & 12,926 & 25,067 &\textbf{ 2,531 }& 2,660 &\textbf{ 2,531 }\\ \cline{3-9}
  &  & BPIC15\textsubscript{4f} & 7,516 & 8,220 & 19,893 &\textbf{ 2,345 }& 2,388 &\textbf{ 2,345 }\\ \cline{3-9}
  &  & BPIC15\textsubscript{5f} & 9,111 & 6,303 & 24,608 &\textbf{ 2,710 }& 2,815 &\textbf{ 2,710 }\\ \cline{3-9}
  &  & BPIC17\textsubscript{f} & 28,474 & 91,585 & 265,169 &\textbf{ 6,357 }& 8,206 &\textbf{ 6,357 }\\ \cline{3-9}
  &  & RTFMP & 3,477 & 146,871 & 7,020 &\textbf{ 354 }& 593 &\textbf{ 354 }\\ \cline{3-9}
  &  & SEPSIS & 5,355 & t/out & 18,821 &\textbf{ 1,179 }& 1,219 &\textbf{ 1,179 }\\ \cline{2-9}
  & private & PRT1 & 2,391 & 66,135 & 10,750 &\textbf{ 664 }& 809 &\textbf{ 664 }\\ \cline{3-9}
  &  & PRT2 & 36,539 & t/out & 95,173 &\textbf{ 18,218 }& 18,531 &\textbf{ 18,218 }\\ \cline{3-9}
  &  & PRT3 & 365 & 2,627 & 6,422 &\textbf{ 332 }& 437 &\textbf{ 332 }\\ \cline{3-9}
  &  & PRT4 & 7,319 & 160,196 & 48,438 &\textbf{ 3,151 }& 6,831 &\textbf{ 3,151 }\\ \cline{3-9}
  &  & PRT6 &\textbf{ 230 }& 1,005 & 6,015 & 249 & 300 & 249 \\ \cline{3-9}
  &  & PRT7 &\textbf{ 256 }& 2,256 & 5,906 & 258 & 422 & 258 \\ \cline{3-9}
  &  & PRT9 & 27,542 & t/out & 18,453 &\textbf{ 1,472 }& 2,183 &\textbf{ 1,472 }\\ \cline{3-9}
  &  & PRT10 & 1,263 & 96,503 & 9,672 &\textbf{ 328 }& 362 &\textbf{ 328 }\\ \hline \hline
\multicolumn{3}{|c|}{ Total time spent (ms):  }   & 754,223 & 2,068,372 & 2,194,897 & 761,068 & 500,228 &\textbf{ 367,072 }\\ \hline
\multicolumn{3}{|c|}{ Total Outperforming:  }   & 7 & 0 & 1 & 23 & 9 &\textbf{ 26 }\\ \hline
\multicolumn{3}{|c|}{ Total second:  }   & 7 & 2 & 1 & 6 &\textbf{ 23 }& 9 \\ \hline
\multicolumn{3}{|c|}{ \#Timeouts:  }   & 1 & 8 &\textbf{ 0 }& 3 & 2 & 2 \\ \hline
 \end{tabular}                 
 }                 
 \vspace{.5\baselineskip}                 
 \caption{Time performance in milliseconds}\label{tb:time_performance}                 
 \vspace{.5\baselineskip}                 
 }}                 
 \end{table}                 

%% file: Extension_recomposing_proof.tex

Theorem~\ref{thm:recompose_alignment_optimal} states that the sequence $\epsilon_c$ returned by Alg.~\ref{alg:constPSPwithSComponents}$^*(\logL,\dafsa,\sysNet)$, the modification of Alg.~\ref{alg:constPSPwithSComponents} constructing reachability graphs as described in Sect.~\ref{sec:s-components_invisible_label_conflicts}, is an alignment of $\dafsa$ to a sound, uniquely-labeled, free-choice workflow net $\sysNet$. In other words, the projection $\epsilon_c$ onto the left-hand component is trace $c$, and the projection on the right-hand component is a path through the reachability graph of $\wNet$. We prove both properties individually.


\begin{proof}[Proof of Thm~\ref{thm:recompose_alignment_optimal}.1]
We show $\lambda(\epsilon_c\trPr{\dafsa}) = c$ by induction on the prefixes $c'$ of $c$ in the for-loop in lines 10-39. For the empty prefix $c'$ before the for-loop, $\epsilon_c = \langle \rangle$. In each iteration of the for-loop with $\mathit{pos}_c \leq |c|$, the prefix $c'$ is extended with $\ell = c(\mathit{pos}_c) \in L$, and the current prefix of $\epsilon_c$ is extended with a synchronization $(\lhide,(n,\ell,n'),\perp)$ (line 27) or $(\match,(n,\ell,n'),(m,\ell,m'))$ (line 33). Thus, the proposition holds for both prefixes. The only other extension of the current prefix of $\epsilon_c$ in Alg.5 is with synchronizations $(\lhide,\perp,(m,\ell,m'))$ in line 18 which do not occur in $\lambda(\epsilon_c\trPr{\dafsa})$.

\end{proof}

Proving Thm~\ref{thm:recompose_alignment_optimal}.2 requires some further notation, definitions, and observations on Petri nets.

For a $\wNet$, let $\scomp = \{\wNet_1,\wNet_2, \dots, \wNet_k\}$ be the set of S-Components of $\wNet$. By the abuse of notation, let $\scomp(t)$ be the set of S-components in which $t$ is contained as $\wNet_i = ((P_i,T_i,F_i,\lambda_i),i_i,o_i) \in \scomp(t)$ iff $t \in T_i$,
for each $t\in T$; sets $\scomp(p), p\in P_i$, are defined accordingly.

In a  sound free-choice net $\wNet$, the pre- and post-sets of a transition $t$ (together) cover the same S-component, which follows from $\wNet$ being covered by S-components~\cite{Esparza1990synthesis_rules} and the free-choice structure:
\begin{equation}\label{eq:label_unique_pre_post_cover}
\bigcup_{p \in \inP{t}} \scomp(p) = \scomp(t) = \bigcup_{p \in \outP{t}} \scomp(p).
\end{equation}
In any free-choice net $\wNet$ with S-components $\{\wNet_1,\wNet_2, \ldots,\wNet_k\}$ and reachability graphs $\reachGraph(\wNet_j) = (M^j,A^j,m_0^j,M_f^j), j=1,\ldots,k$ holds:
\begin{equation}\label{eq:s_component_single_place_marked}
\textit{each reachable marking } m \in M^j \textit{ has the form } m = [p], p\in P_j
\end{equation}
\begin{equation}\label{eq:s_component_arc_single_places}
\textit{each arc } a \in A^j \textit{ has the form } a = ([p],\lambda(t),[p']), p,p'\in P_j, t \in T_j
\end{equation}
Without loss of generality, in a sound, free-choice workflow net $\wNet = (P,T,F,\lambda,m_0,m_f)$ holds $m_0 = [p_0], p_0 \in P, \inTr{p_0} = \emptyset$ and $m_f = \{[p_f]\}, p_f \in P, \outTr{p_f} = \emptyset$.
%


\begin{proof}[Proof of Thm~\ref{thm:recompose_alignment_optimal}.2]
We have to show that $\epsilon_c\trPr{\reachGraph}$ corresponds to a path through the $\reachGraph$ of $\wNet$.
Let $\epsilon_c\trPr{\reachGraph} = (a_1,\ldots,a_s)$.

By equation (\ref{eq:s_component_single_place_marked}), each $m_i = ([p_i^1],\ldots,[p_i^k])$ and for the $k$ S-components of $\wNet$. For such a vector $m = ([p^1],\ldots,[p^k])$ let, $\widehat{m} = \{p^1,\ldots,p^k\}$ be the set of marked places.

We show the proposition by showing that
\begin{inparaenum}[(a)]
	\item\label{item:a} $\widehat{m_0} = m_0^\wNet$,
	\item\label{item:b} each $\widehat{m_i} \in M_{\reachGraph(\wNet)},i=1,\ldots,s'$ is a marking of $\wNet$,
	\item\label{item:c} each $(\widehat{m_i},\ell_i,\widehat{m'_i}) \in A_{\reachGraph(\wNet)}$ is a step in $\wNet$, and
	\item\label{item:d} $\widehat{m'_{s'}} \in M_f^\wNet$.
\end{inparaenum}

Regarding (\ref{item:a}), the initial marking of each S-component $j=1,\ldots,k$ is $m_0^j = [p_0]$ as $m_0^\wNet = [p_0]$. Thus, the proposition holds by $m_0 = ([p_0],\ldots,[p_0])$ in line 9.

We show (\ref{item:b}) and (\ref{item:c}) by induction on the length $i$ of the prefixes of $\epsilon_c$. For $i=0$, (\ref{item:b}) holds for $m_0$ due to (\ref{item:a}), and (\ref{item:c}) holds trivially. For $i > 0$, if $\epsilon_c\trPr{\reachGraph}[i] = \perp$ then there is nothing to show. Otherwise, $\epsilon_c\trPr{\reachGraph}[i] = (\mathit{op},a_D,(m_{i-1},\ell_i,m_i))$ and $m_{i-1} \in A_{\reachGraph(\wNet)}$ by inductive assumption. We have to show: for $([p_{i-1}^1],\ldots,[p_{i-1}^k]) = m_{i-1}$ and $([p_i^1],\ldots,[p_i^k]) = m_i$, $(\widehat{m_{i-1}},\ell_i,\widehat{m_i}) \in A_{\reachGraph(\wNet)}$.

If $\ell_i \neq \ell = c(\mathit{pos}_c)$ in line 11, then $\configs_{\ell} \neq \emptyset$, and then $(m_{i-1},\ell_i,m_i)$ due to line 18, $m_i = m_{i-1}  \blacktriangleright (a_1,\dots,a_k)$, and $\mathit{lab}_x = \mathit{sync}_x = \scomp(t)$ for some transition $t \in T, \netLabel(t) = \ell_i = x$ (by lines 15,16). If there was no such $t$, then $\epsilon_c$ is due to line 40 and the proposition holds by Lem.~\ref{lem:global_alignment_is_alignment}.

Because $\wNet$ is uniquely labeled, $t$ is unique. By line 14, $([p_{i-1}^j],\ell_i,[p_i^j]) \in A_{\reachGraph}^j$ for each $j \in \scomp(t)$.
Due to (\ref{eq:label_unique_pre_post_cover}), the preset of $t$ is marked, $\inTr{t} = \{ p_{i-1}^j \mid j\in \scomp(t)\} \subseteq \widehat{m_{i-1}}$, and $t$ is enabled in $\widehat{m_{i-1}}$. Firing $t$ yields the successor marking $m^* = (\widehat{m_{i-1}} \setminus \inTr{t}) \cup \outTr{t} = (\widehat{m_{i-1}} \setminus \{ p_{i-1}^j \mid j\in \scomp(t) \}) \cup \{ p_i^j \mid j \in \scomp(t) \}$ by (\ref{eq:label_unique_pre_post_cover}). By construction of $(a_1,\dots,a_k)$ in line 17 from $\mathit{sync}_x = \scomp(t)$, we can rewrite $m^* = (\widehat{m_{i-1}} \setminus \{ p_{i-1}^j \mid a_j \neq \perp \}) \cup \{ p_i^j \mid a_j \neq \perp \}$ as Alg.~\ref{alg:remTau}$^*$ ensures transition effects are uniquely identified by their extended labels (see Sect.~\ref{sec:s-components_invisible_label_conflicts}). By line 18, and the definition of $\blacktriangleright$, $m^* = \widehat{m_i}$. Thus, $(\widehat{m_{i-1}} \setminus \inTr{t}) \cup \outTr{t} = \widehat{m_{i}}$ and propositions \ref{item:b} and \ref{item:c} hold.

If $\ell_i = \ell = c(\mathit{pos}_c)$ then $(m_{i-1},\ell_i,m_i)$ due to line 33 and a similar reasoning as above holds as there exists a unique transition $t$ with $\netLabel(t) = \ell_i$ and $([p_{i-1}^j],\ell_i,[p_i^j]) \in A_{\reachGraph}^j$ with $p_{i-1}^j \in \widehat{m_{i-1}}$ for each $j \in \scomp(t)$.

To prove (\ref{item:d}), we know $m_f^\wNet = \{ [p_f] \}$, by $N$ having a unique final place $p_f$. Thus, for $m'_{s'} = ([p^1_{s'}],\ldots,[p^k_{s'}])$, $p^j_{s'} = p_f$ has to hold for all $j=1,\ldots,k$. Suppose that for the recomposed alignment, there exists $j \in \{1,\ldots,k\}$ where in $m'_{s'}$, $[p_{s'}^j] \neq [p_f]$. Each $\epsilon_j$ calculated in line 5 of Alg.~\ref{alg:constPSPwithSComponents}$^*$ is an alignment. Thus, the path $\modPr{\epsilon_j}\trPr{L_j} = (a_1^j,\ldots,a_{s_j}^j)$ through $\reachGraph(\wNet_j)$ ends in the final place $a_{s_j}^j = (m^j,\ell_j,[p_f])$, and thus $\reachGraph(\wNet_j)$ has further arcs from $[p_{s'}^j]$ to $[p_f]$ that should have been considered by Alg.~\ref{alg:constPSPwithSComponents}$^*$. Case distinction:
\begin{itemize}
\item[(i)] For all $j = 1, \ldots, k$, $[p_{s'}^j] \neq [p_f]$ with two possible cases:

    (i-a) There exists some $t \in T$ of $\wNet$ with $\widehat{m'_{s'}} \supseteq \inTr{t}$ and $\netLabel(t) = x$. Then either $\mathit{sync}_x = \mathit{lab}_x$ in line 16 of Alg.~\ref{alg:constPSPwithSComponents}$^*$ and a synchronization with arc $(m'_{s'}, x, m'')$ would have been added to $\epsilon_c$ by the arguments for (c) given above. Or $x = \ell = c(\mathit{pos}_c), \mathit{pos}_c < |c|$ and a corresponding synchronziation would have been added in line 33 of Alg.5. Both cases contradict the algorithm.

    (i-b) There exists no $t \in T$ with $\widehat{m'_{s'}} \supseteq \inTr{t}$. But then $m'_{s'}$ is a deadlock contradicting $\wNet$ being sound.

\item[(ii)] There exist $j,r$, $[p_{s'}^j] \neq [p_f]$ and $[p_{s'}^r] = [p_f]$. By (b) and (c), $\widehat{m'_{s'}} \subseteq \{ p_f, p_{s'}^j \}$ is a reachable marking of $\wNet$ which contradicts $\wNet$ being sound.
\end{itemize}
\end{proof}

%

%% file: journal_paper.bbl
\begin{thebibliography}{10}
\expandafter\ifx\csname url\endcsname\relax
  \def\url#1{\texttt{#1}}\fi
\expandafter\ifx\csname urlprefix\endcsname\relax\def\urlprefix{URL }\fi
\expandafter\ifx\csname href\endcsname\relax
  \def\href#1#2{#2} \def\path#1{#1}\fi

\bibitem{ProcessMiningBook}
W.~{van der Aalst}, Process Mining - Data Science in Action, Second Edition,
  Springer, 2016.

\bibitem{ConformanceCheckingBook}
J.~Carmona, B.~{van Dongen}, A.~Solti, M.~Weidlich, Conformance Checking -
  Relating Processes and Models, Springer, 2018.

\bibitem{AdriansyahDA11}
A.~Adriansyah, B.~van Dongen, W.~{van der Aalst}, Conformance checking using
  cost-based fitness analysis, in: Proc. of EDOC, IEEE, 2011, pp. 55--64.

\bibitem{BVD-Alignment}
B.~{van Dongen}, Efficiently computing alignments: using the extended marking
  equation, in: M.~Montali, I.~Weber, M.~Weske, J.~{vom Brocke} (Eds.),
  Business Process Management - 16th International Conference, BPM 2018,
  Proceedings, Lecture Notes in Computer Science (including subseries Lecture
  Notes in Artificial Intelligence and Lecture Notes in Bioinformatics),
  Springer, Germany, 2018, pp. 197--214.

\bibitem{AugustoCDRMMMS19}
A.~Augusto, R.~Conforti, M.~Dumas, M.~{La Rosa}, F.~Maggi, A.~Marrella,
  M.~Mecella, A.~Soo, Automated discovery of process models from event logs:
  Review and benchmark, {IEEE} Trans. Knowl. Data Eng. 31~(4) (2019) 686--705.

\bibitem{ReissnerCDRA17}
D.~Rei{\ss}ner, R.~Conforti, M.~Dumas, M.~{La Rosa}, A.~Armas-Cervantes,
  Scalable conformance checking of business processes, in: H.~Panetto,
  C.~Debruyne, W.~Gaaloul, M.~Papazoglou, A.~Paschke, C.~Ardagna, R.~Meersman
  (Eds.), On the Move to Meaningful Internet Systems. OTM 2017 Conferences,
  Springer International Publishing, Cham, 2017, pp. 607--627.

\bibitem{RozinatA08}
A.~Rozinat, W.~{van der Aalst}, Conformance checking of processes based on
  monitoring real behavior, Inf. Syst. 33~(1) (2008) 64--95.

\bibitem{Medeiros06}
A.~Alves~de Medeiros, Genetic process mining, Ph.D. thesis, TU/e (2006).

\bibitem{BrouckeMGCBV14}
S.~vanden Broucke, J.~Mu{\~n}oz-Gama, J.~Carmona, B.~Baesens, J.~Vanthienen,
  Event-based real-time decomposed conformance analysis, in: Proc. of OTM,
  Springer, 2014, pp. 345--363.

\bibitem{Munoz-GamaCA14}
J.~Mu{\~n}oz-Gama, J.~Carmona, W.~{van der Aalst}, Single-entry single-exit
  decomposed conformance checking, Inf. Syst. 46 (2014) 102--122.

\bibitem{Adriansyah14}
A.~Adriansyah, Aligning observed and modeled behavior, Ph.D. thesis, TU/e
  (2014).

\bibitem{leoni2017}
M.~{de Leoni}, A.~Marrella, Aligning real process executions and prescriptive
  process models through automated planning, Expert Systems with Applications
  82 (2017) 162--183.

\bibitem{GarciaL17}
L.~Garc{\'\i}a-Ba{\~n}uelos, N.~van Beest, M.~Dumas, M.~{La Rosa}, W.~Mertens,
  Complete and interpretable conformance checking of business processes, IEEE
  Transactions on Software Engineering 44~(3) (2018) 262--290.

\bibitem{NielsenPW1981}
M.~Nielsen, G.~Plotkin, G.~Winskel, Petri nets, event structures and domains,
  part {I}, Theoretical Computer Science 13~(1).

\bibitem{Boudewijn17}
B.~{van Dongen}, J.~Carmona, T.~Chatain, F.~Taymouri, Aligning modeled and
  observed behavior: a compromise between complexity and quality, in: Proc. of
  CAiSE, Springer, 2017.

\bibitem{ALI}
F.~Taymouri, Light methods for conformance checking of business processes
  (2018).

\bibitem{EvolutionaryAllOptimal}
F.~Taymouri, J.~Carmona, An evolutionary technique to approximate multiple
  optimal alignments, in: M.~Weske, M.~Montali, I.~Weber, J.~vom Brocke (Eds.),
  Business Process Management, Springer International Publishing, Cham, 2018,
  pp. 215--232.

\bibitem{TraceSampling}
M.~Bauer, H.~{van der Aa}, M.~Weidlich, Estimating process conformance by trace
  sampling and result approximation, in: T.~Hildebrandt, B.~van Dongen,
  M.~R{\"o}glinger, J.~Mendling (Eds.), Business Process Management, Springer
  International Publishing, Cham, 2019, pp. 179--197.

\bibitem{OnlineConformance}
A.~Burattin, S.~{van Zelst}, A.~Armas-Cervantes, B.~{van Dongen}, J.~Carmona,
  Online conformance checking using behavioural patterns, in: M.~Weske,
  M.~Montali, I.~Weber, J.~vom Brocke (Eds.), Business Process Management,
  Springer International Publishing, Cham, 2018, pp. 250--267.

\bibitem{van2013decomposing}
W.~{van der Aalst}, Decomposing petri nets for process mining: A generic
  approach, Distributed and Parallel Databases 31~(4) (2013) 471--507.

\bibitem{wang2017aligning}
L.~Wang, Y.~Du, W.~Liu, Aligning observed and modelled behaviour based on
  workflow decomposition, Enterprise Information Systems 11~(8) (2017)
  1207--1227.

\bibitem{verbeek2016merging}
H.~Verbeek, W.~{van der Aalst}, Merging alignments for decomposed replay, in:
  International Conference on Application and Theory of Petri Nets and
  Concurrency, Springer, 2016, pp. 219--239.

\bibitem{song2017}
W.~Song, X.~Xia, H.~Jacobsen, P.~Zhang, H.~Hu, Efficient alignment between
  event logs and process models, IEEE Transactions on Services Computing 10~(1)
  (2017) 136--149.

\bibitem{Diller90}
A.~Diller, Z: An introduction to formal methods, John Wiley \& Sons, Inc.,
  1990.

\bibitem{WFNets}
W.~{van der Aalst}, The application of petri nets to workflow management,
  Journal of Circuits, Systems and Computers 08~(01) (1998) 21--66.

\bibitem{Verbeek01}
H.~Verbeek, T.~Basten, W.~{van der Aalst}, {Diagnosing Workflow Processes using
  Woflan}, Comput. J. 44~(4) (2001) 246--279.

\bibitem{mayr84}
E.~Mayr, An algorithm for the general petri net reachability problem, SIAM
  journal on computing 13~(3) (1984) 441--460.

\bibitem{lipton1976}
R.~Lipton, The reachability problem requires exponential space, Research Report
  62, Department of Computer Science, Yale University, New Haven, Connecticut.

\bibitem{murata}
T.~Murata, Petri nets: Properties, analysis and applications, Proceedings of
  the IEEE 77~(4) (1989) 541--580.

\bibitem{daciukJ00}
J.~Daciuk, S.~Mihov, B.~Watson, R.~Watson, Incremental construction of minimal
  acyclic finite-state automata, Computational linguistics 26~(1) (2000) 3--16.

\bibitem{ArmasBDG2016}
A.~Armas-Cervantes, P.~Baldan, M.~Dumas, L.~{Garc{\'\i}a-Ba{\~{n}}uelos},
  Diagnosing behavioral differences between business process models: An
  approach based on event structures, Information Systems 56.

\bibitem{hart68}
P.~Hart, N.~Nilsson, B.~Raphael, A formal basis for the heuristic determination
  of minimum cost paths, {IEEE TSSC} 4~(2) (1968) 100--107.

\bibitem{21Conf}
A.~Syring, N.~Tax, W.~{van der Aalst}, Evaluating Conformance Measures in
  Process Mining Using Conformance Propositions, Springer Berlin Heidelberg,
  Berlin, Heidelberg, 2019, pp. 192--221.

\bibitem{FreeChoice}
J.~Desel, J.~Esparza, Free choice Petri nets, Vol.~40, Cambridge university
  press, 2005.

\bibitem{ILP-Alignment}
A.~{Adriansyah}, B.~{van Dongen}, W.~{van der Aalst}, Conformance checking
  using cost-based fitness analysis, in: 2011 IEEE 15th International
  Enterprise Distributed Object Computing Conference, 2011, pp. 55--64.

\bibitem{AlignmentTechnicalImprovements}
B.~{van Dongen}, Efficiently computing alignments, in: F.~Daniel, Q.~Sheng,
  H.~Motahari (Eds.), Business Process Management Workshops, Springer
  International Publishing, Cham, 2019, pp. 44--55.

\bibitem{PD-Discovery-BM}
A.~Augusto, R.~Conforti, M.~Dumas, M.~{La Rosa}, F.~Maggi, A.~Marrella,
  M.~Mecella, S.~A., Automated discovery of process models from event logs:
  Review and benchmark, IEEE Transactions on Knowledge and Data Engineering
  31~(4) (2019) 686--705.

\bibitem{BPIC12}
B.~{van Dongen},
  \href{https://data.4tu.nl/repository/uuid:3926db30-f712-4394-aebc-75976070e91f}{Bpi
  challenge 2012} (2012).
\newblock \href
  {http://dx.doi.org/10.4121/UUID:3926DB30-F712-4394-AEBC-75976070E91F}
  {\path{doi:10.4121/UUID:3926DB30-F712-4394-AEBC-75976070E91F}}.
\newline\urlprefix\url{https://data.4tu.nl/repository/uuid:3926db30-f712-4394-aebc-75976070e91f}

\bibitem{BPIC13cp}
W.~Steeman,
  \href{https://data.4tu.nl/repository/uuid:c2c3b154-ab26-4b31-a0e8-8f2350ddac11}{Bpi
  challenge 2013, closed problems} (2013).
\newblock \href
  {http://dx.doi.org/10.4121/UUID:C2C3B154-AB26-4B31-A0E8-8F2350DDAC11}
  {\path{doi:10.4121/UUID:C2C3B154-AB26-4B31-A0E8-8F2350DDAC11}}.
\newline\urlprefix\url{https://data.4tu.nl/repository/uuid:c2c3b154-ab26-4b31-a0e8-8f2350ddac11}

\bibitem{BPIC13inc}
W.~Steeman,
  \href{https://data.4tu.nl/repository/uuid:500573e6-accc-4b0c-9576-aa5468b10cee}{Bpi
  challenge 2013, incidents} (2013).
\newblock \href
  {http://dx.doi.org/10.4121/UUID:500573E6-ACCC-4B0C-9576-AA5468B10CEE}
  {\path{doi:10.4121/UUID:500573E6-ACCC-4B0C-9576-AA5468B10CEE}}.
\newline\urlprefix\url{https://data.4tu.nl/repository/uuid:500573e6-accc-4b0c-9576-aa5468b10cee}

\bibitem{BPIC14}
B.~{van Dongen},
  \href{https://data.4tu.nl/repository/uuid:c3e5d162-0cfd-4bb0-bd82-af5268819c35}{Bpi
  challenge 2014} (2014).
\newblock \href
  {http://dx.doi.org/10.4121/UUID:C3E5D162-0CFD-4BB0-BD82-AF5268819C35}
  {\path{doi:10.4121/UUID:C3E5D162-0CFD-4BB0-BD82-AF5268819C35}}.
\newline\urlprefix\url{https://data.4tu.nl/repository/uuid:c3e5d162-0cfd-4bb0-bd82-af5268819c35}

\bibitem{BPIC15}
B.~{van Dongen},
  \href{https://data.4tu.nl/repository/uuid:31a308ef-c844-48da-948c-305d167a0ec1}{Bpi
  challenge 2015} (2015).
\newblock \href
  {http://dx.doi.org/10.4121/UUID:31A308EF-C844-48DA-948C-305D167A0EC1}
  {\path{doi:10.4121/UUID:31A308EF-C844-48DA-948C-305D167A0EC1}}.
\newline\urlprefix\url{https://data.4tu.nl/repository/uuid:31a308ef-c844-48da-948c-305d167a0ec1}

\bibitem{BPIC17}
B.~{van Dongen},
  \href{https://data.4tu.nl/repository/uuid:5f3067df-f10b-45da-b98b-86ae4c7a310b}{Bpi
  challenge 2017} (2017).
\newblock \href
  {http://dx.doi.org/10.4121/UUID:5F3067DF-F10B-45DA-B98B-86AE4C7A310B}
  {\path{doi:10.4121/UUID:5F3067DF-F10B-45DA-B98B-86AE4C7A310B}}.
\newline\urlprefix\url{https://data.4tu.nl/repository/uuid:5f3067df-f10b-45da-b98b-86ae4c7a310b}

\bibitem{RTFMP}
M.~{de Leoni}, F.~{Mannhardt},
  \href{https://data.4tu.nl/repository/uuid:270fd440-1057-4fb9-89a9-b699b47990f5}{Road
  traffic fine management process} (2015).
\newblock \href
  {http://dx.doi.org/10.4121/UUID:270FD440-1057-4FB9-89A9-B699B47990F5}
  {\path{doi:10.4121/UUID:270FD440-1057-4FB9-89A9-B699B47990F5}}.
\newline\urlprefix\url{https://data.4tu.nl/repository/uuid:270fd440-1057-4fb9-89a9-b699b47990f5}

\bibitem{SEPSIS}
F.~{Mannhardt},
  \href{https://data.4tu.nl/repository/uuid:915d2bfb-7e84-49ad-a286-dc35f063a460}{Sepsis
  cases - event log} (2016).
\newblock \href
  {http://dx.doi.org/10.4121/UUID:915D2BFB-7E84-49AD-A286-DC35F063A460}
  {\path{doi:10.4121/UUID:915D2BFB-7E84-49AD-A286-DC35F063A460}}.
\newline\urlprefix\url{https://data.4tu.nl/repository/uuid:915d2bfb-7e84-49ad-a286-dc35f063a460}

\bibitem{Noise-filtering}
R.~Conforti, M.~{La Rosa}, A.~{Ter Hofstede}, Filtering out infrequent behavior
  from business process event logs, IEEE Transactions on Knowledge and Data
  Engineering 29~(2) (2016) 300--314.

\bibitem{InductiveMiner}
S.~Leemans, D.~Fahland, W.~{van der Aalst}, Discovering block-structured
  process models from event logs-a constructive approach, in: International
  conference on applications and theory of Petri nets and concurrency,
  Springer, 2013, pp. 311--329.

\bibitem{SplitMiner}
A.~Augusto, R.~Conforti, M.~Dumas, M.~{La Rosa}, Split miner: Discovering
  accurate and simple business process models from event logs, in: 2017 IEEE
  International Conference on Data Mining (ICDM), IEEE, 2017, pp. 1--10.

\bibitem{StructuredHeuristicsMiner}
A.~Augusto, R.~Conforti, M.~Dumas, M.~L. Rosa, G.~Bruno, Automated discovery of
  structured process models from event logs: The discover-and-structure
  approach, Data Knowl. Eng. 117 (2018) 373--392.

\bibitem{Fodina}
S.~{vanden Broucke}, J.~{De Weerdt}, Fodina: A robust and flexible heuristic
  process discovery technique, Decision Support Systems 100 (2017) 109 -- 118,
  smart Business Process Management.

\bibitem{reissner_2019}
D.~Reissner,
  \href{https://melbourne.figshare.com/articles/Public_benchmark_data-set_for_Conformance_Checking_in_Process_Mining/8081426}{Public
  benchmark data-set for conformance checking in process mining} (May 2019).
\newblock \href {http://dx.doi.org/10.26188/5cd91d0d3adaa}
  {\path{doi:10.26188/5cd91d0d3adaa}}.
\newline\urlprefix\url{https://melbourne.figshare.com/articles/Public_benchmark_data-set_for_Conformance_Checking_in_Process_Mining/8081426}

\bibitem{SyntheticDataset}
J.~{Mu{\~n}oz-Gama},
  \href{https://data.4tu.nl/repository/uuid:44c32783-15d0-4dbd-af8a-78b97be3de49}{Conformance
  checking in the large (dataset)} (2013).
\newblock \href
  {http://dx.doi.org/10.4121/UUID:44C32783-15D0-4DBD-AF8A-78B97BE3DE49}
  {\path{doi:10.4121/UUID:44C32783-15D0-4DBD-AF8A-78B97BE3DE49}}.
\newline\urlprefix\url{https://data.4tu.nl/repository/uuid:44c32783-15d0-4dbd-af8a-78b97be3de49}

\bibitem{PrecisionMeasures}
N.~Tax, X.~Lu, N.~Sidorova, D.~Fahland, W.~M.~P. van~der Aalst,
  \href{https://doi.org/10.1016/j.ipl.2018.01.013}{The imprecisions of
  precision measures in process mining}, Inf. Process. Lett. 135 (2018) 1--8.
\newblock \href {http://dx.doi.org/10.1016/j.ipl.2018.01.013}
  {\path{doi:10.1016/j.ipl.2018.01.013}}.
\newline\urlprefix\url{https://doi.org/10.1016/j.ipl.2018.01.013}

\bibitem{Esparza1990synthesis_rules}
J.~Esparza, Synthesis rules for petri nets, and how they lead to new results,
  in: {CONCUR} '90, Vol. 458 of Lecture notes in Computer Science, Springer,
  1990, pp. 182--198.

\end{thebibliography}
